\documentclass[onecolumn, draftcls,10pt]{IEEEtran}

\usepackage{amsfonts}                                                          
\usepackage{epsfig}
\usepackage{amsthm}
\usepackage{graphicx}        
\usepackage{latexsym}
\usepackage{amssymb}
\usepackage{amsmath}
\usepackage{parskip}
\usepackage{multirow}
\usepackage{cite}
\usepackage{mathtools}
\usepackage{epstopdf}
\usepackage{etoolbox}
\usepackage{array}
\usepackage{mathptmx}
\usepackage[colorlinks=true]{hyperref}
\usepackage[bottom]{footmisc}
\usepackage{tikz}

\usetikzlibrary{decorations.pathreplacing}

\usepackage{verbatim}
\usepackage{calrsfs}

\DeclareMathAlphabet{\pazocal}{OMS}{zplm}{m}{n}

\newcommand{\mb}{\mathbb}

\let\bbordermatrix\bordermatrix
\patchcmd{\bbordermatrix}{8.75}{4.75}{}{}
\patchcmd{\bbordermatrix}{\left(}{\left[}{}{}
\patchcmd{\bbordermatrix}{\right)}{\right]}{}{}

\newcommand{\sr}{\stackrel}

\newcommand{\rar}{\rightarrow}

\newcommand{\tri}{\sr{\triangle}{=}}

\newcommand{\be}{\begin{equation}}
\newcommand{\ee}{\end{equation}}
\newcommand{\bea}{\begin{eqnarray}}
\newcommand{\eea}{\end{eqnarray}}
\newcommand{\bes}{\begin{eqnarray*}}
\newcommand{\ees}{\end{eqnarray*}}
\newcommand{\bce}{\begin{center}}
\newcommand{\ece}{\end{center}}
\newcommand{\beae}{\begin{IEEEeqnarray}{rCl}}
\newcommand{\eeae}{\end{IEEEeqnarray}}

\def\VR{\kern-\arraycolsep\strut\vrule &\kern-\arraycolsep}
\def\vr{\kern-\arraycolsep & \kern-\arraycolsep}

\newcommand{\ben}{\begin{enumerate}}
\newcommand{\een}{\end{enumerate}}

\newcommand{\hso}{\hspace{.1in}}
\newcommand{\hst}{\hspace{.2in}}

\newcommand{\noi}{\noindent}

\newtheorem{theorem}{Theorem}[section]

\newtheorem{remark}{Remark}[section]
\newtheorem{corollary}{Corollary}[section]

\newtheorem{assumptions}{Assumptions}[section]
\newtheorem{definition}{Definition}[section]
\newtheorem{lemma}{Lemma}[section]
\newtheorem{example}{Example}[section]

\begin{document}


\title{Capacity Achieving Distributions \& Information Lossless Randomized Strategies  for Feedback Channels with Memory: The LQG Theory of Directed Information-Part II}



 \author{
   \IEEEauthorblockN{
     Charalambos D. Charalambous\IEEEauthorrefmark{1},
     Christos Kourtellaris\IEEEauthorrefmark{1},
     Sergey Loyka\IEEEauthorrefmark{2}\\}
   \IEEEauthorblockA{
     \IEEEauthorrefmark{1}Department of Electrical and Computer Engineering\\University of Cyprus, 75 Kallipoleos Avenue, P.O. Box 20537, Nicosia, 1678, Cyprus
     \\
     Email: chadcha@ucy.ac.cy, kourtellaris.christos@ucy.ac.cy}\\
   \IEEEauthorblockA{
     \IEEEauthorrefmark{2}School of Electrical Engineering and Computer Science\\ University of Ottawa, Ontarion, Canada\\
     Email: sergey,loyka@ieee.org}
 }

\maketitle

\begin{abstract}
A methodology is developed to realized optimal channel input conditional distributions, which maximize the finite-time horizon directed information,  for channels with memory and feedback,  by information lossless randomized strategies.
The methodology is applied to  general Time-Varying Multiple Input Multiple Output (MIMO) Gaussian Linear Channel Models (G-LCMs) with memory, subject to average transmission cost constraints of quadratic form.
The realizations of optimal distributions by randomized strategies are shown to exhibit a decomposion into a deterministic part and a random part. The decomposition   reveals the dual role of randomized strategies,   to control the channel output process and to transmit new information over the channels. Moreover, a separation principle is shown  between the computation of the optimal deterministic part and the random part of the randomized strategies. The dual role of randomized strategies generalizes the  Linear-Quadratic-Gaussian (LQG) stochastic optimal control theory to directed information pay-offs. 

The characterizations of feedback capacity  are obtained  from the per unit time limits of finite-time horizon directed information, without imposing \'a priori  assumptions, such as, stability of channel models or ergodicity of channel input and output processes. For time-invariant MIMO G-LCMs with memory, it is shown that whether feedback increases capacity, is directly related to the channel parameters and the transmission cost function, through the solutions of Riccati matrix equations, and moreover
for unstable channels, feedback capacity is non-zero, provided the
power exceeds a critical level.
\end{abstract}

\newpage 

\tableofcontents
\newpage

\section{Introduction}
\label{intro}
Stochastic optimal control theory and a variational equality of directed information are applied in \cite{kourtellaris-charalambousIT2015_Part_1}, to identify the  information structures of optimal channel input conditional distributions, 
 ${\cal P}_{[0,n]} \tri \big\{ {\bf P}_{A_i|A^{i-1}, B^{i-1}}:i=0, 1, \ldots, n\big\} \subset {\cal P}_{[0,n]}(\kappa)$, which maximize the finite-time horizon directed information $I(A^n \rar B^n)$,  for channels with memory and feedback,  defined by
 \begin{align}
C_{A^n \rar B^n}^{FB}(\kappa) 
= \sup_{{\cal P}_{[0,n]}(\kappa)}I(A^n \rar B^n), \hst \hst I(A^n \rar B^n) \tri \sum_{i=0}^n I(A^i; B_i|B^{i-1}) \label{FTFI_CC}
\end{align}
where  $A^n\tri \{A_0, A_1, \ldots, A_n\}$ and $B^n \tri \{B_0, B_1, \ldots, B_n\}$, are channel input and output RVs, respectively, and
${\cal P}_{[0,n]}(\kappa) \subset {\cal P}_{[0,n]}$ denotes a subset of distributions, which satisfy a transmission cost constraint and $\kappa$ is the power. Specifically, given a channel distribution and a transmission cost constraint, the optimal channel input conditional  distributions, which maximize $I(A^n \rar B^n)$, are characterized by conditional independence properties,  called  the  ``information structures''. The information structures  of optimal distributions  simplify the resulting finite-time horizon optimization problem called the  ``characterization of Finite Transmission Feedback Information (FTFI) capacity''. In principle, the information structures of optimal distributions and corresponding   characterizations of FTFI capacity, are analogous to those  of memoryless channels without feedback, which are established via  the well-known upper bounds  
\bea
C_{A^n ; B^n} \tri \max_{{\bf P}_{A^n}} I(A^n; B^n)\leq   \max_{{\bf P}_{A_i}: i=0, \ldots, n}\sum_{i=0}^n  I(A_i; B_i) \leq (n+1) \max_{{\bf P}_A}I(A; B) \equiv (n+1) C \label{DMC}
\eea
 which are  achievable if ${\bf  P}_{A_i|A^{i-1}}={\bf P}_{A_i}, i=0, \ldots, n$ and identically distributed, which implies the joint process $\{(A_i, B_i):i=0, \ldots, n\}$ is identically distributed, and hence ergodic. For memoryless channels  with feedback, (\ref{DMC}) remains valid, provided  it is shown, via the converse to the coding theorem, that feedback does not give better upper bounds \cite{cover-thomas2006}. For memoryless channels  the above bounds are applied in the converse part of the coding theorem, to obtain a tight bound on any achievable rate, while for the direct part,  codes  are generated independently according to ${\bf P}_{A^n}^*(da^n)\tri \otimes_{i=0}^n {\bf P}_{A}^*(da_i)$, where ${\bf P}_A^*$ is the one which corresponds to $C$.  \\
However, to make the  transition from memoryless channels to channels with memory,  without any  \'a priori assumptions, such as, stationarity, ergodicity or information stability, it is necessary to investigate the  characterizations of FTFI capacity $C_{A^n \rar B^n}^{FB}(\kappa)$, and its asymptotic properties via the per unit time limiting versions
\bea
C_{A^\infty \rar B^\infty}^{FB}(\kappa)\tri \liminf_{n \longrightarrow \infty} \frac{1}{n+1} C_{A^n \rar B^n}^{FB}(\kappa) 
\eea
Specifically, it is illustrated in this paper that    \\
1) the information structures of optimal channel input distributions and corresponding characterizations of FTFI capacity,  translate into corresponding information structures  for $C_{A^\infty \rar B^\infty}^{FB}(\kappa)$, and moreover  via the converse coding theorem,  tight bounds on any achievable code rate (of feedback codes) can be obtained, while the direct part of the coding theorem can be investigated,  without unnecessary \'a priori assumptions on  the channel, such as, stationarity, ergodicity, or information stability of the joint process $\{(A_i, B_i): i=0, 1, \ldots\}$;  \\
2) the characterizations of the FTFI capacity reveal several hidden properties of the role of feedback   to affect the channel output process, including fundamental properties of  optimal channel input conditional distributions,  which  achieve $C_{A^\infty \rar B^\infty}^{FB}(\kappa)$,  and properties of channel parameters so that $C_{A^\infty \rar B^\infty}^{FB}(\kappa)$  corresponds to  feedback capacity.  

Application examples are developed to unfold  the role of feedback,  in communication systems with memory. 
\subsection{Contributions}

The main contributions of this second part of the two part investigation, are the following.
\begin{description}
\item[i)] Develop a  methodology to realize  optimal channel input conditional distributions,  by information lossless (one-to-one and onto maps) randomized strategies (deterministic functions)  driven by uniform Random Variables (RVs), and to derive   alternative equivalent characterizations of FTFI capacity. In specific application examples, such as, Gaussian channels with memory, in which the capacity achieving optimal channel input conditional distributions are Gaussian,  the uniform RVs transform the randomized strategies into capacity achieving randomized strategies, driven by Gaussian random processes; 


\item[ii)] identify a dual role (and hidden features)  of optimal channel input conditional distributions and  information lossless randomized strategies, to  control the channel output process,  and to transmit new information over the channels, while they achieve the characterizations of FTFI capacity and feedback capacity;

\item[iii)] identify sufficient conditions, in terms of channel parameters, which define the channel conditional distributions and transmission cost functions (if constraints are imposed on the channel input conditional distributions),  for the quantity   $C_{A^\infty \rar B^\infty}^{FB}(\kappa)$ to correspond to feedback capacity.
\end{description}

The channel conditional distributions considered in this paper, are either one of the following two classes.
\begin{align}
&\mbox{Channel Distributions Class A.}  \hso {\bf P}_{B_i|B^{i-1}, A^{i}}(db_i|b^{i-1}, a^{i}) = {\bf P}_{B_i|B^{i-1}, A_i}(db_i|b^{i-1}, a_i),  \label{CD_C1} \\
&\mbox{Channel Distributions Class B.} \hso {\bf P}_{B_i|B^{i-1}, A^{i}}(db_i|b^{i-1}, a^{i})     = {\bf P}_{B_i|B_{i-M}^{i-1}, A_i}(db_i|b_{i-M}^{i-1}, a_i), \hso
    i=0, \ldots,n \label{CD_C4}
\end{align}
where $M$ is a nonnegative integer. The convention  for $M=0$,  is   ${\bf P}_{B_i|B_{i-M}^{i-1}, A_i}(db_i|b_{i-M}^{i-1}, a_i)= {\bf P}_{B_i| A_i}(db_i| a_i), \hso i=0,1, \ldots, n$, that is, the channel degenerates to a memoryless channel.
  \\
The average  transmission cost constraint is of the form 
\begin{align}
{\cal P}_{[0,n]}(\kappa)\tri & \Big\{ {\bf P}_{A_i|A^{i-1}, B^{i-1}},  i=0, \ldots, n:  \frac{1}{n+1} {\bf E}\Big(\sum_{i=0}^n \gamma_i(T^iA^n, T^iB^{n})\Big) \leq \kappa\Big\} \label{cap_fb_3}
\end{align}
where the  transmission cost functions are either one of the following two classes\footnote{There is no loss of generality in considering $\gamma_i^B(a_i, b_{i-K}^i)$, because by the function restriction, they include  $\gamma_i^B(a_i, T^ib^n)=\gamma_i(a_i, b_{i-K}^{i-L})$ and $\gamma_i(a_i)$, for any nonnegative integers $L\geq K$, and similarly for $\gamma_i^A(a_i, b^i)$.}
\begin{align}        
&\mbox{Transmission Cost Functions  Class A.} \hso  \gamma_i(T^i a^n, T^i b^{n}) ={\gamma}_i^{A}(a_i, b^{i}), \\
&\mbox{Transmission Cost Functions  Class B.} \hso \gamma_i(T^i a^n, T^i b^{n}) ={\gamma}_i^{B}(a_i, b_{i-K}^{i}), \hst i=0, \ldots, n.
\end{align}
where $K$ is nonegative and finite. \\
In this paper, the dual role of optimal channel input distributions and randomized strategies, are  only illustrated for  Multiple-Input Multiple Output (MIMO) Gaussian Linear Channel Models  (G-LCMs) with memory, via a provocative  direct connection to  the Linear-Quadratic-Gaussian (LQG) stochastic optimal control theory,  stability of linear stochastic controlled systems, and  Lyapunov and Riccati matrix equations. Indeed, the LQG stochastic optimal control theory generalizes in a natural way to directed information pey-off functionals. For the readers convenience a short summary of these concepts is given in Appendix~\ref{appendix_B}. These tools are necessary to to treat processes $\{(A_i, B_i): i=0,1, \ldots, \}$, which are not assumed \'a priori to be stationary, ergodic or information stable.  Rather, it is shown via these mathematical concepts that the optimal channel input conditional distribution induces asymptotic ergodicity of the joint process $\{(A_i, B_i): i=0,1, \ldots\}$.\\
For these application examples, it is further shown that the optimal randomized strategies, which achieve FTFI capacity decompose into a deterministic part and a random part. Through this decomposition, a separation principle is shown, between the role of randomized strategies to control the channel output process and to transmit new information over the channel. It is also shown that the  deterministic part  corresponds to the optimal solution of the LQG stochastic optimal control problem, while the  random part  is determined from   water-filling  type equations. There is however, a fundamental difference; in  LQG stochastic optimal control theory, randomized strategies do not incur a better performance than deterministic strategies.

One of the important conclusions of this paper, is that   if the channel input process is the control process, the channel output is the controlled process, the pay-off of LQG stochastic optimal control theory \cite{kumar-varayia1986} is replaced by the  directed information from the control process to the controlled process, then randomized control strategies have an operation meaning in terms of conveying  information from the control process to the controlled process, and exhibit all properties of LQG theory.   \\
In  Section~\ref{mo-dr},  a short summary of the main concepts, methods, and  results obtained  in this paper are presented.

\subsection{Literature Review}
Cover and Pombra \cite{cover-pombra1989} investigated the  non-stationary non-ergodic Additive Gaussian Noise (AGN) channel  with memory, defined by 
\begin{align}
 B_i=A_i+V_i, \hst i=0,1, \ldots, n, \hst \frac{1}{n+1} \sum_{i=0}^{n} {\bf E} \big\{ |A_i|^2\big\} \leq \kappa, \hso \kappa \in [0,\infty) \label{c-p1989}    
\end{align}
where $\{V_i : i=0,1, \ldots, n\}$ is a real-valued (scalar) jointly non-stationary Gaussian process, with covariance   $K_{V^n}$, and ``$A^n$ is causally related to $V^n$'' \footnote{See \cite{cover-pombra1989}, page 39, above Lemma~5, which states   ${\bf P}_{A^n, V^n}(da^n, dv^n)= \otimes_{i=0}^n  {\bf P}_{A_i|A^{i-1}, V^{i-1}}(da_i|a^{i-1}, v^{i-1}) \otimes   {\bf P}_{V^n}(dv^n)$.}. The authors in \cite{cover-pombra1989}, applied the maximizing entropy property of Gaussian distributions (and converse coding arguments) to characterize feedback capacity $C_{W; B^\infty}^{FB, CP}(\kappa) \tri \lim_{n \longrightarrow \infty} \frac{1}{n+1} C_{W; B^n}^{FB, CP}(\kappa)$, via the following characterization of FTFI capacity\footnote{In  \cite{cover-pombra1989}, characterization  (\ref{cp1989}) is obtained via the converse to the coding theorem, by showing that $C_{W; B^\infty}^{FB, CP}(\kappa)$ is an achievable upper bound on the mutual information between uniformly distributed source messages $W$ and channel outputs $B^n$, i.e., on $I(W; B^n)$.}.
\begin{align}
 C_{W; B^n}^{FB, CP}(\kappa) 
  = \max_{  \Big\{ (\overline{\Gamma}^n, K_{\overline{Z}^n}): \frac{1}{n+1} {\bf E} \big\{tr\big(A^n (A^n)^T\big)\big\}\leq \kappa, \hso A^n=\overline{\Gamma}^n V^n + \overline{Z}^n \Big\} }  H(B^n) - H(V^n ) \label{cp1989}
  \end{align}
 where    $\overline{Z}^n \tri \{\overline{Z}_i: i=0, 1, \ldots,n\}$ is a correlated Gaussian process $N( { 0}, K_{\overline{Z}^n})$,  orthogonal to $V^{n}\tri \{V_i: i=0, \ldots, n\}$, and  $\overline{\Gamma}^n$ is lower diagonal time-varying matrix with  deterministic entries (i.e., $A_i=\sum_{j=0}^{i-1}\overline{\gamma}_{i,j}V_j +\overline{Z}_i, i=0, \ldots, n$). Although, the authors in \cite{cover-pombra1989} call $\{\overline{Z}_i: i=0, \ldots, n\}$ an innovations like process, this is not to be confused with the standard definition of an {\it Innovations Process}, which is an orthogonal process.  
Yang Kavcic and Tatikonda  \cite{yang-kavcic-tatikonda2007} considered stationary AGN channels with memory, including the Cover and Pombra AGN channel, and applied dynamic programming as a means of  computing   optimal channel input conditional distributions, which maximize directed information.    Kim \cite{kim2010} investigated    the  stationary  version of  the Cover and Pombra AGN channel, under the assumption that  the noise power spectral density corresponds to   a stationary Gaussian autoregressive moving-average model of order $K$, among other things.
  The AGN channel with memory investigated in \cite{cover-pombra1989,kim2010}, dates back to early investigations by  Butman \cite{butman1969,butman1976}.\\
Recently, feedback capacity problems for  finite alphabet channels, without  transmission cost constraints, are investigated, for   the trapdoor channel  by Permuter,  Cuff, Van Roy and Tsachy \cite{permuter-cuff-roy-weissman2010}, for the  the Ising Channel  by Elishco and Permuter \cite{elishco-permuter2014}, for the the Post$(a,b)$ channel by  Permuter, Asnani and Tsachy \cite{permuter-asnani-weissman2013}, for channels in which the input to the channel and the  channel state are related by  a one-to-one mapping by  Tatikonda, Yang and Kavcic \cite{yang-kavcic-tatikonda2005}, and for the Unit Memory Channel Output (UMCO) channel (i.e., $\big\{{\bf P}_{B_i|B_{i-1}, A_i}: i=0, \ldots, n\}$) by Chen and Berger  \cite{chen-berger2005} (under the assumption the optimal channel input distribution is $\big\{{\bf P}_{A_i|B_{i-1}}: i=0, \ldots, n\big\}$). In  \cite{kourtellaris-charalambous2015}, the BSSC($\alpha,\beta)$ is investigated,  with and without feedback and transmission cost. 
%
Coding theorems for  channels with memory with and without feedback, are developed extensively over the years, for example, in    \cite{dobrushin1959,pinsker1964,gallager1968,blahut1987,cover-thomas2006,ihara1993,verdu-han1994,kramer1998,han2003,kramer2003,tatikonda2000,tatikonda-mitter2009,permuter-weissman-goldsmith2009,gamal-kim2011}.


\subsection{Discussion of  Methodology and Main  Results}
\label{mo-dr} 
The methodology and results  listed under i)-iii) (and discussion under  1), 2)),  are obtained by applying 
 the analogy to stochastic optimal control theory depicted in Fig.~\ref{fig3}, which  states the following. The  information measure $I(A^n \rar B^n)$ is the pay-off, the channel output process $\{B_i: i=0, 1, \ldots, n\}$ is the controlled process, and  the channel input process $\{A_i: i=0,1, \ldots, n\}$ is the control process. \\
\begin{figure}
  \centering
    \includegraphics[width=1.03\textwidth]{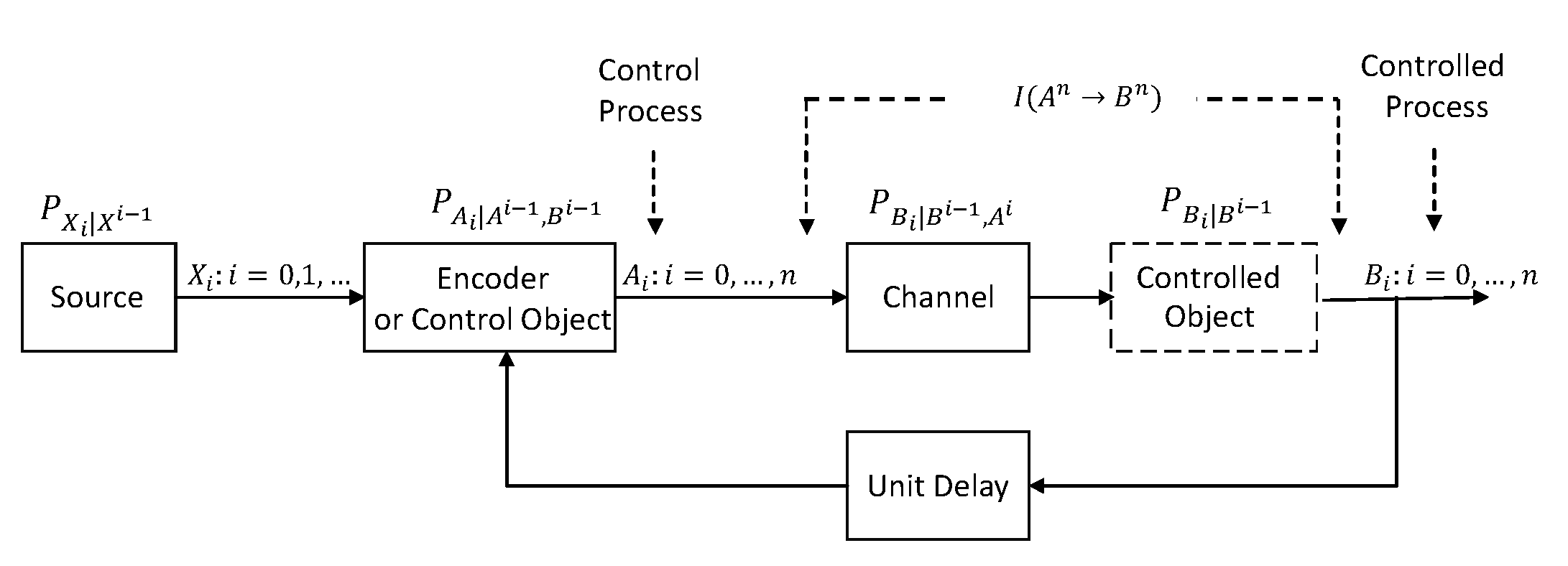}
    \label{fig3}
      \caption{Communication block diagram and its analogy to stochastic optimal control.}
\end{figure}
The main results of the paper are summarized and discussed  below.

\subsubsection{\bf Information Lossless Randomized Strategies and Characterizations of FTFI capacity} Section~\ref{ILE}, describes the methodology  applied in i), to derive alternative equivalent characterizations of FTFI capacity, by realizing optimal channel input conditional distributions, which maximize directed information, using  randomized strategies,  driven by uniform Random Variables.
\\
An alternative equivalent characterization of FTFI capacity is illustrated below. 

{\it Equivalent Characterizations of FTFI Capacity for Class B channels and  transmission cost functions.} Consider  $\Big\{{\bf P}_{A_i|B_{i-M}^{i-1}}, \gamma_i(T^ia^n, T^ib^n)=\gamma_i^B(a_i, b_{i-M}^{i-1}): i=0, \ldots, n\Big\}$, $M \geq 1$.
  In \cite{kourtellaris-charalambousIT2015_Part_1}, it is shown that the maximization of $I(A^n \rar B^n)$ over all distributions ${\bf P}_{A_i|A^{i-1}, B^{i-1}}: i=0, \ldots, n\}$, which satisfy the constraint $\frac{1}{n+1}{\bf E}\big\{ {\gamma}_i^B(A_i, B_{i-M}^{i-1})\big\} \leq \kappa$, satisfy conditional independence
  \bea
   {\bf P}_{A_i|A^{i-1}, B^{i-1}}={\bf P}_{A_i|B_{i-M}^{i-1}}, \hso   i=0, \ldots, n \label{ISI_1}
   \eea
which implies the information structure of the optimal channel input distribution is $B_{i-M}^{i-1}$ for $i=0, 1, \ldots, n$ and  the corresponding 
\begin{align}
\mbox{joint process $\big\{(A_i, B_i): i=0, \ldots, n\big\}$ and output process $\big\{B_i: i=0,\ldots, n\big\}$ are $M-$order Markov}
\end{align}
  and the characterization of the FTFI capacity    is given by the following expression. 
\begin{align}
{C}_{A^n \rar B^n}^{FB,B.M}(\kappa) 
= \sup_{{\bf P}_{A_i|B_{i-M}^{i-1}}, i=0,\ldots,n: \frac{1}{n+1} \sum_{i=0}^n {\bf E}\big\{ {\gamma}_i^B(A_i, B_{i-M}^{i-1})\big\} \leq \kappa    } \sum_{i=0}^n {\bf E}\Big\{
\log\Big(\frac{d {\bf P}_{B_i|B_{i-M}^{i-1}, A_i} (\cdot|B_{i-M}^{i-1},A_i)}{d {\bf P}_{B_i|B_{i-M}^{i-1}}(\cdot|B_{i-M}^{i-1})}(B_i)\Big)
\Big\} . \label{cor-ISR_B.2_intro}
\end{align}
In Theorem~\ref{thm-AC-2}, by  utilizing the characterization of FTFI capacity (\ref{cor-ISR_B.2_intro}), the following are shown.\\
(i) The class of optimal channel input distributions satisfying the average transmission cost constraint,  are realized  by  information lossless randomized  strategies defined by
\begin{align}
&{\cal E}_{[0,n]}^{IL-B.M} (\kappa) \tri  \Big\{ e_i: {\mathbb B}_{i-M}^{i-1} \times [0,1]\longmapsto {\mathbb A}_i,\; a_i=e_i(b_{i-M}^{i-1},u_i), \; i=0, \ldots, n: \hso \mbox{for a fixed $b_{i-M}^{i-1}$ the map } \;  \:   \nonumber \\
&e_i(b_{i-M}^{i-1}, \cdot) \hso \mbox{is one-to-one and onto ${\mb A}_i,$ for $i=0,\ldots, n$}, \hso  \frac{1}{n+1}\sum_{i=0}^n {\bf E}^{e} \Big(\gamma_i^{B}(e_i(B_{i-M}^{i-1}, U_i), B_{i-M}^{i-1})\Big)    \leq \kappa         \Big\} \label{alt_A.1_NCM_A_B}
\end{align}
where $\big\{U_i: i=0, \ldots, n\big\}$ are uniform RVs.\\
(ii) An alternative equivalent characterization of the FTFI capacity (\ref{cor-ISR_B.2_intro}), is given by\footnote{The subscript notation on conditional distributions is suppressed, while superscript notation indicates dependence on the strategies.} 
\begin{align}
{C}_{A^n \rar B^n}^{FB,B.M}(\kappa) = & \sup_{  \big\{e_i(\cdot, \cdot): i=0, \ldots, n\big\} \in {\cal E}_{0,n}^{IL-B.M}(\kappa)        }  \sum_{i=0}^n   {\bf E}^{e}\Big\{\log\Big(\frac{{\bf P}(\cdot| B_{i-M}^{i-1}, e_i(B_{i-M}^{i-1}, U_i))}
  {{\bf P}^{e}(\cdot| B_{i-M}^{i-1})}(B_i)\Big)\Big\}  \label{CMNon-A.1_10_new_Intro} \\
\equiv &   \sup_{ \big\{e_i(b_{i-M}^{i-1},u_i): i=0, \ldots, n\big\} \in {\cal E}_{[0,n]}^{IL-B.M}        }    \sum_{i=0}^n I^{e}(U_i; B_i|B_{i-M}^{i-1})
\end{align}
where the channel output transition probabilities are given by 
\begin{align}
{\bf P}^{e}(db_i|b_{i-M}^{i-1})= \int_{{\mathbb U}_i} {\bf P}(dB_i| B_{i-M}^{i-1}, e_i(b_{i-M}^{i-1},  u_i))   {\bf P}_{U_i}(du_i),  \hso  i=0,1, \ldots, n. \label{LCMNon-C.1_4_B_intro} 
\end{align}
In application examples, the   maximizing strategy, $\big\{e_i^{*}(\cdot, \cdot): i=0, \ldots, n\big\} \in {\cal E}_{[0,n]}^{IL-B.M}(\kappa)$, is transformed to  $e_i^{*}(b_{i-M}^{i-1}, \cdot)=\overline{e}_i^*(b_{i-M}^{i-1}, g_i(\cdot)), g_i: [0,1] \longmapsto {\mathbb  Z}_i, z_i=g_i(u_i)$, where $\big\{Z_i=g_i(U_i): i=0, \ldots, n\big\}$ is a specific random process, which induces the maximizing channel input conditional distribution of the characterization of FTFI capacity (\ref{cor-ISR_B.2_intro}).\\
Further, the following are illustrated, via application examples of Multiple-Input Multiple-Output (MIMO) Gaussian-Linear Channel Models (G-LCM).

\subsubsection{\bf Dual Role of Randomized Information Lossless Strategies \& LQG Theory.} By  utilizing  (\ref{CMNon-A.1_10_new_Intro}), it is shown that,  in general, information lossless randomized strategies have a {\it Dual Role}, specifically,  to
\begin{description}
\item[(a)] optimally control the channel output process $\{B_i: i=0, 1, \ldots, n\}$,  and to

\item[(b)] communicate information optimally, by ensuring no information loss incurs, when randomized strategies are used to realize optimal channel input distributions. 
\end{description}
In Theorem~\ref{DP-UMCO} (Section~\ref{G-LCM-A-LQG})), the dual role of information lossless randomized strategies (and several properties),  are illustrated for  MIMO Time-Varying Gaussian Linear Channel Models G-LCMs with memory.
The following application example illustrates this dual role.  \\
{\it Alternative  Characterization of FTFI Capacity for G-LCM-B.1 and The LQG Theory.} Consider the MIMO G-LCM-B.1, corresponding to channel Class B and transmission cost Class B,  defined by\footnote{The fundamental difference between $Q_{i,i-1} \neq 0$ versus $Q_{i,i-1}=0, i=0, \ldots, n$ and its implications on the maximum rate of transmitting information over this channel,  is discussed shortly. The subsequent statements are derived in Theorem~\ref{DP-UMCO}.} 
\begin{align}
&B_i   = C_{i,i-1} \; B_{i-1} +D_{i,i} \; A_i + V_{i},\hso B_{-1}=b_{-1},  \hso i= 0, \ldots, n, \label{LCM-A.1_a_Intr} \\
 &\frac{1}{n+1} \sum_{i=0}^n {\bf E} \Big\{ \langle A_i, R_i A_i \rangle + \langle B_{i-1}, Q_{i,i-1} B_{i-1} \rangle \Big\}\leq \kappa,  \\
& {\bf P}_{V_i|V^{i-1}, A^i}(dv_i|v^{i-1}, a^i)={\bf P}_{V_i}(dv_i),  \hso  V_i \sim N(0, K_{V_i}), \hso  i=0, \ldots,n, \\
& (C_{i,i-1}, D_{i,i}) \in {\mb R}^{p\times p} \times {\mb R}^{p\times q}, \; (R_i, Q_{i,i-1})\in {\mb R}^{q\times q} \times {\mb R}^{p\times p}, \; R_i=R_i^T \succ 0, Q_{i,i-1}=Q_{i,i-1}^T \succeq 0, \; i=0, \ldots, n
\end{align}
where  $\langle \cdot, \cdot\rangle $ denotes inner product of elements of linear spaces.\\
In Section~\ref{ex-LCM-B.1},  the following are shown (by using (\ref{CMNon-A.1_10_new_Intro}), with $M=1$).\\
(iii)  The optimal conditional channel input conditional distribution is Gaussian of the form $\big\{{\bf P}_{A_i|B_{i-1}}^g(da_i|b_{i-1}): i=0, \ldots, n\big\}$ and satisfies the average transmission cost constraint, and such distributions are realized  by linear and Gaussian  information lossless randomized strategies $e(\cdot) \in {\cal E}_{0,n}^{IL-B.1}(\kappa)$, defined by   the set 
 \begin{align}
&{\cal E}_{[0,n]}^{IL-B.1}(\kappa)   \tri \Big\{ A_i^g=g_i^{B.1}(B_{i-1}^g)+Z_i = \Gamma_{i,i-1} B_{i-1}^g+ Z_i, \hso Z_i \perp B^{g,i-1},\;  \{Z_i: i=0, \ldots, n\} \hso \mbox{independent  process},\nonumber \\
& Z_i \sim N(0, K_{{Z}_i}), \: K_{Z_i} \succeq 0, \hso i=0,\ldots, n: \frac{1}{n+1}{\bf E}^{g^{B.1}} \Big\{ \sum_{i=0}^n     \Big[\langle A_i^g, R_{i,i} A_i^g \rangle + \langle B_{i-1}^g, Q_{i,i-1} B_{i-1}^g \rangle \Big] \Big\} \leq \kappa \Big\} \label{Rand_Stra}
\end{align}
where $ \cdot\perp \cdot$ means the processes are independent. Thus,  information lossless randomized strategies in   ${\cal E}_{[0,n]}^{IL-B.1}(\kappa)$ are decomposed  into two orthogonal parts, one of which is an innovations process (i.e., independent process). \\
(iv) The characterization of FTFI capacity  of the MIMO-G-LCM.B.1 is given by the following expression. 
\begin{align}
{C}_{A^n \rar B^n}^{FB,B.1} (\kappa)= \sup_{\big\{\big(\Gamma_{i,i-1}, Z_i\big),   i=0,\ldots,n   \big\} \in  {\cal E}_{[0,n]}^{IL-B.1}(\kappa)   }  \frac{1}{2} \sum_{i=0}^n  \ln \frac{ | D_{i,i} K_{Z_i} D_{i,i}^T+K_{V_i}|}{|K_{V_i}|} .  \label{LCM_B.1_3_Intr_1}
\end{align}


The decomposition 
\bea
A_i^g =\Gamma_{i,i-1}B_{i-1}^g +Z_i\equiv g_i^{B.1}(B_{i-1}^g)+Z_i,\hso  i=0, \ldots, n
\eea
 implies that  the feedback function $\{g_i^{B.1}\equiv \Gamma_{i,i-1}: i=0, \ldots, n\}$ is  the feedback control law or strategy, which controls the output process $\{B_i^g: i=0, \ldots, n\}$, while the orthogonal innovations process $\{Z_i: i=0, \ldots, n\}$ is responsible to convey new information to the output process, both chosen to maximize (\ref{LCM_B.1_3_Intr_1}). 


In Theorem~\ref{DP-UMCO}, this interpretation is utilized to solve  the  extremum problem (\ref{LCM_B.1_3_Intr_1}), via its relation to the Linear-Quadratic-Gaussian (LQG) stochastic optimal control theory (with randomized controls). The following are shown. \\
(v) The characterization of FTFI capacity is  given by 
\begin{align}
C_{A^n\rar B^n}^{FB, B.1}(\kappa) = \inf_{s \geq 0} \Big\{ -\int_{{\mathbb R}^p} \langle b_{-1}, P(0) b_{-1}\rangle {\bf P}_{B_{-1}}(db_{-1}) +  r(0)\Big\}  \label{FT_IC}
\end{align}
where  $\big\{P(i): i=0, \ldots,n\big\}$ is a solution of the Riccati difference matrix equation
\begin{align}
P(i) &=C_{i,i-1}^T P(i+1) C_{i,i-1}+s Q_{i,i-1} \nonumber \\
&-C_{i,i-1}^T P(i+1) D_{i,i}\Big(D_{i,i}^T P(i+1) D_{i,i}+s R_{i,i}\Big)^{-1} \Big(C_{i,i-1}^T P(i+1) D_{i,i}\Big)^T, \hso i=0, \ldots, n-1, \hso P(n)=s Q_{n,n-1} \label{DP-UMCO_C12_aa_intro}
\end{align}
$s\geq 0$ is the Lagrange multiplier associated with the transmission cost constraint,    
the optimal random part of the strategy $\{K_{Z_i}^*: i=0, \ldots, n\}$ (covariance of innovations  process) is found from a sequential water filing problem 
\begin{align}
r(i)= & r(i+1) + \sup_{ K_{Z_i}\in {\mb  S}_+^{q\times q}}\Big\{ \frac{1}{2} \log \frac{ | D_{i,i} K_{Z_i} D_{i,i}^T+K_{V_i}|}{|K_{V_i}|} - tr\Big(s\; R_{i,i} K_{Z_i}\Big) 
- tr \Big(P(i+1) \Big[D_{i,i} K_{Z_i}D_{i,i}^T+ K_{V_i}\Big]\Big)\Big\}, \hso i=0, \ldots, n-1,  \label{DP-UMCO_C13_NN}  \\
r(n)=&\sup_{ K_{Z_n}\in  {\mb S}_+^{q\times q}} \Big\{\frac{1}{2} \log \frac{ | D_{n,n} K_{Z_n} D_{n,n}^T+K_{V_n}|}{|K_{V_n}|}+s (n+1)\kappa- tr\Big( s R_{n,n} K_{Z_n}\Big)\Big\} \label{DP-UMCO_C20_NN}
\end{align}
and the  optimal deterministic part of the randomized strategy, $\{g_i^{B.1,*}(\cdot): i=0, \ldots, n\}$, is given by 
\begin{align}
g_i^{B.1,*}(b_{i-1})=& -\Big(D^T P(i+1) D + s R\Big)^{-1} D^T P(i+1) C \: b_{i-1}\equiv \Gamma_{i,i-1}^* b_{i-1}, \hso i=0, \ldots, n-1, \hso g_n^{B.1,*}(b_{n-1})=0. \label{STR_1}  
\end{align}
The above solution illustrates the separation principle, between the computation of the deterministic part $\{g_i^{B.1}(B_{i-1}): i=0, \ldots, n\}$ and random part $\{Z_i\sim K_{Z_i}: i=0, \ldots, n\}$ of the randomized strategy,
in  that, the latter  can be found, by first computing the former. Moreover, the properties of solutions $\big\{P(i): i=0, \ldots,n\big\}$ to the Riccati equation, such as, $P(i)\succ 0$ or $P(i) \succeq 0,  i=0, \ldots, n$ (positive definite or positive semi definite),  depend on ``s'' and the properties of the parameters of the channel and the transmission cost function, $\Big\{ C_{i,i-1}, D_{i,i}, R_{i,i}, Q_{i,i-1}: i=0, \ldots, n\Big\}$. \\
(vi) The optimal strategy (\ref{STR_1}) is precisely  the solution of the following  LQG stochastic optimal control problem \cite{kumar-varayia1986}. This connection is made explicit in Remark~\ref{rem-LQG}.
(vii) If the channel is time-invariant  with $\Big\{ C_{i,i-1}=C, D_{i,i}=D, K_{V_i}=K_V, R_{i,i}=R, i=0, \ldots, n, Q_{i,i-1}=Q, i=0, \ldots, n-1, Q_{n,n-1}=M\Big\}$, from (\ref{FT_IC}),  whether $C_{A^\infty \rar B^\infty}^{FB, B.1}(\kappa)\tri \lim_{\longrightarrow \infty} \frac{1}{n+1} C_{A^n \rar B^n}^{FB, B.1}(\kappa)$ exists and corresponds to feedback capacity is determined from the properties of solutions to the following algebraic matrix Riccati   equation.
\begin{align}
P=C^T P C+s Q -C^T P D\Big(D^T P D+s R\Big)^{-1} \Big(C^T P D\Big)^T. \label{DP-UMCO_C12_aa_intro_alg}
\end{align} 
Moreover, whether feedback increases capacity is  determined from the solutions of matrix Riccati   equation (\ref{DP-UMCO_C12_aa_intro_alg}).

\subsubsection{\bf Application Example: Feedback versus No Feedback \& The Infinite Horizon LQG Theory.} 
In Sections~\ref{putl-ex},  the per unit time limit of the characterizations of FTFI capacity of  Time-Invariant G-LCM-Bs are investigated. \\It is shown that whether feedback increases capacity, is determined from the unique solution  of the Riccati algebraic matrix equation (\ref{DP-UMCO_C12_aa_intro_alg}). This is established via  direct connections to the infinite-horizon LQG stochastic optimal control theory and stability of linear stochastic controlled systems, and associated Lyapunov  equations and Riccati matrix equations. 
Indeed, even if the channel defined by (\ref{LCM-A.1_a_Intr}) is unstable (i.e., any of the eigenvalues of matrix $C$ is greater than one), under certain conditions, which are specified by  $(C, D, R, Q, K_V)$, the optimal deterministic part of the randomized strategy stabilizes the channel via feedback,  ensures existence of a unique invariant joint distribution of the joint process $\{(A_i, B_i): i =0, \ldots, n\}$,  marginal distribution of the channel output process, and ensures that $C_{A^\infty \rar B^\infty}^{FB}(\kappa)$ exists and corresponds to feedback capacity. \\
The following simple example illustrates several hidden properties of optimal channel input distributions, and that  feedback  capacity and capacity without feedback  are determined from the properties of the solutions to the algebraic matrix Riccati equation (\ref{DP-UMCO_C12_aa_intro_alg}). 

{\it Special Case-the Time-Invariant Scalar Channel with $p=q=1, R=1,  Q=0$ and $(C,D)$ arbitrary.} For these choices of parameters the following are  shown (and independently in Remark~\ref{rem_G-LCM-B.1-LSCS}). The  steady state solutions of Riccati (quadratic) equation   (\ref{DP-UMCO_C12_aa_intro_alg}),  and corresponding optimal determinist part of the randomized  strategy are given by the following equations. 
\begin{align}
& P\Big(D^2 P + s \Big[1-C^2\Big]\Big)=0 \; \Longrightarrow \; P_1=  0, \;  P_2=s \frac{C^2-1}{D^2}, \label{Exact_Sol_1}\\
& g^{B.1,*}(b) = \Gamma^* b, \hst \Gamma^*= - \Big(D^2 P +s\Big)^{-1} D P C= \left\{ \begin{array}{lll} 0 & \mbox{if}\hso  P=P_1 \\
-\frac{C^2-1}{CD} & \mbox{if} \hso P=P_2.  \end{array} \right.
\end{align}
where  $P=P_1$ implies the optimal channel input distribution  does not use feedback. \\
(vii) {\it The Feedback Capacity}.  The optimal strategy which achieves feedback capacity ${C}_{A^\infty \rar B^\infty}^{FB}(\kappa)$ is given by
\bea
\Big(\Gamma^*, K_Z^*\Big)= \left\{ \begin{array}{llll} (0, \kappa), & \kappa \in [0, \infty) & \mbox{if}\hso  |C|<1 \\
\Big(-\frac{C^2-1}{CD},\frac{D^2 \kappa + K_V (1-C^2)}{C^2 D^2}\Big), & \kappa \in [\kappa_{min}, \infty), \hso \kappa_{min} \tri   \frac{(C^2-1)K_V}{D^2} & \mbox{if} \hso |C|>1 \\
\Big(-\frac{C^2-1}{CD},0\Big), & \kappa \in [0,\kappa_{min}], \hso  & \mbox{if} \hso |C|>1
 \end{array} \right. \label{dual_CAP_CI}
\eea
and the corresponding feedback capacity is given by the following expression.
\begin{align}
{C}_{A^\infty \rar B^\infty}^{FB}  (\kappa)= \left\{ \begin{array}{llll}  \frac{1}{2} \ln \frac{  D^2 \; \kappa +K_{V}}{K_{V}} & \mbox{if}  & |C|<1, & i.e., \:  K_Z^*=\kappa  \\
 \frac{1}{2} \ln \frac{ D^2 K_Z^* +K_{V}}{K_{V}} & \mbox{if}  & |C|>1, &  \kappa \in [\kappa_{min}, \infty) \\
 0 & \mbox{if}  & |C|>1, &  \kappa \in [0, \kappa_{min}]. 
\end{array} \right.  \label{dual_CAP}
\end{align}
The feedback capacity expression (\ref{dual_CAP}), illustrates that there are multiple regimes, depending on whether the channel is stable, that is, $|C|<1$   or unstable $|C|>  1$. Moreover, for unstable channels $|C|>1$, feedback capacity is zero, unless the power $\kappa$ allocated for transmission,  exceeds the critical level $\kappa_{min}$. From the above expressions, it follows that the capacity achieving input distribution satisfies the  
conditional independence 
\begin{align}
{\bf P}_{A_i|B_{i-1}}^*(da_i| b^{i-1})=   \left\{ \begin{array}{lll}   {\bf P}_{A_i}^{g,*}(da_i) \sim N(0, \kappa), \hso \kappa \in [0, \infty) & \mbox{if} & |C|<1  \\
{\bf P}_{A_i|B_{i-1}}^{g,*}(da_i|b_{i-1}) \sim N(\Gamma^*, K_{Z}^*), \hso \kappa \in [\kappa_{min}, \infty) & \mbox{if} & |C|>1 .
\end{array} \right.
\end{align}
This shows that if the channel is stable, $|C|<1$, then feedback does not increase capacity, for the following reasons. As far as the limit ${C}_{A^\infty \rar B^\infty}^{FB}  (\kappa)$ is concerned,  the joint and output processes asymptotically converges to   ergodic processes, and   there is no  incentive to apply feedback, since the controlled process-the channel output process $\{B_i: i=0, \ldots, n\}$, does not appear, neither in the transmission cost constraint nor in the characterization of the FTFI capacity expression given by (\ref{LCM_B.1_3_Intr_1}). However, if $Q\neq 0$, then the controlled process $\{B_i: i=0, \ldots, n\}$ is represented in the pay-off, and hence there is an incentive to apply feedback.   \\  
(ix) {\it Capacity Without Feedback. } The capacity of channel (\ref{LCM-A.1_a_Intr}) without feedback, denoted by $C_{A^\infty; B^\infty}^{noFB}(\kappa)$, is  obtained directly from the characterization of FTFI capacity (\ref{dual_CAP_CI}), (\ref{dual_CAP}).
Clearly, for stable channels, i.e., $|C|<1$, the capacity without feedback $C_{A^\infty; B^\infty}^{noFB}(\kappa)$, is precisely that of a memoryless channel that corresponds to  (\ref{LCM-A.1_a_Intr}) with $C=0$, i.e., $B_i=D A_i+ V_i, i=0, \ldots$, i.e.,   the memory of the channel (\ref{LCM-A.1_a_Intr}), does not increase  capacity.  Moreover,  if channel (\ref{LCM-A.1_a_Intr})  is unstable, i.e., $|C|>1$, then $C_{A^\infty; B^\infty}^{noFB}(\kappa)=0$, for any power $\kappa \in [0, \infty)$.\\
In conclusion, the expressions of capacity without feedback and feedback capacity, coincide for the case of stable channels, i.e., $|C|<1$. This is attributed to the dual role of randomized strategies, specifically, the  role of the deterministic part to control the channel output process. Since in this case, the channel is  stable and $Q=0$, no role is assigned to  the randomized strategy, except to transmit information. However, if $Q\neq 0$ but the channel is stable, i.e., $|C|<1$, the above observation may not hold. \\
(x) For unstable channels,  there is a universal lower bound on the feedback capacity, which is obtained by evaluating ${C}_{A^\infty \rar B^\infty}^{FB}  (\kappa) $ at $\kappa=\kappa_{min}$, expressed  in terms of  the  logarithm of the unstable eigenvalues of the channel, as follows.
\begin{align}
\mbox{If} \hso |C| >1 \hso \mbox{then} \hso {C}_{A^\infty \rar B^\infty}^{FB, B.1}  (\kappa) \geq \ln |C|, \hst  \forall \kappa \in [\kappa_{min}, \infty). \label{Exact_Sol_2}
\end{align}
The lower bound can be interpreted as the least number of bits which should be communicated to the encoder to ensure a non-zero rate is feasible. 
 
\subsubsection{\bf Generalizations to Gaussian Channels with Arbitrary Memory.} All properties discussed above are shown to hold,    for general  MIMO G-LCM-B.1 and  G-LCM-B; they are obtained  by invoking  properties of matrix algebraic Riccati  equations.
These  properties  illustrate fundamental connections between capacity of channels with feedback, without feedback  and linear stochastic controlled system theory, and the LQG stochastic optimal control theory. 

\subsubsection{\bf Relation Between Characterizations of FTFI Capacity and Coding Theorems.} In Section~\ref{c-thms}, the 
 the importance of  the characterizations of FTFI capacity  are discussed in the context the direct and converse parts of channel coding theorems. Specifically, sufficient conditions are identified so that the per unit time limits of the characterizations of FTFI capacity, corresponds to feedback capacity.

\section{Information Structures of Channel Input Distributions  of Extremum Problems of Feedback Capacity}
\label{IS_CA}
In this section, the notation used throughout the paper is established, and the information structures of optimal channel input distributions, which maximize directed information,   are recalled from \cite{charalambous-stavrou2013aa}.

\begin{align}
& {\mathbb Z}: \hso  \mbox{set of  integer};  \nonumber \\
& {\mathbb N}: \hso  \mbox{set of nonnegative integers} \hso \{0, 1,2,\dots\}; \nonumber \\
& {\mathbb R}: \hso  \mbox{set of  real numbers};  \nonumber \\
& {\mathbb C}: \hso \mbox{set of complex numbers;} \nonumber \\
& {\mathbb R}^n: \hso  \mbox{set of  $n$ tuples of real  numbers}; \nonumber \\
&{\mb S}_+^{p \times p}: \hso \mbox{set of symmetric positive semidefine $p\times p$ matrices $A \in {\mathbb R}^{p \times p}$}; \nonumber \\
&\langle \cdot, \cdot \rangle: \hso \mbox{inner product of elements of vectors spaces;} \nonumber \\
&{\mb S}_{++}^{p \times p}: \hso \mbox{subset of positive definite matrices of the set ${\mb S}_+^{p \times p}$;} \nonumber \\
&{\mathbb D}_o \tri \big\{c \in {\mathbb C}: |c| <1\big\}: \hso \mbox{open unit disc of the space of compex numbers ${\mathbb C}$}; \nonumber \\
&spec(A) \subset {\mathbb C}: \hso \mbox{spectrum of a matrix $A \in {\mathbb R}^{p \times p}$ (set of all its eigenvalues)}; \nonumber \\
&(\Omega, {\cal F}, {\mathbb P}): \mbox{probability space, where ${\cal F}$ is the $\sigma-$algebra generated by subsets of $\Omega$}; \nonumber \\ 
& {\cal  B}({\mathbb  W}): \hso \mbox{Borel $\sigma-$algebra of a given topological space  ${\mathbb W}$};  \nonumber \\
&{\cal M}({\mathbb W}): \hso \mbox{set of all probability measures on ${\cal  B}({\mathbb W})$ of a Borel space ${\mathbb W}$}; \nonumber\\
&{\cal K}({\mathbb V}|{\mathbb W}): \hso \mbox{set of all stochastic kernels on  $({\mathbb V}, {\cal  B}({\mathbb V}))$ given $({\mathbb W}, {\cal  B}({\mathbb W}))$ of Borel spaces ${\mathbb W}, {\mathbb V}$}.\nonumber\\
&X \perp Y: \hso \mbox{Independence of RVs $X$ and  $Y$}.\nonumber 
\end{align}

All spaces are complete separable metric spaces, i.e.,  Borel spaces. This generalization is adopted   to treat simultaneously discrete, finite alphabet,  real-valued ${\mathbb R}^k$ or complex-valued ${\mathbb C}^k$ random processes for any positive integer $k$,  etc. 
 The product measurable space of the two measurable spaces  $({\mb X}, {\cal  B}({\mb X}))$ and $({\mb Y}, {\cal  B}({\mb Y}))$ is denoted by $({\mb X} \times {\mb Y}, {\cal  B}({\mb X})\otimes  {\cal  B}({\mb Y}))$, where  ${\cal  B}({\mb X})\otimes  {\cal  B}({\mb Y})$ is the product $\sigma-$algebra generated by rectangles $\{A \times B:  A \in {\cal  B}({\mb X}), B\in  {\cal  B}({\mb Y})\}$. \\
The probability distribution $ {\bf P}(\cdot) \equiv {\bf P}_X(\cdot)$ on  $({\mb X}, {\cal  B}({\mb X}))$ induced by a Random Variable (RV)  on  $(\Omega, {\cal F}, {\mathbb P})$ by the mapping $X: (\Omega, {\cal F}) \longmapsto ({\mb X}, {\cal  B}({\mb X}))$ is defined by   \footnote{The subscript on $X$ is often omitted.}.
\begin{align}
{\bf P}(A) \equiv  {\bf P}_X(A)  \tri {\mathbb P}\big\{ \omega \in \Omega: X(\omega)  \in A\big\},  \hso  \forall A \in {\cal  B}({\mb X}).
 \end{align}
If the cardinality of ${\mb X}$ is finite then the RV is finite-valued  and it is called a finite alphabet RV. \\
Given another RV $Y: (\Omega, {\cal F}) \longmapsto ({\mb Y}, {\cal  B}({\mb Y}))$, then  ${\bf P}_{Y|X}(dy| X)(\omega)$ is called the conditional distribution of RV $Y$ given RV $X$. The conditional distribution of RV $Y$ given $X=x$ is denoted by ${\bf P}_{Y|X}(dy| X=x)  \equiv {\bf P}_{Y|X}(dy|x)$. 
The family of such conditional distributions on $({\mb Y}, {\cal B}({\mb Y})$ parametrized by $x \in {\mb X}$, is defined by $${\cal K}({\mb Y}| {\mb X})\tri \big\{{\bf P}(\cdot| x) \in {\cal M}({\mathbb Y}): \hso x \in {\mathbb X}\hso \mbox{and  $\forall F \in {\cal  B}({\mb Y})$, the function ${\bf P}(F|\cdot)$ is ${\cal  B}({\mb X})$-measurable.}\big\}.$$

The channel input and  channel output alphabets are  sequences of  measurable spaces $\{({\mb A}_i,{\cal  B}({\mb A }_i)):i\in\mathbb{N}\}$ and  $\{({\mb  B}_i,{\cal  B}({\mb  B}_i)):i\in\mathbb{N}\}$, respectively, and 
their history spaces  are the product spaces ${\mb A}^{\mathbb{N}}\tri {{\times}_{i\in\mathbb{N}}}{\mb A}_i,$ ${\mb  B}^{\mathbb{N}}\tri {\times_{i\in\mathbb{N}}}{\mb  B}_i$.  These spaces are endowed with their respective product topologies, and  ${\cal  B}({\Sigma}^{\mathbb{N}})\tri \otimes_{i\in\mathbb{N}}{\cal  B}({\Sigma }_i)$  denote the $\sigma-$algebras on ${\Sigma }^{\mathbb{N}}$, where ${\Sigma}_i \in  \big\{{\mb A}_i, {\mb  B}_i\big\}$,  ${\Sigma}^{\mathbb{N}} \in  \big\{{\mb A}^{\mathbb N}, {\mb  B}^{\mathbb N}\big\}$,  generated by cylinder sets. Similarly,  for  ${\cal  B}({\Sigma }^n)$, when $n \in {\mb N}$ is finite.   Points in ${\Sigma}^{n}$ are denoted by  $z^n\tri \{z_0,z_1,\ldots,z_n\}\in{\Sigma}^n$, while  points in ${\Sigma }_k^m \tri \times_{j=k}^m {\Sigma}_j$ are denoted by $z_{k}^m \tri \{z_k, z_{k+1}, \ldots, z_m\} \in {\Sigma}_k^m$,   $(k, m)\in   {\mathbb N} \times {\mathbb N}$. 

\noi{\bf Channel Distribution with Memory.}  A sequence of stochastic kernels or distributions defined by 
\begin{align}
{\cal C}_{[0,n]} \tri \Big\{ {\bf P}_{B_i|B^{i-1}, A^i}  =Q_i(db_i|b^{i-1},a^{i})\in {\cal K}({\mb  B}_i| {\mb  B}^{i-1} \times {\mb A}^i) :  i=0,1, \ldots, n \Big\}. \label{channel1}
\end{align}
 At each time instant $i$ the conditional distribution of channel output $B_i$  is affected causally by previous channel output symbols $b^{i-1} \in {\mb B}^{i-1}$ and current and previous channel input symbols $a^{i} \in {\mb A}^i, i=0,1, \ldots, n$. 

\noi{\bf Channel Input Distribution with Feedback.}  A  sequence of stochastic kernels defined by 
\bea
{\cal P}_{[0,n]} \tri  \Big\{  {\bf P}_{A_i|A^{i-1}, B^{i-1}}= P_i(da_i|a^{i-1},b^{i-1})\in  {\cal K}({\mb A}_i| {\mb A}^{i-1} \times {\mb  B}^{i-1}):   i=0,1, \ldots, n \Big\}. \label{rancodedF}
\eea
At each time instant $i$ the conditional  distribution of channel input $A_i$  is affected causally by past  channel inputs and  output symbols  $(a^{i-1}, b^{i-1}) \in {\mb A}^{i-1} \times {\mb B}^{i-1}, i=0,1, \ldots, n$. 

\noi{\bf Transmission Cost.}  The cost of transmitting and receiving symbols $a^n\in {\mb A}^n, b^n \in {\mb B}^n $ over the  channel    is a  measurable function $c_{0,n}:{\mb A}^n\times{\mb  B}^{n} \longmapsto [0,\infty)$. The set of channel  input distributions with transmission cost is defined by 
\begin{align}
&{\cal P}_{[0,n]}(\kappa) \tri   \Big\{  P_i(da_i|a^{i-1}, b^{i-1}) \in {\cal K}({\mb A}_i| {\mb A}^{i-1}\times {\mb  B}^{i-1}),  i=0, \ldots, n: \nonumber \\
&\frac{1}{n+1} {\bf E}^{ P} \Big( c_{0,n}(A^n, B^{n}) \Big)\leq  \kappa\Big\}  \subset {\cal P}_{[0,n]}, \hst  c_{0,n}(a^n, b^{n}) \tri \sum_{i=0}^n \gamma_i(T^ia^n, T^ib^{n}), \;      \kappa \in [0,\infty) \label{rc1}
\end{align} 
where ${\bf E}^{ P }(\cdot)$ denotes expectation with respect to the the joint distribution, and superscript ``P'' indicates its dependence  on the choice of  conditional distribution $\{P_i(da_i|a^{i-1}, b^{i-1}) : i=0, \ldots, n\} \in {\cal P}_{[0,n]}(\kappa)$.

{\bf  FTFI Capacity.}  Given any channel input conditional  distribution $\big\{{ P}_i(da_i|a^{i-1}, b^{i-1}): i=0,1, \ldots, n\big\} \in {\cal P}_{[0,n]}(\kappa)$ and any channel distribution $\big\{Q(db_i| b^{i-1}, a^{i}): i=0,1, \ldots, n\big\}\in {\cal C}_{[0,n]}$,   the  induced joint distribution  ${\bf P}^{P}(da^n, db^n)$  is uniquely defined, as follows.
\begin{align}
 {\mathbb P}\big\{A^n \in d{a}^n, B^n \in d{b}^n\big\}  \tri\;  &
{\bf P}^P(da^n, db^n)
= \otimes_{j=0}^n \Big( { \bf P}(db_j|b^{j-1}, a^j) \otimes {\bf P}(da_j| a^{j-1}, b^{j-1})\Big) \label{CIS_2gde2new} \\
=&\otimes_{j=0}^n \Big(Q_j(db_j|b^{j-1}, a^j)\otimes P_j(da_j|a^{j-1}, b^{j-1})\Big). \label{CIS_2gg_new} 
\end{align}
The joint distribution of $\big\{B_i: i=0, \ldots, n\big\}$ and its conditional distribution are defined by\footnote{Throughout the paper the superscript notation ${\bf P}^P(\cdot), \Pi_{0,n}^P(\cdot), etc., $ indicates the dependence of the distributions on the channel input conditional distribution.}
\begin{align}
{\mathbb  P}\big\{B^n \in db^n\big\} \tri \; & {\bf P}^{P}(db^n) =  \int_{{\mb A}^n}  {\bf P}^{ P}(da^n, db^n)   \label{CIS_3g}\\
\equiv \; &   \Pi_{0,n}^{P}(db^n) =\otimes_{i=0}^n \Pi_i^{ P}(db_i|b^{i-1}) \label{MARGINAL} \\
\Pi_i^{ P}(db_i|b^{i-1})= \; &  \int_{{\mb A}^i} Q_i(db_i|b^{i-1}, a^i)\otimes P_i(da_i|a^{i-1}, b^{i-1}) \otimes {\bf P}^{P}(da^{i-1}|b^{i-1}), \hso i=0, \ldots, n. \label{CIS_3a}
\end{align}
The above distributions are parametrized by  either a fixed $B^{-1}=b^{-1} \in {\mathbb B}^{-1}$ or a fixed distribution ${\bf P}(db^{-1})=\mu(db^{-1})$. \\
 Directed  information   $I(A^n \rar B^n)$  is defined   by 
\begin{align}
I(A^n\rar B^n) \tri  &\sum_{i=0}^n {\bf E}^{{ P}} \Big\{  \log \Big( \frac{dQ_i(\cdot|B^{i-1},A^i) }{d\Pi_i^{{ P} }(\cdot|B^{i-1})}(B_i)\Big)\Big\} \label{CIS_6_a}  \\
=& \sum_{i=0}^n \int_{{\mb A}^{i} \times {\mb  B}^{i}   }^{}   \log \Big( \frac{ dQ_i(\cdot|b^{i-1}, a^i) }{d\Pi_i^{{ P}}(\cdot|b^{i-1})}(b_i)\Big) {\bf P}^{ P}( da^i, db^i)  \label{CIS_6}
\end{align}

The  FTFI capacity $C_{A^n \rar B^n}^{FB}(\kappa)$ and without transmission cost constraints $C_{A^n \rar B^n}^{FB}$ are  defined as follows.  
\begin{align}
C_{A^n \rar B^n}^{FB}(\kappa) \tri \sup_{ {\cal P}_{[0,n]}(\kappa) } I(A^n\rar B^n), \hst C_{A^n \rar B^n}^{FB} \tri \sup_{ {\cal P}_{[0,n]}} I(A^n\rar B^n). \label{prob2}
\end{align}
For the per unit time limiting version $C_{A^\infty\rar  B^\infty}^{FB}(\kappa)$ of   $C_{A^n\rar  B^n}^{FB}(\kappa)$, to represent  feedback capacity, and thus characterize the supremum of all achievable rates (via direct and converse  channel coding theorems), the following  assumption is imposed  throughout the paper. For any process $\{X_i: i=0, \ldots, \}$, which may represent     the source process
 to be encoded and transmitted over the channel,  the following conditional independence, pointed out  by Massey  \cite{massey1990} holds.
\begin{align}
{\bf P}_{B_i|B^{i-1}, A^i, X^k}={\bf P}_{B_i|B^{i-1}, A^i} \hso \Longleftrightarrow \hso X^k \leftrightarrow (A^i, B^{i-1})  \leftrightarrow B_i, \hso  \forall k \in \{0,1, \ldots,\},\hso i=0, \ldots, \label{CI_Massey} 
\end{align}

The next two theorems are  derived in \cite{kourtellaris-charalambousIT2015_Part_1}; they are recalled, because of they are  extensively used in this paper. \\
The first theorem gives the characterization of FTFI capacity for channel distributions of Class A and transmission costs of Class A or B.  \\

 \begin{theorem}\cite{kourtellaris-charalambousIT2015_Part_1} (Characterization of FTFI capacity for  channels of class A)\\
  \label{thm-ISR}
Suppose the channel distribution is of Class A defined by   (\ref{CD_C1}).\\ 
Define the restricted class of channel input distributions $\overline{\cal P}^{A}_{[0,n]} \subset {\cal P}_{[0,n]}$ by
\begin{align}
\overline{\cal P}_{[0,n]}^{A}\tri & \Big\{ \big\{  P_i(da_i| a^{i-1}, b^{i-1}): i=0,1, \ldots, n\big\} \in  {\cal P}_{[0,n]}: \nonumber \\
&     P_i(d{a}_i|a^{i-1}, b^{i-1})=\pi_i(d{a}_i|b^{i-1})-a.a.
  (a^{i-1}, b^{i-1}),i=0,1,\ldots,n\Big\} . \label{rest_set_1}
\end{align}
The following hold.\\
(a) The maximization of $I(A^n \rar B^n)$ defined by (\ref{CIS_6_a}) over  ${\cal P}_{[0,n]}$ occurs in $\overline{\cal P}_{[0,n]}^{A} \subset {\cal P}_{[0,n]} $ and  the characterization of FTFI capacity is given by the following expression.
\begin{align}
{C}_{A^n \rar B^n}^{FB,A} 
  = \sup_{\big\{\pi_i(da_i | b^{i-1}) \in {\cal M}({\mb A}_i) : i=0,\ldots,n\big\}}
  \sum_{i=0}^{n}{\bf E}^{\pi}\Big\{\log\Big(\frac{dQ_i(\cdot| B^{i-1}, A_i)}
  {d\Pi_i^{\pi_i}(\cdot| B^{i-1})}(B_i)\Big)\Big\} \label{CIS_18}
  \end{align}
  where  the transition probabilities of $\{B_i: i=0, \ldots, n\}$ and joint distribution of $\{(A_i, B_i): i=0, \ldots, n\}$ are given by the following expressions.
\begin{align}  
  \Pi_i^{\pi}(db_i| b^{i-1}) =& \int_{ {\mb A}_i} Q_i(db_i |b^{i-1}, a_i) \otimes   {\pi}_i(da_i | b^{i-1}), \hso i=0, \ldots, n, \label{CIS_10a}\\
{\bf P}^{ \pi }(da^i, db^i) =& \otimes_{j=0}^i \Big(Q_j(db_j|b^{j-1},a_j) \otimes {\pi}_j(da_j| b^{j-1}) \Big).   \label{CIS_13_G}  
\end{align} 
(b) Suppose the following two conditions hold.
\begin{align}
(b.1)& \hso {\gamma}_i(T^ia^n, T^ib^{n})=\gamma_i^{A}(a_i, b^{i}) \hso \mbox{or} \hso {\gamma}_i(T^ia^n, T^ib^{n})=\gamma_i^{B.K}(a_i, b_{i-K}^{i}), \hso i=0, \ldots, n,    \label{CIS_19} \\
(b.2)& \hso  C_{A^n \rar B^n}^{FB, A}(\kappa) \tri \sup_{\big\{P_i(da_i| a^{i-1 }, b^{i-1}) :i=0,\ldots,n\big\}\in  {\cal P}_{[0,n]}(\kappa) } I(A^{n}
\rar {B}^{n})   \label{ISDS_6b}   \\
=& \inf_{s \geq  0} \sup_{ \big\{ P_i(da_i|a^{i-1}, b^{i-1}): i=0,\ldots, n\big\}\in {\cal P}_{[0,n]}   } \Big\{ I(A^{n}
\rar {B}^{n})  - s \Big\{ {\bf E}^P\Big(c_{0,n}(A^n,B^{n})\Big)-\kappa(n+1) \Big\}\Big\}. \label{ISDS_6cc}
\end{align}   
The maximization of $I(A^{n}\rar {B}^{n})$ defined by (\ref{CIS_6_a}) over  $\big\{P_i(da_i| a^{i-1 }, b^{i-1}) :i=0,\ldots,n\big\}\in  {\cal P}_{[0,n]}(\kappa)  $ occurs in $\overline{\cal P}_{[0,n]}^{A} \bigcap  {\cal P}_{[0,n]}(\kappa) $, that is, $\big\{P_i(da_i| a^{i-1}, b^{i-1})=\pi_i(da_i|b^{i-1})-a.a. (a^{i-1}, b^{i-1}),  i=0, \ldots, n\big\}$, and the FTFI capacity is given by the following expression.
\begin{align}
{C}_{A^n \rar B^n}^{FB,A}(\kappa)   
  = \sup_{\pi_i(da_i | b^{i-1}) \in {\cal M}({\mb A}_i), i=0,\ldots,n:  \frac{1}{n+1} {\bf E}^{\pi}\big\{c_{0,n}(A^n, B^{n})\big\} \leq \kappa  }
  \sum_{i=0}^{n}{\bf E}^{\pi}\Big\{\log\Big(\frac{dQ_i(\cdot| B^{i-1}, A_i)}
  {d\Pi_i^{\pi_i}(\cdot| B^{i-1})}(B_i)\Big)\Big\}. \label{CIS_18_TC}
  \end{align}

\end{theorem}

\ \

\begin{remark}(Equivalence of constraint and unconstraint problems)\\
The equivalence of constraint and unconstraint problems in Theorem~\ref{thm-ISR}, follows from  Lagrange's duality theory of optimizing convex functionals over convex sets  \cite{dluenberger1969}. Specifically, from \cite{charalambous-stavrou2013aa}, it follows that the set of distributions ${\bf P}^{C1}(da^n|b^{n-1})\tri \otimes_{i=0}^n P_i(da_i|a^{i-1}, b^{i-1}) \in {\cal M}({\mb A}^n)$ is convex, and this uniquely defines ${\cal P}_{[0,n]}$ and vice-versa,  directed information as a functional of  ${\bf P}^{C1}(da^n|b^{n-1}) \in {\cal M}({\mb A}^n)$ is convex, and  by the linearity the constraint set ${\cal P}_{[0,n]}(\kappa)$ expressed in ${\bf P}^{C1}(da^n|b^{n-1})$, is convex. Hence, if their exists a maximizing distribution and the so-called Slater condition holds (i.e., a sufficient condition is the existence of an interior point to the constraint set), then  the constraint and unconstraint problems are equivalent. 
\end{remark}

The next theorem gives the characterization of FTFI capacity for channel distributions of Class B, and transmission cost functions of Class A or B. \\

\begin{theorem}\cite{kourtellaris-charalambousIT2015_Part_1} 
\label{cor-ISR_C4}
(Characterization of FTFI capacity of channel class B and transmission costs of class A or B)\\
(a) Suppose the channel distribution is of Class B, that is, $Q_i(db_i|b_{i-M}^{i-1}, a_i), i=0, \ldots, n$.\\
Then   the maximization of $I(A^n \rar B^n)$ defined by (\ref{CIS_6_a})  over  ${\cal P}_{[0,n]}$ occurs in the subset 
\begin{align}
\sr{\circ}{\cal P}_{[0,n]}^{B.M} \tri     \Big\{   P_i(da_i | a^{i-1}, b^{i-1})= \pi_i(da_i|b_{i-M}^{i-1})-a.a. (a^{i-1}, b^{i-1}): i=0, 1, \ldots, n\Big\} \subset {\cal P}_{[0,n]} . \label{OCID_1}
\end{align}
 and the characterization of the FTFI feedback  capacity is given by the following expression. 
\begin{align}
{C}_{A^n \rar B^n}^{FB,B.M} 
=& \sup_{\big\{\pi_i(da_i | b_{i-M}^{i-1}) \in {\cal M}({\mb A}_i) : i=0,\ldots,n\big\}} \sum_{i=0}^n {\bf E}^{ \pi}\left\{
\log\Big(\frac{dQ_i(\cdot|B_{i-M}^{i-1},A_i)}{v_i^{ \pi}(\cdot|B_{i-M}^{i-1})}(B_i)\Big)
\right\}  \label{cor-ISR_25a_a_c_C4}\\
 \equiv &\sup_{\big\{\pi_i(da_i |b_{i-M}^{i-1}) \in {\cal M}({\mb A}_i) : i=0,\ldots,n\big\}} \sum_{i=0}^n I(A_i; B_i|B_{i-M}^{i-1})
\end{align}
where 
\begin{align}
  v_i^{\pi}(db_i | b_{i-M}^{i-1}) =& \int_{  {\mb A}_i }   Q_i(db_i |b_{i-M}^{i-1}, a_i) \otimes   {\pi}_i(da_i | b_{i-M}^{i-1}), \hso i=0, \ldots, n, \label{cor-ISR_31_c_C4}\\
  {\bf P}^{\pi}(da^i,  d b^i)=& \otimes_{j=0}^i\Big( Q_j(db_j|b_{j-M}^{j-1}, a_j) \otimes \pi_j(da_j| b_{j-M}^{j-1})\Big). \label{cor-ISR_32_c_C4} 
 \end{align}
(b) Suppose the channel distribution is of Class B (i.e., as in (a)),  a transmission cost is imposed  ${\cal P}_{0,n}(\kappa)$, corresponding to transmission cost functions  $\big\{\gamma_i^B(a_i, b_{i-K}^{i}), i=0, \ldots, n\big\},$
and the analogue of  Theorem~\ref{thm-ISR}, (b.2) holds.\\ 
The maximization of $I(A^n \rar B^n)$ defined by (\ref{CIS_6_a}) over $\big\{P_i(da_i|a^{i-1},b^{i-1}), i=0, \ldots, n\big\} \in{\cal P}_{0,n}(\kappa)$  occurs in $\sr{\circ}{\cal P}_{[0,n]}^{B.J} \bigcap {\cal P}_{[0,n]}(\kappa)$,  
and the characterization of FTFI capacity is given by the following expression.
\begin{align}
{C}_{A^n \rar B^n}^{FB,B.J}(\kappa) 
= \sup_{\pi_i(da_i | b_{i-J}^{i-1}) \in {\cal M}({\mb A}_i), i=0,\ldots,n: \frac{1}{n+1} {\bf E}^{\pi}\big\{ c_{0, n}(A^n, B^{n-1}\big\} \leq \kappa    } \sum_{i=0}^n {\bf E}^{ \pi}\left\{
\log\Big(\frac{dQ_i(\cdot|B_{i-M}^{i-1},A_i)}{d\nu_i^{{ \pi}}(\cdot|B_{i-J}^{i-1})}(B_i)\Big)
\right\}  \label{cor-ISR_B.2}
\end{align}
where $J=\max\{K, M\}$ and 
\begin{align}
{\bf P}^{\pi}(db^i, da^i) =&\otimes_{j=0}^i\Big(Q_j(db_j|b_{j-M}^{j-1}, a_j)\otimes \pi_j(da_j|b_{j-J}^{j-1})\Big),  \hso i=0, \ldots, n, \label{cor-ISR_B.2_1} \\
\nu_i^{\pi}(db_i|b_{i-J}^{i-1}) =&\int_{{\mb A}_i} Q_i(db_i|b_{i-M}^{i-1}, a_i)\otimes \pi_i(da_i|b_{i-J}^{i-1}).\label{cor-ISR_B.2_2}
\end{align}
(c)  Suppose the channel distribution is of Class B (i.e., as in (a)), and the maximization of $I(A^n \rar B^n)$ defined by (\ref{CIS_6_a}), is over channel input conditional distributions with transmission cost ${\cal P}_{0,n}(\kappa)$, corresponding to $\{\gamma_i^A(a_i, b^{i}): i=0, \ldots, n\}$, and the analogue of Theorem~\ref{thm-ISR},  (b.2) holds.\\ 
The maximization  of $I(A^n \rar B^n)$ defined by (\ref{CIS_6_a}) over $\big\{P_i(da_i|a^{i-1},b^{i-1}), i=0, \ldots, n\big\} \in{\cal P}_{[0,n]}(\kappa)$  occurs in $\overline{\cal P}_{[0,n]}^A \bigcap{\cal P}_{[0,n]}(\kappa) $.  
\end{theorem}

\section{Realization of Optimal Channel Input Distributions by  Information Lossless Randomized Functions}
\label{ILE}
In this section, alternative characterizations of FTFI capacity given in Theorem~\ref{thm-ISR}, Theorem~\ref{cor-ISR_C4},  are  obtained by  realizing the optimal channel input conditional distributions by information lossless randomized strategies, driven by uniform RVs. Application examples to Gaussian Linear Channel Models with memory are given  in Section~\ref{Con-LQG} and Section~\ref{putl-ex}.\\
 The  principle idea exploited  is  based on a lemma derived in \cite{gihman-skorohod1979}. This lemma  states that, for  any family of conditional distributions (on Polish spaces), conditioned on an information structure (i.e., parametrized by the conditioning variables),   there exist deterministic functions, measurable with respect to the conditioning information structure and an additional real-valued uniform RV taking values in  $[0,1]$, such that conditional distributions  can be replaced by the Lebesgue measure of  such  deterministic functions.  \\
The lemma is stated below.\\

\begin{lemma}[Lemma 1.2 in \cite{gihman-skorohod1979}](Realization of conditional distributions by randomized strategies)\\
\label{lemma-gs1989}
Let ${\bf P }(\cdot| w)$ be a family of measures on  Polish space $({\cal V}, {\mathbb B}({\cal V}))$, $w \in {\cal W}$, (i.e., $( {\cal W}, {\mathbb B}({\cal W})$) a measurable space). \\Let ${\mathbb B}([0,1])$ be the $\sigma-$algebra of Borel sets on $[0,1]$ and ${\bf m}(\cdot)$ the Lebesgue measure on $[0,1]$. \\
If ${\bf P}(A| w)$ is ${\mathbb B}({\cal W})-$measurable in $ w \in {\cal W}$ for all $A \in {\mathbb B}({\cal V})$, then there exists a family of functions $f: {\cal W} \times [0,1]\longmapsto {\cal V}, (w,t) \longmapsto  a\tri f(w,t)$, measurable with respect to ${\mathbb B}({\cal W}) \otimes {\mathbb B}([0,1])$ such that 
\bea
{\bf m}\Big\{t\in [0,1]: f(w, t) \in A\Big\}= {\bf P}(A|w), \hso \forall A \in {\mathbb B}({\cal V}) .  \label{df}
\eea
\end{lemma}

Since Lemma~\ref{lemma-gs1989}  holds for general, complete separable metric spaces ${\cal V}, {\cal W}$,  it also  holds for  arbitrary alphabets, such as, continuous, countable, finite etc. Note that the function $f(w, \cdot)$ in Lemma~\ref{lemma-gs1989} is a randomization with respect to a uniform RV taking values in $[0,1]$, and that arbitrary distributed RVs are generated by uniform RVs. 

\subsection{Recursive Nonlinear Channel Models}
Without loss of generality, the material of this section are developed for nonlinear  channel models, which induce channel distributions of Class A, B, as defined  below.\\
  
\begin{definition}(Nonlinear channel models and transmission costs) \\
\label{exa_A_D}
(a) {NCM-A} Nonlinear Channel Models A are defined by   nonlinear recursive models and   transmission cost functions, as follows.
\begin{align}
B_i   =& h_i^{A}(B^{i-1}, A_i, V_i),  \hso B^{-1}=b^{-1}, \hso   i=0, \ldots, n, \hst \frac{1}{n+1}\sum_{i=0}^n {\bf E} \big\{ \gamma_i^{A}(A_i, B^{i}) \big\}  \leq \kappa \label{NCM-A_D}
\end{align} 
where $\{V_i: i=0,1, \ldots, n\}$ is the noise process. The underlying assumptions are the following.\\
 {\bf Assumption A.(i)}  The alphabet spaces include any of the following.
\begin{align}
&\mbox{(a) Continuous Alphabets:} \;
{\mb B_i}\tri {\mb R}^p, \; {\mb A_i}\tri {\mb R}^q, \;  {\mb V}_i \tri {\mb R}^r, \hso i=0, 1, \ldots, n; \label{CA_A.D} \\
&\mbox{(b) Finite Alphabets:} \;
{\mb B_i}\tri \big\{1, \ldots, p\big\}, \; {\mb A_i}\tri \big\{1, \ldots, q\},  \; {\mb V}_i \tri \big\{1, \ldots, r\big\}, \; i=0, 1, \ldots, n;  \label{DA_A.D} \\
&\mbox{(c) Combinations of Continuous and Discrete (Finite or Countable) Alphabets.}
\end{align}
The above simply illustrates  that no assumption is imposed on the alphabet spaces, and (a)-(c) are specific cases.\\
 {\bf  Assumption A.(ii)} $h_i^{A}: {\mb B}^{i-1} \times {\mb A}_i  \times {\mb V_i} \longmapsto {\mb B}_i, \gamma_i^{A}:  {\mathbb A}_i\times  {\mb B}^{i} \longmapsto {\mb A}_i$  and $h_i^A(\cdot, \cdot, \cdot), \gamma_i^A(\cdot, \cdot)$ are  measurable functions, for $i=0, 1, \ldots, n$; \\
{\bf Assumption A.(iii).}  The noise process $\{V_i: i=0, \ldots, n\}$  satisfies 
\bea
{\bf P}_{V_i|V^{i-1}, A^i}(dv_i|v^{i-1}, a^i)=  {\bf P}_{V_i}(dv_i)-a.a.(v^{i-1}, a^i),\hso i=0, \ldots, n. \label{CI_A}
\eea
By (\ref{CI_A}),  the following  consistency condition holds.
\begin{align}
{\mb P}\Big\{B_i \in \Gamma \Big| B^{i-1}=b^{i-1}, A^i=a^i\Big\}=\; & {\bf P}_{V_i}\Big(V_i: h_i^{A}(b^{i-1}, a_i, V_i) \in \Gamma \Big), \hso \Gamma \in {\cal B}({\mb B}_i)\label{SCNCMD-A.D} \\
=\; & Q_i(\Gamma |b^{i-1}, a_i),  \hso  i=0,1, \ldots, n.   \label{NCM-A.D_CD}
\end{align}
Model (\ref{NCM-A_D}) induces a conditional channel distribution of  Class A. 
The convention is that  transmission  starts at time $i=0$, and  the  initial data $B^{-1}=b^{-1}\equiv b_{-\infty}^{-1}$ are either specified or their distribution is fixed. There is no loss of generality to assume the above convention, because all material also hold if the following alternative convention is considered.  $B_0=h_0(B^{-1},A_0, V_0)\equiv h_0(A_0,V_0), \gamma_0^A(A_0, B^0) \equiv \gamma_0^A(A_0), B_1=h_1(B^0, A_1, V_1)\equiv h_1(B_0, A_1, V_1), \gamma_1^A(A_1, B^1)=\gamma_1(A_1, B_0, B_1),  \ldots, n,$ and $B^{-1}=0$, that is, no information is available prior to transmission.    

(b) {NCM-B} Nonlinear Channel Models B are defined as follows.  
\begin{align}
B_i   = h_i^{B.M}(B_{i-M}^{i-1}, A_{i}, V_i),  \hso B_{-M}^{-1}=b_{-M}^{-1}, \hso   i=0, \ldots, n, \hst \frac{1}{n+1}\sum_{i=0}^n {\bf E} \big\{ \gamma_i^{B.K}(A_i, B_{i-K}^{i}) \big\}  \leq \kappa \label{NCM-C_D_B}
\end{align} 
where $\{V_i: i=0,1, \ldots, n\}$ is the noise process. 
The underlying assumptions are the following.    \\
 {\bf Assumption B.(i)}  Assumptions A.(i)-A.(iii) hold with appropriate changes.\\
By (\ref{CI_A}),   the following  consistency condition holds.
\begin{align}
{\mb P}\Big\{B_i \in \Gamma \Big| B^{i-1}=b^{i-1}, A^i=a^i\Big\}=\; & {\bf P}_{V_i}\Big(V_i: h_i^{B.M}(b_{i-M}^{i-1}, a_i, V_i) \in \Gamma \Big), \hso \Gamma \in {\cal B}({\mb B}_i)\label{SCNCMD-A.D_B} \\
=\; & Q_i(\Gamma |b_{i-M}^{i-1}, a_i),  \hso  i=0,1, \ldots, n.   \label{NCM-A.D_CD_B}
\end{align}
\end{definition}

It is not necessary to introduced additional NCMs which are combinations of channels of Class A or B and transmission costs of Class A or B, because these are included in the above models.


\subsection{Alternative Characterization of FTFI Capacity for NCM-A }
\label{exa_A.1}
Consider the NCM-A given by (\ref{NCM-A_D}) (i.e.,  Definition~\ref{exa_A_D}, (a)).  By invoking Lemma~\ref{lemma-gs1989}, and a property called information lossless randomized strategies, an alternative characterization of the FTFI capacity   given in Theorem~\ref{thm-ISR}, (b),  
is   obtained, as stated in the next theorem. \\

\begin{theorem}(Characterization of FTFI capacity for NCM-A by information lossless randomized strategies)\\
\label{thm-AC-1}
Consider the characterization of FTFI capacity, $C_{A^n \rar B^n}^{FB, A}(\kappa)$, given in Theorem~\ref{thm-ISR}, (b), for the NCM-A of Definition~\ref{exa_A_D}, (a).  \\
Then the following hold.\\
(a) The consistency conditions CON.A.(1), (2) stated below hold.\\
{ CON.A.(1).} there exist functions  $e_i^{A}(\cdot, \cdot)$ measurable with respect to the information structure ${\cal I}_i^{e^{A}} \tri \{b^{i-1}, u_i\},  i=0,1, \ldots, n$  and defined by 
\begin{align}
e_i^{A}: {\mathbb B}^{i-1} \times {\mathbb U}_i \longmapsto {\mathbb A}_i,  \hso  {\mb U}_i \tri [0,1], \hst   a_i= e_i^{A}(b^{i-1}, u_i), \hso i=0,1, \ldots, n \label{UN_IL}
\end{align}
 such that $\big\{U_i: i=0, \ldots, n\big\}$ are uniform RVs on $[0,1]^{n+1}$ and 
\begin{align}
{\bf P}_{A_i| B^{i-1}}(da_i| b^{i-1})= {\bf P}_{U_i} \Big( U_i : e_i^{A}(b^{i-1}, U_i)  \in da_i \Big), \hso i=0,1, \ldots, n.   \label{LCM-A.1_5_new}
 \end{align}
{ CON.A.(2).} i) $A_i$ is conditionally independent  of   $A^{i-1}$ given $B^{i-1}$, for each $i=0,1, \ldots, n$,  ii) $U_i$ is independent of $\big(U^{i-1}, V^{i-1}\big), i=0,1, \ldots, n$, and iii) $V_i$ is independent of $\big(V^{i-1}, U^{i}\big), i=0,1, \ldots, n$.

(b) The set of all channel input distribution $\overline{\cal P}_{[0,n]}^{A}$ defined by (\ref{rest_set_1})  is realized by  strategies $\{e_i^{A}(\cdot, \cdot): i=0,1, \ldots, n\}$, and the following hold\footnote{Superscript notation ``${\bf E}^{e^{A}}\{\cdot \}$'' indicates the dependence of the joint distribution on the  strategy $\{e_i^{A}(\cdot, \cdot):i=0,\ldots, n\} $.}. 
\begin{align}
&A_i =e_i^{A}(B^{i-1}, U_i), \hso i=0,1, \ldots, n,   \label{LCMNonG-A.1_6a_new}    \\
&B_i   = h_i^{A}(B^{i-1}, A_i, V_i), \hso B^{-1}=b^{-1}, \hso i=0,1, \ldots, n,    \label{LCMNon-A.1_6b_new}\\
&\frac{1}{n+1}\sum_{i=0}^n {\bf E}^{e^{A}} \Big\{\gamma_i^{A}(e_i^{A}(B^{i-1}, U_i), B^{i}) \Big\}   \leq \kappa. \label{LCMNon-A.1_6c_new}
\end{align} 
(c) Define the restricted class of randomized strategies called information lossless randomized strategies, as follows.  
\begin{align}
&{\cal E}_{[0,n]}^{IL-A} (\kappa) \tri  \Big\{ e_i^{A}(b^{i-1}, u_i) \:\: \mbox{defined by (\ref{UN_IL})},\:  \:  \mbox{and for a fixed $b^{i-1}$,  the function } \;  \: e_i^{A}(b^{i-1}, \cdot)  \nonumber \\
& \mbox{is one-to-one and onto ${\mb A}_i$, for $i=0,\ldots, n$}: \hso  \frac{1}{n+1}\sum_{i=0}^n {\bf E}^{e^{A}} \Big\{\gamma_i^{A}(e_i^{A}(B^{i-1}, U_i), B^{i})\Big\}    \leq \kappa         \Big\}. \label{alt_A.1_NCM}
\end{align}
Then an  alternative equivalent characterization of FTFI capacity ${C}_{A^n \rar B^n}^{FB,A}(\kappa)$ defined by (\ref{CIS_18}), is given by the following expression.
\begin{align}
&{C}_{A^n \rar B^n}^{FB,A}(\kappa)={C}_{A^n \rar B^n}^{FB,IL-A} (\kappa) \tri  \sup_{ {\bf P}_{U^n},   \big\{e_i^{A}(b^{i-1}, u_i): i=0, \ldots, n\big\} \in {\cal E}_{[0,n]}^{IL-A}(\kappa)        }  \sum_{i=0}^n   {\bf E}^{e^{A}}\Big\{\log\Big(\frac{Q_i(\cdot| B^{i-1}, e_i^{A}(B^{i-1}, U_i))}
  {\Pi_i^{e^{A}}(\cdot| B^{i-1})}(B_i)\Big)\Big\}  \label{CMNon-A.1_10_new} \\
&\hst \hst \hst \hst \hst \hst \hst \hst \hso \equiv    \sup_{ {\bf P}_{U^n},    \big\{e_i^{A}(b^{i-1}, u_i): i=0, \ldots, n\big\} \in {\cal E}_{[0,n]}^{IL-A}(\kappa)        }      \sum_{i=0}^n I^{e^{A}}(U_i; B_i|B^{i-1}), \label{CMNon-A.1_10_new_o} \\
&{\Pi }^{e^{A}}(db_i|b^{i-1}) 
=\int_{ {\mb U}_i}  Q_i(db_i| b^{i-1}, e_i^{A}(b^{i-1}, U_i))\otimes {\bf P}_{U_i}(du_i). \label{LCMNon-A.1_4_new} 
\end{align}
\end{theorem}
\begin{proof} See Appendix~\ref{appendix_thm-AC-1}. 
\end{proof}

\ \

\begin{remark}(Comments on Theorem~\ref{thm-AC-1})\\
\label{rem_AC_1}
Given a specific NCM-A, it can be shown that the  maximization over ${\bf P}_{U^n}$ in (\ref{CMNon-A.1_10_new_o}) is not required, because, for a fixed $B^{i-1}=b^{i-1}$,  the optimal channel input distribution $\pi_i^*(da_i|b^{i-1})$ can be generated via proper choice of the function $e_i^{A}(b^{i-1}, \cdot)$, as a composition of two functions, $e_i^{A}(b^{i-1}, \cdot)=\overline{e}_i^{A}(b^{i-1}, g_i(\cdot)), g_i: [0,1] \longmapsto {\mathbb  Z}_i, z_i=g_i(u_i)$, where $\big\{Z_i=g_i(U_i): i=0, \ldots, n\big\}$ is a specific random process, i.e., its  distribution is specific and depends on the channel distribution,   and that this composition of functions  induces the  optimal conditional channel input distribution $\pi_i^*(da_i|b^{i-1})$, for $i=0, \ldots, n$.  For example, if the channel distribution is memoryless, i.e.,  $Q_i(db_i|b^{i-1},a_i)=Q_i(db_i|a_i)$ and $\gamma_i^A(a_i, b^{i-1})=\gamma_i(a_i)$, for $i=0, \ldots, n$, and the distribution,  which maximizes the characterization of FTFI capacity is ${\mb P}\big\{A_i \leq a_i\big\}\tri F_{A_i}^*(a_i), i=0, \ldots, n$, then $a_i=e_i(u_i)$ and the optimal functions in (\ref{CMNon-A.1_10_new_o}) are given by $e_i^{*}(u_i)={ F}_{A_i}^{*,-1}(u_i), i=0, \ldots, n$.  This is due to the fact an arbitrary distributed RVs can be generated from uniform RVs.
In general, the maximization  in (\ref{CMNon-A.1_10_new_o}) can be solved using  dynamic programming \cite{kumar-varayia1986,vanschuppen2010}.
\end{remark}

\subsection{Alternative Characterization of FTFI Capacity for NCM-B }
\label{exa_A.1_B.1}
Consider the NCM-B defined by (\ref{NCM-C_D_B}) (i.e.,   Definition~\ref{exa_A_D}, (b)). 
By   Theorem~\ref{cor-ISR_C4}, (c), the corresponding optimal channel input distribution are of the form 
$\big\{ \pi_i(da_i| b_{i-J}^{i-1}):  i=0, 1, \ldots, n\big\}$,   $ J \tri \max\{M,K\}$. Clearly, all the material of Section~\ref{exa_A.1} apply to NCM-B.  The analog of Theorem~\ref{thm-AC-1} is stated for future reference. \\

\begin{theorem}(Characterization of FTFI capacity for NCM-B by information lossless randomized strategies)\\
\label{thm-AC-2}
Consider the characterization of FTFI capacity, $C_{A^n \rar B^n}^{FB, B.J}(\kappa)$, given in Theorem~\ref{cor-ISR_C4}, (c), for the NCM-B of Definition~\ref{exa_A_D}, (b).  \\
Then the following hold.\\
(a) The consistency conditions CON.B.(1), (2) stated below hold.\\
\noi { CON.B.(1).} There exists a function $e_i^{B.J}(\cdot), J\tri \max\{M, K\}$ measurable with respect to the information structure ${\cal I}_i^{e^{B.J}} \tri \{b_{i-J}^{i-1}, u_i\},  i=0,1, \ldots, n$  defined by
\bea
e_i^{B.J}: {\mathbb B}_{i-J}^{i-1}\times \times  {\mathbb U}_i \longmapsto {\mathbb A}_i, \hso  {\mb U}_i \tri [0,1], \hso  a_i= e_i^{B.J}(b_{i-J}^{i-1}, u_i),  i=0,1, \ldots, n \label{CI_B}
\eea
where $\{U_i: i=0,1, \ldots, n\}$ are  uniform distributed on  $[0,1]^{n+1}$  such that 
\begin{align}
 {\bf P}_{A_i| B_{i-J}^{i-1}}(da_i| b_{i-J}^{i-1})={\bf P}_{U_i} \Big( U_i : e_i^{B.J}(b_{i-J}^{i-1}, U_i)  \in da_i \Big), \hso i=0,1, \ldots, n, \hso J\tri \max\{M, K\}. \label{NCM-C.1_10_B}
 \end{align}
\noi { CON.B.(2).} i) $A_i$ is conditionally independent of $\big\{A^{i-1}, B^{i-J-1}\big\}$ given $\big\{B_{i-J}^{i-1}\big\}$ for $i=0, \ldots, n$, ii) $U_i$ is independent of $\Big(U^{i-1}, V^{i-1}\Big), i=0, \ldots, n$, iii) $V_i$ is independent of $\Big(V^{i-1}, U^i\Big)$, $i=0, \ldots, n$.

(b) The set of all channel input distribution $\sr{\circ}{\cal P}_{[0,n]}^{B.J}$ defined by (\ref{OCID_1})   is realized by  strategies $\{e_i^{B.J}(\cdot, \cdot): i=0,1, \ldots, n\}$, and the following hold. 
\begin{align}
&A_i =e_i^{B.J}(B_{i-J}^{i-1}, U_i), \hso i=0,1, \ldots, n,   \label{LCMNonG-C.1_6a_B}    \\
&B_i   =h_i^{B.M}(B_{i-M}^{i-1}, e_i^{B.J}(B_{i-J}^{i-1}, U_i), V_i),\hso B_{-M}^{-1}=b_{-M}^{-1}, \hso i=0,1, \ldots, n,    \label{NCM-B.1_11_B}\\
&\frac{1}{n+1}\sum_{i=0}^n {\bf E}^{e^{B.J}} \Big\{\gamma_i^{B.K}(e_i^{B.J}(B_{i-J}^{i-1}, U_i), B_{i-K}^{i})\Big\}    \leq \kappa. \label{NCM-C.1_12_B}
\end{align} 

(c) The restricted class of randomized strategies, defined by  
\begin{align}
&{\cal E}_{[0,n]}^{IL-B.J} (\kappa) \tri  \Big\{ e_i^{B.J}(b_{i-J}^{i-1}, u_i) \hso \mbox{defined by (\ref{CI_B}) and for a fixed $b_{i-J}^{i-1},$ the map } \;  \: e_i^{B.J}(b_{i-J}^{i-1}, \cdot)  \nonumber \\
& \mbox{is one-to-one and onto ${\mb A}_i$ for $i=0,\ldots, n$}: \frac{1}{n+1}\sum_{i=0}^n {\bf E}^{e^{B.J}} \Big(\gamma_i^{B.K}(e_i^{B.J}(B_{i-J}^{i-1}, U_i), B_{i-K}^{i})\Big)    \leq \kappa         \Big\}. \label{NCM-C.1_13_B}
\end{align}
is information lossless,  and an  alternative characterization of FTFI capacity $C_{A^n \rar B^n}^{FB, B.J}(\kappa)$,  is  given by the following expression.
\begin{align}
&{C}_{A^n \rar B^n}^{FB,B.J}(\kappa)={C}_{A^n \rar B^n}^{FB,IL-B.J} (\kappa) \tri  \sup_{ {\bf P}_{U^n},   \big\{e_i^{B.J}(b_{i-J}^{i-1}, u_i): i=0, \ldots, n\big\} \in {\cal E}_{[0,n]}^{IL-B.J}(\kappa)        } \Bigg\{\nonumber \\
& \hst \hst \hst \hst \hst \hst \hst \hst \hst \hst
 \sum_{i=0}^n   {\bf E}^{e^{B.J}}\Big\{\log\Big(\frac{dQ_i(\cdot| B_{i-M}^{i-1}, e_i^{B.J}(B_{i-J}^{i-1}, U_i))}
  {d\nu_i^{e^{B.J}}(\cdot| B_{i-J}^{i-1})}(B_{i-1})\Big)\Big\} \Bigg\}  \label{NCM-C.1_14_B}   \\
& \hst \hst \hst \hst \hst \hst \hst \hst \hst  \equiv    \sup_{{\bf P}_{U^n}, \big\{e_i^{B.J}(b_{i-J}^{i-1},u_i): i=0, \ldots, n\big\} \in {\cal E}_{[0,n]}^{IL-B.J}(\kappa)        }    \sum_{i=0}^n I^{e^{B.J}}(U_i; B_i|B_{i-J}^{i-1}) \\
&\nu_i^{e^{B.J}}(db_i|b_{i-J}^{i-1})=  \int_{{\mathbb U}_i} Q_i(dB_i| B_{i-M}^{i-1}, e_i^{B.J}(b_{i-J}^{i-1},  u_i)) \otimes  {\bf P}_{U_i}(du_i),  \hso  i=0,1, \ldots, n. \label{LCMNon-C.1_4_B} 
\end{align}
\end{theorem}
\begin{proof} (a). This is obtained by  utilizing the information structure of the optimal channel input distribution $\{ \pi_i(da_i|b_{i-J}^{i-1})\equiv {\bf P}_{A_i| B_{i-J}^{i-1}}(a_i| b_{i-J}^{i-1}) : i=0,1, \ldots\}$, and   Lemma~\ref{lemma-gs1989}.\\
(b) This follows from (a).\\
(c)  By Theorem~~\ref{cor-ISR_C4}, Part C,   $I(A^i; B_i|B^{i-1})=I(A_i; B_i|B_{i-J}^{i-1})={\bf E}^{\pi}\Big\{\log\Big(\frac{dQ_i(\cdot| B_{i-M}^{i-1}, A_i)}
  {d\nu_i^{\pi^M}(\cdot| B_{i-J}^{i-1})}(B_i)\Big)\Big\}, i=0, \ldots, n$. By an application of Theorem 3.7.1 in Pinsker~\cite{pinsker1964} and Corollary following it,  the  following sequence of identities hold.
\begin{align}
& \mbox{For}  \hso i=0, \ldots, n, \hso I(A_i; B_i|B_{i-J}^{i-1})=I(A_i, U_i; B_i|B_{i-J}^{i-1})
=I(U_i;B_i|B_{i-J}^{i-1}) \nonumber \\
&=  {\bf E}^{e^{B.J}}\Big\{\log\Big(\frac{Q_i(dB_i| B_{i-M}^{i-1}, e_i^{B}(B_{i-J}^{i-1}, U_i))}
  {\nu_i^{e^{B.J}}(dB_i| B_{i-J}^{i-1})}\Big)\Big\} \hso \mbox{if and only if} \hso \{e_i^{B.J}(\cdot, \cdot): i=0, \ldots, n\} \in {\cal E}_{[0,n]}^{IL-B.J} (\kappa).
\end{align}
The alternative characterization of FTFI capacity is obtained by utilizing the strategies ${\cal E}_{[0,n]}^{IL-B.J} (\kappa)$.
\end{proof}

\ \

\begin{remark}(Alternative characterizations)\\
The main point to be made is that the alternative characterizations can be used to transformed the characterizations of FTFI capacity, which are extremum problems with  respect to  channel input conditional distributions  to equivalent characterizations, which are extremum problems over  deterministic functions driven by uniform RVs. The connection to  uniform RVs can be further exploited in the context of transmitting information over the channel at a rate below the per unit time limiting version of the characterizations of FTFI capacity (when it corresponds to feedback capacity). However, this direction is not pursued further    in this paper.  
\end{remark}

\section{Characterizations of FTFI Capacity and Feedback Capacity of Gaussian LCMs \& The LQG Theory}   
\label{Con-LQG}
In this section, the characterizations of FTFI capacity given in Section~\ref{ILE} are applied to Gaussian Linear  Channel Models (G-LCMs) (special cases of NCM-A, NCM-B of Definition~\ref{exa_A_D}), to obtain the following. 
\begin{description}
\item[(a)] Characterizations of FTFI capacity for Multiple Input Multiple Output (MIMO)  G-LCMs;

\item[(b)] characterizations of FTFI capacity for MIMO G-LCMs via connections to finite horizon Linear-Quadratic-Gaussian (LQG) stochastic optimal control theory, Riccati difference matrix equations,  and water filling solutions of MIMO channels;

\item[(c)] unfold a dual role of the randomized strategies, which realize optimal channel channel input processes corresponding to the characterizations of FTFI capacity, to control the channel output process and to transmit new information over the channel.  
\end{description}

The characterizations of feedback capacity and its connections to the infinite horizon LQG stochastic optimal control theory and stability theory of linear control systems is treated in Section~\ref{putl-ex}, by investigating  per unit time limiting versions of the results obtained in this section.

\subsection{Characterizations of FTFI Capacity for Gaussian Linear Channel Models A} 
\label{exa_CG-NCM-A.1}
Consider a Gaussian Linear Channel Model A (G-LCM-A) (i.e., a special case of the {NCM-A} given by (\ref{NCM-A_D})), and  
defined as follows.
\begin{align}
&B_i   = \sum_{j=0}^{i-1} C_{i,j} B_{j} +D_{i,i} \; A_i + V_{i}, \hso B_0=D_{0,0}A_0+V_0, \hso i= 1, \ldots, n,   \label{LCM-A.1} \\
&\frac{1}{n+1} \sum_{i=0}^n {\bf E} \Big\{ \langle A_i, R_{i,i} A_i \rangle + \langle B^{i-1}, Q_{i}(i-1) B^{i-1}\rangle  \Big\} \leq \kappa, \\
&C_{i,j}\in {\mb R}^{ p \times  p},\hso D_{i, i} \in {\mb R}^{q \times q},  \hso R_{i,i}\in {\mb S}_{++}^{q \times q}, \hso  Q_0(-1)=0, \: Q_{i}(i-1) \in {\mb S}_+^{i p \times i p}, \hso i=0, \ldots, n, \: j=0, \ldots, i-1 \label{LCM-A.1_1_a}
\end{align} 
where $B^{i}=(B_0, B_1, \ldots, B_i)$,  at time $i=0$, $A_0$ does not use feedback, and   the following assumption holds.

{\bf Assumption A.1.(i).} 1) Definition~\ref{exa_A_D}, Assumption A.(ii), A.(iii) hold, and 2) the noise process $\{V_i: i=0, \ldots, n\}$  is Gaussian distributed, specified by 
\begin{align}
V_i \sim N( 0, K_{V_i}), \hso i.e., \hso  \mu_{V_i} \tri  {\bf E}\big\{V_i\big\}=0,\hso K_{V_i} \tri Cov (V_i, V_i) = {\bf E} \big\{V_i V_i^T\big\}, \hso i=0,1, \ldots, n. \label{Gau_Ind}
\end{align}

The following theorem states that the optimal channel input distribution is Gaussian, and it is realized by information lossless Gaussian randomized strategies, which are expressed via the  decomposition  $A_i= g_i(B^{i-1})+ Z_i$, in which  $g_i(B^{i-1})\perp Z_i,  i=0,     \ldots, n$,  $\{g_i(\cdot): i=0, \ldots, n\}$ is a deterministic function of the feedback information process,  and $\{Z_i: i=0, \ldots, n\}$  is an orthogonal innovations process.\\

\begin{theorem}(Characterization  of FTFI capacity for G-LCM-A)\\
\label{G-LCM-A-CA}
Consider the G-LCM-A defined by (\ref{LCM-A.1})-(\ref{LCM-A.1_1_a}), and suppose  Assumption A.1.(i) holds. Let $\{(A_i^g, B_i^g):i =0, \ldots, n\}$ denote a jointly Gaussian process satisfying (\ref{Gau_Ind}). \\
 Then the following hold.\\
(a) The optimal channel input distribution $\{\pi(da_i| b^{i-1})\equiv \pi^g(da_i| b^{i-1}):i=0, \ldots, n\}$ is Gaussian and the  characterization of FTFI Feedback Capacity is given by the following expression.
\begin{align}
{C}_{A^n \rar B^n}^{FB,G-LCM-A} (\kappa) 
 \tri &  \sup_{{\cal P}_{[0,n]}^{G-LCM-A}(\kappa)  } H(B^{g,n})- H(V^n)    \label{CG-LCM_A.1_3}
\end{align}
where
\bea
{\cal P}_{[0,n]}^{G-LCM-A}(\kappa) \tri \Big\{{\pi}_i^{g}(da_i | b^{i-1}),  i=0,\ldots,n: \frac{1}{n+1} \sum_{i=0}^n  {\bf E}^{\pi^g} \Big( \langle A_i^g, R_{i,i} A_i^g \rangle + \langle B^{g,i-1}, Q_{i}(i-1) B^{g,i-1}\rangle  \Big)\leq \kappa  \Big\}.
\eea
(b) The alternative equivalent characterization of the FTFI capacity is given by the following expressions. 
\begin{align} 
&{C}_{A^n \rar B^n}^{FB,G-LCM-A} (\kappa) 
 \tri   \sup_{\big\{ \big(\Gamma_{i}(i-1), K_{Z_i}\big),  i=0,\ldots,n: \frac{1}{n+1} \sum_{i=0}^n  {\bf E} \big( \langle A_i^g, R_{i,i} A_i^g \rangle + \langle B^{g,i-1}, Q_{i}(i-1) B^{g,i-1}\rangle  \big)\leq \kappa   \big\}    } H(B^{g,n})- H(V^n),\label{EXTR_A}  \\
& H(B^{g, n})-H(V^n)=\sum_{i=0}^n H(B_i^g|B^{g,i-1})-H(V^n)= \frac{1}{2} \sum_{i=0}^n  \log \frac{ | D_{i,i} K_{Z_i} D_{i,i}^T +K_{V_i}|}{|K_{V_i}|}, \label{G-lCM-A-CR} \\
 &A_i^g= \sum_{j=0}^{i-1}\Gamma_{i,j} B_j^g + Z_i,\hst  i=0,1, \ldots, n, \label{LCM-A.1_7} \\
 & \hst \equiv \Gamma_{i}(i-1) B^{g,i-1} + Z_i, \label{LCM-A.1_7_a}   \\
 &B_i^g   =   \sum_{j=0}^{i-1}   C_{i,j} B_{j}^g + D_{i,i} A_i^g   +V_{i}   =    \sum_{j=0}^{i-1} \Big( C_{i,j} + D_{i,i} \Gamma_{i,j}  \Big) B_{j}^g + D_{i,i} Z_i + V_{i}, \label{LCM-A.1_8}\\
 & \hst \equiv \Big(C_{i}(i-1) + D_{i,i} \Gamma_{i}(i-1) \Big) B^{g,i-1}+ D_{i,i} Z_i + V_i,  \\
 &i) \hso Z_i  \:\:  \mbox{is independent of} \; (A^{g,i-1},B^{g,i-1}), \: i=0, \ldots, n, \hso ii) \hso  Z^i \; \mbox{is independent of} \hso  V^i, \hso i=0, \ldots, n, \\
 & iii) \hso  \Big\{Z_i \sim N(0, K_{Z_i}): i=0,1, \ldots, n\Big\} \hso  \mbox{is an orthogonal innovations or independent Gaussian process}.
\end{align}
\end{theorem}

\begin{proof} The derivation is based on the maximum entropy property of Gaussian distribution, as in  Cover and Pombra \cite{cover-pombra1989}, with some variations to account for the difference of the Model considered, and  the decomposition (\ref{LCM-A.1_7_a}) expressed in terms of an orthogonal process $\{Z_i: i=0, \ldots, n\}$.  The details are given in Appendix~\ref{appendix_G-LCM-A-CA}.
\end{proof}

%
%
%
 
\begin{remark}(Extremum solution of the G-LCM-A)\\
\label{rem-sep_A}
(a) To establish     the connection of decomposition (\ref{LCM-A.1_7}) to the Cover and Pombra \cite{cover-pombra1989} realization of Gaussian channel input distributions in the characterization given by (\ref{cp1989}),  iterate (\ref{LCM-A.1_7}) by invoking  the corresponding channel output process (\ref{LCM-A.1_8}), to express the process $\{A_i^g: i=0, \ldots, n\}$ in terms of the channel noise process $\{V_i: i=0, \ldots, n\}$ and linear combinations of the process $\{Z_i: i=0, \ldots, n\}$, as follows.
\begin{align}
A^{g,n} = \overline{\Gamma}^n V^n + \overline{Z}^n, \hso \Big\{\overline{Z}_i: i=0, \ldots, n\Big\} \hso \mbox{Gaussian and Correlated}
\end{align}
where $\overline{\Gamma}^n$ is a lower diagonal matrix with time-varying deterministic entries, and $\overline{Z}^n$ is Gaussian processes $N(0, K_{\overline{Z}^n})$, and $V^n \perp \overline{Z}^n$. However, for such an equivalent realization, it is very difficult to optimize the corresponding characterization of FTFT capacity given by (\ref{EXTR_A}), even in the special case, $Q_i(i-1)=0, i=0, \ldots, n$, because the process $\Big\{\overline{Z}_i: i=0, \ldots, n\Big\}$ is not an orthogonal innovations process. Any past attempts to solve the Cover and Pombra \cite{cover-pombra1989}, characterization given by (\ref{cp1989}), for any $n$, that is, corresponding to the nonstationary nonergodic case, have been unsuccessful. Previous attempts are extensively elaborated in \cite{kim2010}.   \\
(b) Although, at first glance, the problem of determining the optimal matrices $\{\Gamma_{i}^*(i-1), K_{Z_i}^*\big),  i=0,\ldots,n\}$, which correspond to the extremum problem (\ref{EXTR_A}), appears difficult, even in special cases,
one possible re-formulation, 
is to  compactly representing (\ref{EXTR_A}), as follows.\\
 From (\ref{LCM-A.1_7}), (\ref{LCM-A.1_8}), it is always possible to find lower diagonal matrices  $\{ (C_{[i,i]}, \Gamma_{[i,i]}): i=0, \ldots, n\}$ and matrix $\{ D_{[i,i]}: i=0, \ldots, n\}$, such that the following hold.
\begin{align}
&A^{g,i}= \Gamma_{[i,i]}B^{g,i} + Z^i, \hso i=0, \ldots, n, \\
&B^{g,i}= C_{[i,i]}B^{g,i}+ D_{[i,i]} A^i + V^i, \hso i=0, \ldots, n.
\end{align}
From  the above  expression,  the covariance of the channel output process is given, as follows. 
\begin{align}
K_{B^{i-1}} \tri & {\bf E} \Big\{ B^{g, i-1}  \big(B^{g, i-1}\big)^T \Big\}, \hso i=0,1, \ldots, n, \\
 =& \Big(I-C_{[i-1,i-1]}-D_{[i-1,i-1]}\Gamma_{[i-1,i-1]}\Big)^{-1}D_{[i-1,i-1]}\Big( K_{Z^{i-1}} + K_{V^{i-1}}\Big) D_{[i-1,i-1]}^T \nonumber \\
& \hst \hst \hst  \Big(I-C_{[i-1,i-1]}-D_{[i-1,i-1]}\Gamma_{[i-1,i-1]}\Big)^{-1, T}, \hso spec\Big(C_{[i-1,i-1]}+D_{[i-1,i-1]}\Gamma_{[i-1,i-1]}\Big)<1. \label{COV_A}
\end{align}
The condition $spec\Big(C_{[i-1,i-1]}+D_{[i-1,i-1]}\Gamma_{[i-1,i-1]}\Big)<1, i=0, \ldots, n$ is equivalent  to the  existence of  a sequence  $\{\Gamma_{i,j}: i=0, \ldots, n, j=0, \ldots, i-1\}$, which ensures  the eigenvalues of the channel output process lie in the open unit disc in the space of complex numbers ${\mb C}$. \\
Utilizing the above representations, the average transmission cost constraint   is given by 
\begin{align}
{\cal P}_{[0,n]}^{G-LCM-A}(\kappa) \tri &\Big\{\Big(\Gamma_{i}(i-1), K_{Z_i}\Big), i=0, \ldots, n: \sum_{i=0}^n {\bf E} \Big( \langle A_i^g, R_{i,i} A_i^g\rangle + \langle B^{g,i-1}, Q_{i}(i-1) B^{g,i-1}\rangle\Big) \nonumber \\
&=\sum_{i=0}^n tr\Big(R_{i,i} \Gamma_{i}(i-1) K_{B^{i-1}} \Gamma_{i}^T(i-1) + R_{i,i} K_{Z_i}+ Q_{i}(i-1)K_{B^{i-1}} \Big) \leq \kappa \Big\}.
\end{align}

Hence, the FTFI capacity is characterized by 
\begin{align}
{C}_{A^n \rar B^n}^{FB,G-LCM-A} (\kappa)=&\max_{  \big\{\big(\Gamma_{i-1}(i), K_{Z_i}\big), i=0, \ldots, n\big\} \in {\cal P}_{[0,n]}^{G-LCM-A}(\kappa) \hso \mbox{and (\ref{COV_A}) holds}  }  \frac{1}{2} \sum_{i=0}^n\log \frac{ | D_{i,i} K_{Z_i} D_{i,i}^T+K_{V_i}|}{|K_{V_i}|} \label{LCM-A.1_10_LCM_1}
\end{align}
Extremum problem  (\ref{LCM-A.1_10_LCM_1}) is a deterministic optimization problem. However,  although  compactly represented and attractive, it is not at all easy to optimize, because the functional dependence of $\{K_{B^{i-1}}: i=0,1, \ldots, n\}$ on $\{\Gamma_{i}(i-1), K_{Z_i}: i=0, \ldots, n\}$, is  very complex. 
Hence, this re-formulation  is not pursued any further. Rather,  
extremum problem  (\ref{LCM-A.1_10_LCM_1}) is re-visited  in Section~\ref{G-LCM-A-LQG}, where closed form expressions are obtained via direct connections to   Linear Quadratic Gaussian (LQG) stochastic optimal control problems. \\

\end{remark}

\subsection{Characterizations of FTFI Capacity for Gaussian Linear Channel Models B.1}
\label{ex-LCM-B.1}
Consider the Gaussian Linear Channel Model B.1  (G-LCM-B.1) (i.e.,  a special case of NCM-B with $M=1$), and  defined by 
\begin{align}
&B_i   = C_{i,i-1} B_{i-1} +D_{i,i}  A_i + V_{i},\hso B_{-1}=b_{-1},  \hso i= 0, \ldots, n, \label{LCM-A.1_a} \\
&\frac{1}{n+1} \sum_{i=0}^n {\bf E} \Big\{ \langle A_i, R_{i,i} A_i \rangle + \langle B_{i-1}, Q_{i,i-1} B_{i-1} \rangle \Big\}\leq \kappa, \hso   R_{i,i} \in {\mb S}_{++}^{q\times q}, \hso Q_{i,i-1} \in {\mb S}_{+}^{p \times p}, \hso i=0, \ldots, n \label{LCM-A.1_aa}
\end{align} 
under  the following assumption.\\
 {\bf Assumption B.1.(i).} 1) Assumption B.(i) of Definition~\ref{exa_A_D} holds. 2) the noise   $\big\{V_i \sim N( 0, K_{V_i}): i=0,1, \ldots, n \big\}$ is independent and Gaussian distributed.

Clearly, all statements regarding the G-LCM-A, defined by (\ref{LCM-A.1}) (given in Section~\ref{exa_CG-NCM-A.1}), 
can be specialized to G-LCM-B.1.
The following statements are listed for future reference.

{\bf Characterization of the FTFI Capacity.} The characterization of the FTFI capacity  of G-LCM.B.1 is given by
\begin{align}
{C}_{A^n \rar B^n}^{FB,G-LCM-B.1} (\kappa) 
 =&  \sup_{\big\{\pi_i^g(da_i | b_{i-1}),  i=0,\ldots,n   \big\} \in  \sr{\circ}{\cal P}_{[0,n]}^{G-LCM-B.1}(\kappa)   }   \sum_{i=0}^n H(B_i^g|B_{i-1}^g) - H(V^n)    \label{LCM_B.1_3} 
\end{align}
where 
\begin{align}
&\sr{\circ}{\cal P}_{[0,n]}^{G-LCM-B.1}(\kappa)   \tri \Big\{ \pi_i^g(da_i|b_{i-1}), i=0,\ldots, n: \frac{1}{n+1} \sum_{i=0}^n  {\bf E}^{\pi^g} \big\{   \langle A_i^g, R_{i,i} A_i^g \rangle + \langle B_{i-1}^g, Q_{i,i-1} B_{i-1}^g \rangle  \big\} \leq \kappa \Big\} \\
&{\mb P}\big\{B_i^g \leq b_i | B_{i-1}^g=b_{i-1}\big\}= \int_{{\mathbb A}_i} {\mb P}\big\{V_{i}\leq b_i - C_{i,i-1} b_{i-1}- D_{i,i}  a_i  \Big\} \pi_i^g(da_i|b_{i-1}),  \hso  i=0,1, \ldots, n \label{LCM-B.1_4} 
\end{align}
that is,   $\{ \pi_i^g(da_i|b_{i-1})\equiv {\bf P}_{A_i| B_{i-1}}^g(a_i| b_{i-1}) : i=0,1, \ldots, n\}$ is   Gaussian satisfying  the average transmission cost constraint, implying $\{{\bf P}_{B_i|B_{i-1}}(b_i|b_{i-1})\equiv {\bf P}_{B_i|B_{i-1}}^g(b_i|b_{i-1}): i=0,1, \ldots, n\}$ is also  Gaussian.

{\bf Alternative Characterization of FTFI Capacity.}
The set of all  channel input conditional distribution is realized by randomized  strategies, as follows. 
\begin{align} 
 &A_i^g=e_i^{B.1}(B_{i-1}^g, Z_i) =\Gamma_{i,i-1} B_{i-1}^g + Z_i,\hst  i=0, \ldots, n,\label{LCM-B.1_7} \\
 &B_i^g   =\Big(C_{i,i-1} + D_{i,i} \Gamma_{i,i-1}  \Big) B_{i-1}^g + D_{i,i} Z_i + V_{i}, \hso B_{-1}^{g}=b_{-1},  \hso i= 0, \ldots, n,  \label{LCM-B.1_8}
 \\
 &i) \hso Z_i  \hso  \mbox{independent of}\hso  \Big(A^{g, i-1}, B^{g,i-1}\Big), \hso ii) \hso Z^i \hso \mbox{independent of} \hso V^i, \hso  \mbox{for} \hso i=0, \ldots, n, \\
 &iii) \Big\{Z_i \sim N(0, K_{Z_i}):  i=0, \ldots, n\Big\} \hso \mbox{an independent Gaussian process}.
 \end{align}
 The following are easily obtained, from the above equation.
 \begin{align}
 &{\bf \mu}_{B_i|B_{i-1}} \tri  {\bf E} \Big\{ B_i^g\Big| B_{{i-1}}^{g}\Big\} = \Big(C_{i,i-1} + D_{i,i} \Gamma_{i,i-1}  \Big) B_{i-1}^g, \hso i=0,\ldots, n, \label{PUT_C.1} \\
 &{ K}_{B_i|B_{i-1}} \tri  {\bf E} \Big\{ \Big(B_i^g -  {\bf \mu}_{B_i|B_{i-1}} \Big)  \Big(B_i^g -  {\bf \mu}_{B_i|B_{i-1}} \Big)^T  \Big| B_{i-1}^g\Big\}=D_{i,i} K_{Z_i}D_{i,i}^T+ K_{V_i}, \hso i=0, \ldots, n, \\
&  K_{B_{i}} \tri   {\bf E} \Big\{ B_{i}^g  \big(B_{i}^g\big)^T \Big\}, \hso i=0,1, \ldots, n \hso \mbox{satisfies the discrete time-varying Lyapunov equation}\\
 &K_{B_{i}}=\Big(C_{i,i-1}+D_{i,i} \Gamma_{i,i-1}\Big)  K_{B_{i-1}} \Big(C_{i,i-1}+ D_{i,i}\Gamma_{i,i-1}\Big)^T +D_{i,i} K_{Z_i} D_{i,i}^T +K_{V_i}, \hso i=0, \ldots, n, \label{MP_1_1}\\
 & K_{B_{-1}}= \mbox{Given}. \label{MP_1_a}
 \end{align}
 Consequently, the alternative characterization of the FTFI capacity is given, as follows. 
 \begin{align}
 &{C}_{A^n \rar B^n}^{FB,G-LCM-B.1} (\kappa) ={C}_{A^n \rar B^n}^{FB,IL-G-LCM-B.1} (\kappa) \nonumber \\
 & \tri   \sup_{\big\{  \big(\Gamma_{i,i-1}, K_{Z_i}\big),  i=0,\ldots,n   \big\} \in  {\cal E}_{[0,n]}^{IL-LCM-B.1}(\kappa) \hso \mbox{and (\ref{MP_1_1}), (\ref{MP_1_a}) hold}  } \sum_{i=0}^n H(B_i^{g}| B_{i-1}^{g})- H(V^n),\label{EXTR_B.1}  \\
& \sum_{i=0}^n H(B_i^{g}| B_{i-1}^{g})-H(V^n)= \frac{1}{2} \sum_{i=0}^n  \log \frac{ | D_{i,i} K_{Z_i} D_{i,i}^T +K_{V_i}|}{|K_{V_i}|}, \label{G-lCM-A-CB.1} \\
&{\cal E}_{[0,n]}^{IL-G-LCM-B.1}(\kappa) \tri \Big\{\big(\Gamma_{i,i-1}, K_{Z_i}\big),  i=0,\ldots,n: \sum_{i=0}^n {\bf E} \Big( \langle A_i^g, R_{i,i} A_i^g\rangle + \langle B_{i-1}^g, Q_{i,i-1} B_{i-1}^g \rangle  \Big) \label{MP_1} \\
&\hst \hst \hst =\sum_{i=0}^n tr\Big(R_{i,i} \Gamma_{i-1,i} K_{B_{i-1}} \Gamma_{i,i-1}^T + R_{i,i} K_{Z_i}  + Q_{i,i-1} K_{B_{i-1}} \Big) \leq \kappa
 \Big\}. \nonumber 
\end{align}
This is a classical deterministic optimization problem of a dynamical system, described by the covariance of the channel output process $\{K_{B_i}:i =0, \ldots, n\}$, and  satisfying the discrete time-varying Lyapunov type difference equation (\ref{MP_1_1}), (\ref{MP_1_a}), where  $\big\{K_{B_i}: i=0, \ldots, n\big\}$  is the controlled object, while the control object is  $\big\{(\Gamma_{i,i-1}, K_{Z_i}): i=0, \ldots, n\big\}$, and it  is  chosen to maximize the pay-off. Discrete time-varying Lyapunov type difference equations are extensively utilized  in stability analysis of time-varying linear controlled systems. \\
The next section elaborates further on  the  direct connection between the characterization of FTFI capacity and Discrete-time Lyapunov matrix equations  its per unit time limiting version, and  linear stochastic controlled systems.

\subsubsection{Relations of FTFI capacity and Feedback Capacity of G-LCM-B.1 \& Linear Stochastic Controlled Systems}.\\
\label{rem_G-LCM-B.1-LSCS}
(a) The recursive equation (\ref{MP_1_1}) satisfied by the  covariance $\{K_{B_i}: i=0, \ldots, n\}$  of the output process $\{B_i^g: i=0, \ldots, n\}$ is a Lyapunov type matrix difference equation.  It is possible to apply calculus of variations to determine the pair $\{(\Gamma_{i,i-1}, K_{Z_i}) \in {\mathbb R}^{q \times p}\times {\mb S}_{+}^{q \times q}: i=0, \ldots, n\}$, which maximizes (\ref{EXTR_B.1}). However, since this is done in a subsequent section via dynamic programming, this direction is not pursued any further. 

For the remaining discussion,  the  properties of time-invariant Lyapunov  difference and algebraic equations,   given in Appendix~\ref{appendix_B}, Theorem~\ref{lemma_vanschuppen} are utilized to  analyze the FTFI capacity and feedback capacity of the G-LCM-B.1.\\
(b) Suppose the coefficients of the G-LCM-B.1 defined by (\ref{LCM-A.1_a}), (\ref{LCM-A.1_aa}) are time-invariant, and the parameters of the optimal channel input distributions induced by (\ref{LCM-B.1_7}), are restricted to  time-invariant, i.e., 
\begin{align}
&C_{i, i-1}=C, \; D_{i,i}=D, \; K_{V_i}=K_V,\; R_{i,i}=R,\hso i=0, \ldots, n, \hso Q_{i,i-1}=Q, \; i=0, n-1,\;  Q_{n,n-1}=M,\\
&(\Gamma_{i,i-1}, K_{Z_i})=(\Gamma, K_{Z}), \hso i=0, \ldots, n.
\end{align}
Recursive substitution gives 
\begin{align}
K_{B_{i}}= \Big(\big[C+D \Gamma\big]^i \Big)  K_{B_{0}} \Big(\big[C+ D\Gamma\big]^i\Big)^T +\sum_{j=0}^{i-1} \Big(\big[C+D \Gamma\big]^j \Big)\Big(  D K_{Z} D^T +K_{V}\Big)\Big(\big[C+ D\Gamma\big]^j\Big)^T, \hso i=1, \ldots, n. \label{rec_LA}
 \end{align}

Suppose  the set of all eigenvalues of $(C+D\Gamma )$ lie in the open unit disc of the space of complex numbers ${\mb C}$, i.e., $spec\big(C+D\Gamma\big) \subset {\mb D}_o$.
Then, irrespectively of the initial covariance $K_{B_{0}}$, the limit,  $\lim_{n \longrightarrow \infty} K_{B_i}=K_B$ exists and satisfies the Lyapunov algebraic matrix equation
\begin{align}
K_{B}=\Big(C+D \Gamma\Big)  K_{B} \Big(C+D\Gamma\Big)^T +D K_{Z} D^T +K_{V} \hso \mbox{and $K_{B}\succeq 0$ is a unique solution} \label{MP_1_1_TI}
\end{align}
In addition,  if $K_{B_{0}}=K_B$, then the solution of the Lyapunov matrix difference equation (\ref{MP_1_1}) with time-invariant coefficients  is time-invariant.\\
The per unit time limiting version of the characterization  of the FTFI capacity is given by the following expression.
 \begin{align}
 &{C}_{A^\infty \rar B^\infty}^{FB,G-LCM-B.1} (\kappa) 
 \tri   \sup_{\big\{  \big(\Gamma, K_{Z}\big) \in {\mb R}^{q\times p} \times S_+^{q \times q} \big\} \in  {\cal E}_{[0,\infty]}^{IL-G-LCM-B.1}(\kappa), \hso \mbox{ (\ref{MP_1_1_TI}) holds}  } \frac{1}{2}\log \frac{ | D K_{Z} D^T +K_{V}|}{|K_{V}|} ,\label{EXTR_B.1_SE} \\
&{\cal E}_{[0,\infty]}^{IL-G-LCM-B.1}(\kappa) \tri \Big\{\big(\Gamma, K_{Z}\big)\in {\mb R}^{q\times p} \times {\mb S}_+^{q \times q}:  tr\Big(R \Gamma K_{B} \Gamma^T + R K_{Z}  + Q K_{B} \Big) \leq \kappa
 \Big\}, \hso spec\big(C+D\Gamma\big) \subset {\mb D}_o. \label{EXTR_B.1_SE_a}
\end{align}
If  $spec\big(C+D\Gamma\big) \subset {\mb D}_o$, then the joint distribution of the joint process $\{A_i^g, B_i^g): i=0, \ldots, \}$ and its marginals are asymptotically ergodic, and hence  (\ref{EXTR_B.1_SE}) is the feedback capacity.  Appendix~\ref{appendix_B}, Theorem~\ref{lemma_vanschuppen}, gives sufficient conditions, which imply $spec\big(C+D\Gamma\big) \subset {\mb D}_o$,  and existence of per unit time limiting version of the characterization  of the FTFI capacity, and  existence of  unique invariant distribution of the joint process $\{(A_i^g, B_i^g): i=0, \ldots, \}$.  
 The complete analysis is done in Section~\ref{putl-ex} via dynamic programming. \\
 Next, the scalar channel is analyzed  to provide an alternative derivation of  statements described  by (\ref{dual_CAP_CI})-(\ref{Exact_Sol_2}), derived through the algebraic Riccati equation (\ref{DP-UMCO_C12_aa_intro_alg}).\\
(i) Scalar Channel. Suppose $p=q=1$ and $R=1, Q=0$. The explicit solution of feedback capacity (\ref{EXTR_B.1_SE}) is  obtained below.  
From (\ref{MP_1_1_TI}), then 
\bea
K_B= \frac{D^2 K_Z+K_V}{1- \big(C+ D\Gamma\big)} \hso \mbox{if} \hso |C+D\Gamma|<1. \label{COV_1}
\eea
The constraint optimization problem (\ref{EXTR_B.1_SE}) is convex,  and by  substituting (\ref{COV_1}) into (\ref{EXTR_B.1_SE_a}),  it is equivalent to the following unconstraint optimization (see \cite{dluenberger1969}). 
\begin{align}
J(K_Z^*,s^*) \tri \inf_{s \geq 0} \sup_{\Gamma \in {\mb R}, K_Z \geq 0}\Big\{  \frac{1}{2}\log \frac{D^2 K_{Z} +K_{V}}{K_{V}}-s \Big(\Gamma^2 \frac{D^2 K_Z+K_V}{1- \big(C+ D\Gamma\big)}+ K_Z -\kappa\Big)\Big\}, \hso spec\big(C+D\Gamma\big) \subset {\mb D}_o.
\end{align}
where $s \geq 0$ is the Lagrange multiplier associated with the constraint. The above problem gives the following optimal solution.
\begin{align}
&\mbox{If} \hso |C|< 1  \hso \mbox{then} \hso \Gamma^*=0, \hso    K_Z^*=\kappa, \hst \kappa \in [0, \infty).  \\
&\mbox{If}    \hso  |C| > 1 \hso \mbox{then} \hso \Gamma^*=-\frac{C^2-1}{CD}, \hso   K_Z^*= \frac{D^2\kappa +K_V(1-C^2)}{C^2D^2}\geq 0, \hso \kappa \in [\kappa_{min}, \infty), \\
& s^*=\frac{1}{2}\frac{D^2}{D^2 \kappa + K_V}\in [s_{min}^*, \infty), \hso   \kappa_{min}\tri \frac{(C^2-1)K_V}{D^2}, \hso s_{min}^* \tri \frac{1}{2}\frac{D^2}{C^2 K_V}.
\end{align}
The feedback capacity is obtained by substituting the optimal values $(\Gamma^*, K_Z^*)$ into 
(\ref{EXTR_B.1_SE}) to deduce the following expression.
\begin{align}
{C}_{A^\infty \rar B^\infty}^{FB, G-LCM-B.1}  (\kappa)= \left\{ \begin{array}{llll}  \frac{1}{2} \ln \frac{  D^2 \; \kappa +K_{V}}{K_{V}} & \mbox{if}  & |C|<1, & i.e., \:  K_Z^*=\kappa  \\
 \frac{1}{2} \ln \frac{ D^2 K_Z^* +K_{V}}{K_{V}} & \mbox{if}  & |C|>1, &  \kappa \in [\kappa_{min}, \infty) \\
 0 & \mbox{if}  & |C|>1, &  \kappa \in [0, \kappa_{min}]. 
\end{array} \right.  \label{dual_CAP_Alt}
\end{align}
This is precisely the feedback capacity obtained in (\ref{dual_CAP}), using the solutions of the Riccati equation. \\
The following universal bould on feedback capacity is obtained, by evaluating the middle identity in (\ref{dual_CAP_Alt}) at $K_{Z}^*\equiv K_Z^*(\kappa)\Big|_{\kappa =\kappa_{min}}$. 
\begin{align}
\mbox{If} \hst |C|>1 \hst \mbox{then} \hst {C}_{A^\infty \rar B^\infty}^{FB, G-LCM-B.1}  (\kappa) \geq \ln |C|, \hst \forall \kappa \in [\kappa_{min}, \infty).
\end{align}
The above solution corresponds, precisely, to   statements described  by (\ref{dual_CAP_CI})-(\ref{Exact_Sol_2}) and  obtained via the solutions of the algebraic  Riccati equation (\ref{DP-UMCO_C12_aa_intro_alg}). \\
Moreover, the above solution illustrates the direct connection to linear stochastic systems and stability theory  via Lyapunov equations.  The general MIMO G-LCM-B.1 is addressed  in Section~\ref{opt_ex-LCM-B.1}, by invoking dynamic programming. 

\subsection{Characterization of FTFI Capacity of G-LCM-B.1 and The LQG Theory of Directed Information}
\label{opt_ex-LCM-B.1}
 The objective of this section is  to completely solve  the extremum problem corresponding to the characterization of FTFI capacity of the G-LCM-B.1, and to gain insight on how to solve more general versions, such as, the G-LCM-B (i.e., when the channel distribution depends on arbitrary memory), and the G-LCM-A. \\This is done by re-formulating such extremum problems, using   Linear Quadratic Gaussian (LQG) stochastic optimal control theory,  with randomized strategies (instead of deterministic as in the standard  LQG theory \cite{vanschuppen2010,kumar-varayia1986}). Via this re-formulation,  the optimal  deterministic part of the randomized strategy, $\{\Gamma_{i,i-1}^*: i=0, \ldots, n\}$, is found explicitly, in terms of  solutions of Riccati matrix difference equations, while the random part 
 $\{K_{Z_i}^*: i=0, \ldots, n\}$, is determined from a sequential water filling problem, similar to that of MIMO memoryless channels \cite{teletar1999}.   \\
The subsequent methodology is based the following simple observations. 

\begin{description}
\item[(i)] Define the randomized stategy of  the  equivalent characterization of  FTFI capacity given by (\ref{EXTR_B.1})-(\ref{MP_1}), as follows.
\bea
A_i^g \tri  U_{i}^g + Z_i, \hst  U_{i}^g \tri  g_{i}^{B.1}(B_{i-1}^g)\equiv \Gamma_{i,i-1}B_{i-1}^g,\hso i=0, \ldots, n \label{DEC_IN}
\eea
where $\{U_i^g: i=0, \ldots, n\}$  is the deterministic part of the strategy and $\{Z_i: i=0, \dots, n\}$ its random part. Then   $\{U_i^g: i=0, \ldots, n\}$ is the control process, chosen to control the channel output process $\{B_i^g: i=0, \ldots, n\}$, and $\{Z_i: i=0, \ldots, n\}$ is the innovations process, chosen to transmit new information over the channel. 
\item[(ii)] Apply  
 dynamic programming to determine recursively the optimal deterministic strategy $\{g_{i}^{B.1,*}(\cdot): i=0, \ldots, n\}$ and the optimal randomized process $\{Z_i: i=0, \ldots, n \}$ (i.e., $\{K_{Z_i}^*: i=0, \ldots, n\}$), from  which the optimal  solution $\big\{(\Gamma_{i,i-1}^*, K_{Z_i}^*): i=0, \ldots, n\big\}$,  can be constructed. 
 
\end{description} 
 
Indeed, this methodology  unfolds  all  consequences and the role of the control process $\{U_i^g: i=0, \ldots, n\}$  to affect the controlled process $\{B_i^g: i=0, \ldots, \}$, for the extremum problem of FTFI capacity characterization, and its per unit time limiting version, the feedback capacity.   


The next theorem establishes the direct connection between LQG stochastic optimal control theory and the characterization of  FTFI capacity, for MIMO G-LCM-B.1.\\

\begin{theorem}(Optimal strategies of FTFI capacity of G-LCM-B.1)\\
\label{DP-UMCO}
Consider the G-LCM-B.1 defined by (\ref{LCM-A.1_a}), (\ref{LCM-A.1_aa}), under Assumptions B.1.(i).\\
(a)
 Define 
\bea
A_i^g \tri  U_{i}^g + Z_i, \hso U_{i}^g = g_{i}^{B.1}(B_{i-1}^g)\equiv \Gamma_{i,i-1}B_{i-1}^g, \hso i=0, \ldots, n \label{DP_UMCO_C1}
\eea
where $\{U_i^g: i=0, \ldots, n\}$ is the deterministic part of the randomized strategy (control part) and $\{Z_i: i=0, \ldots, n\}$ is the random part.  Then 
\bea
B_i^g= C_{i,i-1}B_{i-1}^g+ D_{i,i} U_{i}^g + D_{i,i}Z_i + V_i, \hso i=0, \ldots, n, \hso B_{i-1}^g=b_{-1} \label{DP_UMCO_C2}
\eea
and the equivalent characterization of the FTFI capacity is given by
\begin{align}
{C}_{A^n \rar B^n}^{FB,G-LCM-B.1} (\kappa) 
 ={C}_{A^n \rar B^n}^{FB,IL-G-LCM-B.1} (\kappa) =  \sup_{\big\{(g_i^{B.1}(\cdot), K_{Z_i}),   i=0,\ldots,n   \big\} \in  {\cal E}_{[0,n]}^{B.1}(\kappa)   }  \sum_{i=0}^n H(B_i^g|B_{i-1}^g) - H(V^n)  \label{LCM_B.1_3_C30} 
\end{align}
where 
\begin{align}
&\sum_{i=0}^n H(B_i^g|B_{i-1}^g) - H(V^n) =(\ref{G-lCM-A-CB.1}), \\ 
&{\cal E}_{[0,n]}^{B.1}(\kappa)   \tri \Big\{g_i^{B.1}: {\mathbb R}^p \longmapsto {\mathbb R}^q,\hso  u_i=g_i^{B.1}(b_{i-1}),  \hso  K_{Z_i} \in {\mb S}_{+}^{q \times q}, \hso i=0,\ldots, n:\nonumber \\
&\hst \hst  \frac{1}{n+1}{\bf E}^{g^{B.1}} \Big( \sum_{i=0}^n     \Big[\langle A_i^g, R_{i,i} A_i^g \rangle + \langle B_{i-1}^g, Q_{i,i-1} B_{i-1}^g \rangle \Big] \Big) \leq \kappa \Big\}. \label{LCM_B.1_3_C40}
\end{align}
For the rest of the statements assume  there exist an $\big\{(B_i^g, g_i^{B.1}(\cdot), Z_i): i=0, \ldots, \big\}$ in the Hilbert space of square summable sequences,  such that the feasible set in (\ref{LCM_B.1_3_C40}) has an interior point (convexity of pay-off functional and constraint set can be shown).\\
(b)  The cost-to-go  $C_i^{B.1}: {\mb R}^p \longmapsto {\mb R}$ (corresponding to  (\ref{LCM_B.1_3_C30})),  from time ``$i$'' to the terminal time   ``$n$'' given the value of the output $B_{i-1}^g=b_{i-1}$ is defined   by 
\begin{align}
C_i^{B.1}(b_{i-1})\tri &  \sup_{ \Big\{ (U_j^g,K_{Z_j})\in {\mathbb R}^q\times {\mb S}_{+}^{q\times q}, U_j^g=g_j^{B.1}(B_j^g),\:  j=i, \ldots, n   \Big\} }  \Bigg\{ \frac{1}{2} \sum_{j=i}^n\log \frac{ | D_{j,j} K_{Z_j} D_{j,j}^T+K_{V_j}|}{|K_{V_j}|} -\sum_{j=i}^n  tr\Big(s R_{j,j} K_{Z_j}\Big)  + s(n+1)\kappa        \nonumber \\
&- s{\bf E}^{g^{B.1}} \Big\{ \sum_{j=i}^n  \left[\langle U_j^g, R_{j,j} U_j^g\rangle + \langle B_{j-1}^g, Q_{j,j-1} B_{j-1}^g \rangle\right]     \Big| B_{i-1}^g=b_{i-1}\Big\} \Bigg\} \label{LCM_B.1_3_C50}
\end{align}
where $s \geq 0$ is the Lagrange multiplier associated with  the average transmission cost constraint (\ref{LCM_B.1_3_C40}). \\
(c) The dynamic programming recursions are given by the following equations.
\begin{align}
&C_n^{B.1}(b_{n-1})= \sup_{  (u_n,K_{Z_n})\in {\mathbb R}^q\times {\mb S}_+^{q\times q}   }   \Big\{ \frac{1}{2} \log \frac{ | D_{n,n} K_{Z_n} D_{n,n}^T+K_{V_n}|}{|K_{V_n}|} - tr\Big(s R_{n,n} K_{Z_n}\Big) + s(n+1)\kappa \nonumber \\
&- s \Big[ \langle u_n, R_{n,n} u_n\rangle + \langle b_{n-1}, Q_{n,n-1} b_{n-1} \rangle  \Big]   \Big\}, \label{DP-UMCO_C10} \\
&C_i^{B.1}(b_{i-1})=  \sup_{  (u_i,K_{Z_i})\in {\mathbb R}^q\times {\mb S}_+^{q\times q}  }   \Bigg\{ \frac{1}{2} \log \frac{ | D_{i,i} K_{Z_i} D_{i,i}^T+K_{V_i}|}{|K_{V_i}|}- tr\Big(s R_{i,i} K_{Z_i}\Big)  \nonumber \\
&- s \Big[ \langle u_i, R_{i,i} u_i\rangle + \langle b_{i-1}, Q_{i,i-1} b_{i-1} \rangle  \Big]   + {\bf E}^{g^{B.1}}\Big\{  C_{i+1}^{B.1}(B_{i}^g)       \Big| B_{i-1}^g=b_{i-1}\Big\}  \Bigg\}, \hso i=0, \ldots, n-1.
\end{align}
(d) The optimal deterministic part of the randomized strategy,  $\{g_i^{B.1,*}(\cdot): i=0, \ldots, n\}$,  and the corresponding covariance $K_{B_{i}} \tri   {\bf E} \big\{ B_{i}^g  \big(B_{i}^g\big)^T \big\},  i=0,1, \ldots, n$, are   given by the following equations.
\begin{align}
&g_i^{B.1,*}: {\mathbb R}^p \longmapsto {\mathbb R}^q, \hso i=0, \ldots, n, \hso \Gamma^*: \{0,1,\ldots, n\} \longmapsto {\mathbb R}^{q\times p}, \hso P : \{0,1,\ldots, n\} \longmapsto {\mb S}_{+}^{p\times p}, \\
&g_i^{B.1,*}(b_{i-1})= F^*(i)\; b_{i-1}\equiv \Gamma_{i,i-1}^*b_{i-1}, \hso i=0,\ldots, n, \\
& F^*(n)=\Gamma_{n,n-1}^*=0, \hso F^*(i)=-H_{22}^{-1}(i) H_{12}^T(i) \\
& H_{11}(i)= C_{i,i-1}^T P(i+1) C_{i,i-1}+s Q_{i,i-1}, \hso H_{12}(i)= C_{i,i-1}^T P(i+1) D_{i,i}, \hso H_{22}(i)= D_{i,i}^T P(i+1) D_{i,i}+s R_{i,i}, \\
&P(i)=H_{11}(i) - H_{12}(i) H_{22}^{-1}(i)H_{12}^T(i), \hso i=0, \ldots, n-1, \label{DP-UMCO_C12} \\
&P(i) =C_{i,i-1}^T P(i+1) C_{i,i-1}+s Q_{i,i-1}-C_{i,i-1}^T P(i+1) D_{i,i}\Big(D_{i,i}^T P(i+1) D_{i,i}+s R_{i,i}\Big)^{-1} \Big(C_{i,i-1}^T P(i+1) D_{i,i}\Big)^T \label{DP-UMCO_C12_a}\\
&P(n)=s Q_{n,n-1}, \label{DP-UMCO_C12_aa}\\
 &K_{B_i}=\Big(C_{i,i-1} + D_{i,i} \Gamma_{i,i-1}^*\Big)  K_{B_{i-1}} \Big(C_{i,i-1}+D_{i,i}\Gamma_{i,i-1}^*\Big)^T +D_{i,i} K_{Z_i} D_{i,i}^T +K_{V_i}, \hso i=0, \ldots, n, \label{MP_1_New_1}\\
 & K_{B_{-1}}= \mbox{Given}. \label{MP_1_New_2}
\end{align}
(e) The solution of the dynamic programming equations is given by the following equations.
\begin{align}
C_i^{B.1}(b_{i-1})= -\langle b_{i-1}, P(i) b_{i-1}\rangle +   r(i) , \hso i=0, \ldots, n
\end{align}
where $\{P(i):=0, \ldots, n\}$ satisfies the backward recursive  Riccati equation (\ref{DP-UMCO_C12_a}), (\ref{DP-UMCO_C12_aa}), the process $\{r(i): i=0, \ldots,n \}$ satisfies the backward recursion
\begin{align}
r(i)= & r(i+1) + \sup_{ K_{Z_i}\in {\mb  S}_+^{q\times q}}\Big\{ \frac{1}{2} \log \frac{ | D_{i,i} K_{Z_i} D_{i,i}^T+K_{V_i}|}{|K_{V_i}|} - tr\Big(s\; R_{i,i} K_{Z_i}\Big) \nonumber \\
&- tr \Big(P(i+1) \Big[D_{i,i} K_{Z_i}D_{i,i}^T+ K_{V_i}\Big]\Big)\Big\}, \hso i=0, \ldots, n-1,  \label{DP-UMCO_C13}  \\
r(n)=&\sup_{ K_{Z_n}\in  {\mb S}_+^{q\times q}} \Big\{\frac{1}{2} \log \frac{ | D_{n,n} K_{Z_n} D_{n,n}^T+K_{V_n}|}{|K_{V_n}|}+s (n+1)\kappa- tr\Big( s R_{n,n} K_{Z_n}\Big)\Big\} \label{DP-UMCO_C20}
\end{align}
and moreover the optimal deterministic part of the randomized strategy is given by
\begin{align}
&g_i^{B.1,*}(b_{i-1})= -\Big(D_{i,i}^T P(i+1) D_{i,i} + s R_{i,i}\Big)^{-1} D_{i,i}^T P(i+1) C_{i,i-1} \: b_{i-1}\equiv \Gamma_{i,i-1}^* b_{i-1}, \hso i=0, \ldots, n-1, \label{str_1}\\
&g_n^{B.1,*}(b_{n-1})=0. \label{str_2}
\end{align}
(f) The optimal covariance (the random part of the randomized strategy) $\{K_{Z_i}^*:i=0, \ldots, n\}$ and $s^* \geq 0$ are found from the problem
\begin{align}
\sup_{s \geq 0} \big\{-\langle b_{-1}, P(0) b_{-1}\rangle +   r(0)\Big\}\hso \mbox{subject to (\ref{DP-UMCO_C13}), (\ref{DP-UMCO_C20}), (\ref{DP-UMCO_C12_a}), (\ref{DP-UMCO_C12_aa}). }
\end{align}

(g) The characterization of FTFI capacity (for any $s\geq 0$ corresponding to $\kappa$) is  given by 
\begin{align}
C_{A^n\rar B^n}^{FB, G-LCM-B.1}(\kappa) = -\int_{{\mathbb R}^p} \langle b_{-1}, P(0) b_{-1}\rangle {\bf P}_{B_{-1}}(db_{-1}) +  r(0).
\end{align}

\end{theorem} 
\begin{proof} See Appendix~\ref{appendix_C}.
\end{proof}
 
The derivation the closed form expressions given in  Theorem~\ref{DP-UMCO}, for the  G-LCM.B.1  is attributed to the decomposition of the randomized information lossless strategies (\ref{DEC_IN}), where the innovations process is an orthogonal process, and the separation principle,  established via dynamic programming. \\
It appears these two features are vital and should be  incorporated in other extremum problems of feedback capacity, such as,  the Cover and Pombra \cite{cover-pombra1989}  characterization of FTFI capacity  given by  (\ref{c-p1989}) or any of its variants \cite{kim2010}. Specifically, the orthogonality of $\{Z_i: i=0, \ldots, n\}$ is missing in the characterization obtained in Cover and Pombra \cite{cover-pombra1989} (although the authors  continue to call this process an innovations process). These points are further ellaborated below.\\
 
 \begin{remark}(Relation to Cover and Pombra  \cite{cover-pombra1989})\\
\label{kim-2010}
As pointed out in Remark~\ref{rem-sep_A}, it is difficult to obtain closed form solutions  to the extremum problem of the Cover and Pombra \cite{cover-pombra1989} scalar AGN channel, without re-visiting the derivation to obtain a realization of optimal channel input distribution having the specific decomposition (\ref{DEC_IN}). In fact, the only known explicit solution to the characterization of Cover and Pombra \cite{cover-pombra1989} scalar AGN channel, is the one obtained by Kim  in \cite{kim2010}, under the assumption  of stationary ergodicity, when   the noise is stationary ergodic and first-order Markov. The main tools applied  in \cite{kim2010} are Power Spectral densities and their relation to  scalar Riccati algebraic equations (for scalar-valued channel input and output processes). It appears very difficult   to extend the main theorems found in \cite{kim2010} to non-stationary, multidimensional processes, because the author's starting point is  the characterization derived by   Cover and Pombra \cite{cover-pombra1989}, and  there is no direct connection to LQG stochastic optimal control theory. Moreover, as illustrated in Section~\ref{rem_G-LCM-B.1-LSCS}, there are are various regimes for feedback capacity, and whether feedback increases capacity, depends on the \'a priori assumptions imposed on the channel. This point  should be accounted for when  analyzing feedback channels.
\end{remark}

 \ \

\begin{remark}(Connections to LQG stochastic optimal control theory)
\label{rem-LQG}\\
(a) Theorem~\ref{DP-UMCO} illustrates the dual role of the randomized strategies (\ref{DP_UMCO_C1}) in extremum problems of directed information. Specifically,  the  optimal deterministic part (\ref{str_1}), (\ref{str_2})  controls the channel output process, precisely as  in LQG stochastic optimal control theory (if $s=1$) \cite{kumar-varayia1986}. However,  its optimal random part $\{Z_i: i=0, \ldots, n\}$  found from (\ref{DP-UMCO_C13}), (\ref{DP-UMCO_C20}),  ensures an optimal innovations process with covariance $\{K_{Z_i}^*: i=0, \ldots, n\}$ is transmitted over the channel,  to achieve the characterization of FTFI capacity, and to meet the average transmission  cost constraint.\\
 Note that from (\ref{DP_UMCO_C1})-(\ref{LCM_B.1_3_C30})  it follows directly that 
\bea
C_{A^n\rar B^n}^{FB, G-LCM-B.1}(\kappa)=0 \hst \mbox{if} \hso K_{Z_i}^*=0: i=0, \ldots, n.
\eea
Hence, the FTFI capacity is zero and consequently, its per unit time limit the feedback capacity is zero, although the output process can be stabilized (under appropriate conditions).\\
This re-confirms and strengthens the following well-known fact of LQG stochastic optimal control or decision theory. Among all all non-Markov randomized policies $\pi_{[0,n]}^{RS} \tri \Big\{{\bf P}_{A_i|A^{i-1},B^{i-1}}: i=0, \ldots,n\Big\}$,  the optimal strategy of the Linear-Quadratic-Gaussian (LQG) Stochastic Optimal Control Problem 
\begin{align}
& J(\pi_{[0,n]}^{RS,*})\tri \inf_{\big\{ {\bf P}_{A_i|A^{i-1},B^{i-1}}: i=0, \ldots, n\big\}}   \frac{s}{n+1} \sum_{i=0}^n {\bf E} \Big\{ \langle A_i, R_{i,i} A_i \rangle + \langle B_{i-1}, Q_{i,i-1} B_{i-1} \rangle \Big\}, \\
&\mbox{subject to} \hso B_i   = C_{i,i-1} \; B_{i-1} +D_{i,i} \; A_i + V_{i},\hso B_{-1}=b_{-1},  \hso i= 0, \ldots, n
\end{align}
is Gaussian and Markov of the form $A_i= g_i^M(B_{i-1})+ Z_i, Z_i \sim N(0, K_{Z_i}), Z_i \perp B^{i-1}, i=0, \ldots,n$, $\{Z_i: i=0, \ldots, n\}$ an orthogonal process,  and occurs in the subclass of nonrandom or deterministic policies $\Big\{ \big(g_i^M(b_{i-1}), Z_i\big) =\big(g_i^{B.1,*}(b_{i-1}), 0\big): i=0, \ldots, n\Big\}$, i.e., ${\bf P}_{A_i|A^{i-1},B^{i-1}}^*={\bf P}_{A_i|B_{i-1}}^*=\delta_{A_i}(g_i^{B.1,*}(b_{i-1}))$,  is a delta measure concentrated at $g_i^{B.1, *}(\cdot)$, $i=0, \ldots, n$).\\
This fact alone, makes directed information very attractive for  designing controllers, which stabilize  controlled dynamical systems,  and ensure information is communicated from, say, the control process to the controlled process.\\ 
 (b) The optimal random part of the strategy is found from a sequential     version of   a water filling solution, (\ref{DP-UMCO_C13}),  (\ref{DP-UMCO_C20}), that depends on the solution of a Riccati difference equation. \\
(c) The extremum solution illustrates a separation between the role of control (deterministic part of the strategy) and the role of information transmission (random part of the strategy).\\
(d)  The material discussed in Section~\ref{mo-dr}, regarding the G-LCM-B.1, given by  (\ref{LCM-A.1_a_Intr})-(\ref{Exact_Sol_2}), and   relating feedback capacity, capacity without feedback and LQG stochastic optimal control theory, are direct consequences of the above theorem, specifically,  the per unit time limiting version of Theorem~\ref{DP-UMCO}, which  is investigated in Section~\ref{putl-ex}. \\
\end{remark}

\subsection{Characterization of FTFI Capacity of G-LCM-B and The LQG Theory}
\label{opt_ex-LCM-B.1}
Consider the G-LCM-B.J (a generalization of the G-LCM-B.1),  defined by 
\begin{align}
&B_i   = \sum_{j=1}^M C_{i,i-j} B_{i-j} +D_{i,i}  A_i + V_{i}, \hso B_{-M}^{-1}=b_{-M}^{-1}, \hso  i= 0, \ldots, n, \label{LCM-A.1_a_B}  \\
&\frac{1}{n+1} \sum_{i=0}^n {\bf E} \Big\{ \langle A_i, R_{i,i} A_i  \rangle+ \langle B_{i-K}^{i-1},Q_{K}(i-1)B_{i-K}^{i-1} \rangle  \Big\}\leq \kappa,  \\
& J \tri \max\{M, K\}, \hso  R_{i,i} \in S_+^{q \times q},\hso Q_K(-1)=0, \hso Q_{K}(i-1) \in S_+^{K p \times Kp}, \hso i=0, \ldots, n,  \hso 
\mbox{Assumption B.1(i)  holds}.\label{LCM-A.1_a_BB}
\end{align} 
It can be verified, by repeating the derivation of Theorem~\ref{G-LCM-A-CA}, if necessary, that the optimal channel input conditional distribution is Gaussian of the form $\{\pi_i^g(da_i|b_{i-J}^{i-1}): i=0, \ldots, n\}$, and  that all material presented in Section~\ref{opt_ex-LCM-B.1}, generalize to  G-LCM-B.J. 

 \subsection{Characterization of FTFI Capacity of G-LCM-A and The LQG Theory}
 \label{G-LCM-A-LQG}
Consider the G-LCM-A defined by (\ref{LCM-A.1}), in which  the channel distribution is not of limited memory, but instead the memory is increasing with time. It is possible to repeat the derivation of Theorem~\ref{DP-UMCO}, with some modifications, as follows. Write  
\bea
A_i^g \tri  U_{i}^g + Z_i, \hso U_{i}^g = g_{i}^{A}(B^{g, i-1})\equiv \Gamma_{i}(i-1) B^{g,i-1}, \hso A_0=Z_0,  \hso i=1, \ldots, n \label{DP_GLCM-A1}
\eea
where $\{U_i^g: i=0, \ldots, n\}$ is the deterministic part of the randomized strategy and $\{Z_i: i=0, \ldots, n\}$ is the random part. Then 
\bea
B_i^g= C_{i}(i-1)B^{g,i-1}+ D_{i,i} U_{i}^g + D_{i,i}Z_i + V_i, \hso B_0^g= D_{0,0} A_0^g + V_0, \hso i=1, \ldots, n .\label{DP_GLCM-A2}
\eea
Clearly,  the dimension of the process $\{S_i^g\tri B^{g,i-1}: i=0, \ldots, n\}$  increases with time $i=0,1, \ldots,n$. 
Nevertheless,  it is claimed that,  by repeating the derivation of Theorem~\ref{DP-UMCO}, the optimal randomized strategy can be found, as a function of solutions to  Riccati difference equations, which at each time $i$, increases in dimension from the previous time $i-1$.

\section{Feeback Capacity of G-LCM-B \& The Infinite Horizon LQG  Theory of Directed Information}
\label{putl-ex}
In this section, the per unit limiting version of G-LCM-B is investigated, and the characterization of  feedback capacity is derived,  irrespectively of whether the eigenvalues of the channel matrix $C$, that is, $spec\big(C\big)$ lie in the open disc of the unit circle in ${\mb C}$. 
Specifically, the characterizations of FTFI capacity given in Section~\ref{ILE} are applied to Gaussian Linear  Channel Models (G-LCMs) of Definition~\ref{exa_A_D}, to obtain the following. 
\begin{description}
\item[(a)] Characterizations of feedback capacity for Multiple Input Multiple Output (MIMO)  G-LCMs,  via the per unit time limit of the characterizations of FTFI capacity of MIMO G-LCMs;

\item[(c)] Relations between  infinite horizon LQG stochastic optimal control theory, linear stochastic feedback controlled  systems, feedback capacity and capacity without feedback. 
\end{description}

\subsubsection{\bf Feedback Capacity of G-LCM-B.1 \& Infinite Horizon LQG Theory} Consider first, the G-LC-B.1 (see (\ref{LCM-A.1_a}), (\ref{LCM-A.1_aa}), since the extension to the general model G-LCM-B.J,  can be treated as discussed in Section~\ref{opt_ex-LCM-B.1} 
  \\
The next theorem establishes a hidden connection between,  infinite horizon per unit time  LQG stochastic optimal control theory,  directed information stability (see (\ref{IS-O_1}), (\ref{IS-O_2})), and   optimal transmission rates. Moreover, through the computation of the feedback capacity, a  separation principle is established,  between the role of deterministic part of the randomized strategy to stabilize  unstable channels, and the role of its random part  to transmit new information.  \\

\begin{theorem}(Feedback capacity of  TI-G-LCM-B.1)\\
\label{DP-UMCO_IH}
Consider the time-invariant version of G-LCM-B.1 (see (\ref{LCM-A.1_a}), (\ref{LCM-A.1_aa}) defined by 
\begin{align}
&B_i   = C\; B_{i-1} +D \; A_i + V_{i},\hso B_{-1}=b_{-1}, \hso \; K_{V_i}=K_V \in S_{++}^{p\times p}, \hso i= 0, \ldots, n, \label{LCM-A.1_a_TI} \\
&\frac{1}{n+1} \sum_{i=0}^n {\bf E} \Big\{ \langle A_i, R A_i \rangle + \langle B_{i-1}, Q B_{i-1} \rangle \Big\}\leq \kappa, \hso   R \in {\mb S}_{++}^{q\times q}, \; Q \in {\mb S}_{+}^{p \times p} \label{LCM-A.1_aa_TI}
\end{align} 
 called TI-G-LGM.B.1. \\
Assume the following conditions hold (see Appendix for definitions and implications).
\begin{align}
&i) \hso \mbox{the pair $(C, D)$ is stabilizable} \label{STA} \\
&ii) \hso \mbox{the pair $(G, C)$ is detectable, where} \hso Q=G^T G, \hso G \in {\mb S}_{+}^{p \times p}. \label{DET}
\end{align}
Moreover, assume the set of channel input conditional distributions is restricted to time-invariant distributions, i.e., $\{\pi_i^g(da_i|b_{i-1})=\pi^{g,\infty}(da_i|b_{i-1}): i=0, \ldots, n \}$.  \\
Then the  following hold.\\
(a) Define 
\bea
A_i^g \tri  U_{i}^g + Z_i, \hso U_{i}^g = g^{B.1}(B_{i-1}^g)\equiv \Gamma B_{i-1}^g, \hso i=0, \ldots, n \label{DP_UMCO_C1_TI}
\eea
where $\{U_i^g: i=0, \ldots, n\}$ is the deterministic part of the randomized strategy and $\{Z_i: i=0, \ldots, n\}$ is its random part. Then 
\bea
B_i^g= C B_{i-1}^g+ D U_{i}^g + D Z_i + V_i, \hso i=0, \ldots, n. \label{DP_UMCO_C2_TI}
\eea
Define 
\begin{align}
&{C}_{A^n \rar B^n}^{FB,B.1} (\kappa) 
 \tri \sup_{\big\{(g^{B.1}(\cdot), K_{Z}),   i=0,\ldots,n   \big\} \in  {\cal E}_{[0,n]}^{B.1}(\kappa)   }   \sum_{i=0}^n H(B_i^g|B_{i-1}^g) - H(V^n) , \hst {C}_{A^\infty \rar B^\infty}^{FB, B.1} (\kappa)\tri \liminf_{n \longrightarrow \infty} \frac{1}{n+1}{C}_{A^n \rar B^n}^{FB, B.1} (\kappa),  \label{LCM_B.1_3_C30_IH} \\
&{\cal E}_{[0,n]}^{B.1}(\kappa)   \tri \Big\{g^{B.1}: {\mathbb R}^p \longmapsto {\mathbb R}^q,\hso  u_i=g^{B.1}(b_{i-1}),  \hso  K_{Z} \in {\mb S}_{+}^{q \times q}, \hso i=0,\ldots, n: \nonumber \\
& \hso \frac{1}{n+1}{\bf E}^{g^{B.1}} \Big\{ \sum_{i=0}^n     \Big[\langle A_i^g, R A_i^g \rangle + \langle B_{i-1}^g, Q B_{i-1}^g \rangle \Big] \Big\} \leq \kappa \Big\} \label{fea_IH}
\end{align}
and assume there exist an $(\big\{B_i^g: i=0, \ldots, \big\}, g^{B.1}(\cdot), K_Z)$ such that the feasible set in (\ref{fea_IH}) has an interior point. \\
Then ${C}_{A^\infty \rar B^\infty}^{FB, B.1} (\kappa)$ is the per unit time version of the characterization of FTFI capacity corresponding to  (\ref{LCM_B.1_3_C30}), that is,   
\begin{align}
{C}_{A^\infty \rar B^\infty}^{FB, B.1} (\kappa)  =& {C}_{A^\infty \rar B^\infty}^{FB, G-LCM-B.1} (\kappa)\tri \liminf_{n \longrightarrow \infty}\frac{1}{n+1} {C}_{A^n \rar B^n}^{FB, G-LCM-B.1} (\kappa)\\
=&{C}_{A^\infty \rar B^\infty}^{FB, IL-G-LCM-B.1} (\kappa)\tri \liminf_{n \longrightarrow \infty} \frac{1}{n+1}{C}_{A^n \rar B^n}^{FB, IL-G-LCM-B.1} (\kappa)  \label{DP_UMCO_C2_TI_1}
\end{align}  
(b)   The pair $\Big(J^{B.1,*},C^{B.1}(b)\Big), J^{B.1,*} \in {\mathbb R}$, $C^{B.1}: {\mathbb R}^p \longmapsto {\mathbb R}$ satisfies the following dynamic programming equation (corresponding to  ${C}_{A^\infty \rar B^\infty}^{FB,B.1} (\kappa)$).   
\begin{align}
J^{B.1,*} + C^{B.1}(b) = &\sup_{ (u,K_{Z})\in {\mathbb R}^q\times {\mb S}_{+}^{q\times q}  }  \Bigg\{ \frac{1}{2} \log \frac{ | D K_{Z} D^T+K_{V}|}{|K_{V}|} - tr\Big(s R K_{Z}\Big)  + s \kappa   - s \Big[ \langle u, R u\rangle + \langle b, Q b \rangle  \Big] \nonumber \\
&  + {\bf E}^{g^{B.1}}\Big\{  C^{B.1}(B_{i}^g)       \Big| B_{i-1}^g=b\Big\}  \Bigg\} \label{DP_IH}
\end{align}
where    $s \geq 0$ is found from the average transmission cost constraint. \\
(c) The optimal stationary policy $g^{B.1,*}(\cdot)$  and corresponding covariance matrix $K$ of  $\{B_i^*: i=0, \ldots, n\}$ are  given by the following equations.
\begin{align}
&g^{B.1,*}: {\mathbb R}^p \longmapsto {\mathbb R}^q, \hso \Gamma \in {\mathbb R}^{q\times p}, \hso P \in {\mb S}_{+}^{p\times p}, \\
&g^{B.1,*}(b)= \Gamma^*b, \\
&\Gamma^*=-H_{22}^{-1} H_{12}^T= - \Big(D^T PD +s R\Big)^{-1}D^T P C, \\
& H_{11}= C^T P C+s Q, \hso H_{12}= C^T P D, \hso H_{22}= D^T P D+sR , \\
&P=H_{11} - H_{12} H_{22}^{-1}H_{12}^T \label{DP-UMCO_C12_IH} \\
&P=C^T P C+s Q-C^T P D\Big(D^T P D+s R\Big)^{-1} \Big(C^T P D\Big)^T, \label{DP-UMCO_C12_IH_new}\\
 &K=\Big(C + D \Gamma^*\Big)  K \Big(C+D\Gamma^*\Big)^T +D K_{Z} D^T +K_{V}, \label{MP_1_New_3}\\
& spec\Big(C+D \Gamma^*\Big)= spec\Big(C-D\Big(D^T P D+s R\Big)^{-1} D^T P C \Big) \subset {\mathbb D}_o. \label{cond_STA}
\end{align}
(d) The solution of the dynamic programming is  given by 
\begin{align}
C^{B.1}(b)=&-\langle b, P b\rangle, \label{DP-UMCO_C13_IH_aa} \\
J^{B.1,*}= &  \sup_{ K_{Z}\in  {\mb S}_+^{q\times q}}\Big\{  \frac{1}{2} \log \frac{ | D K_{Z} D^T+K_{V}|}{|K_{V}|}+s\kappa- tr\Big(s\; R K_{Z}\Big)- tr \Big(P\Big[D K_{Z}D^T+ K_{V}\Big]\Big)\Big\}  \label{DP-UMCO_C13_IH}  \\
g^{B.1,*}(b)=&- \Big(D^T PD +s R\Big)^{-1}D^T P C \; b. \label{DP-UMCO_C13_IH_aaa} 
\end{align}
(e) The optimal covariance $K_{Z}^*$ and $s\geq 0$ are found from the optimization problem
\begin{align}
\inf_{s \geq 0} J^{B.1,*} \hso \mbox{subject to (\ref{DP-UMCO_C12_IH_new})}.   \label{RP_R1}
\end{align}
The average transmission cost constraint evaluated on the optimal strategy  is given by
\begin{align}
{\bf E}^{g^{B.1,*}} \Big\{ \langle g^{B.1,*}(B^*), R g^{B.1,*}(B^*)\rangle + \langle B^*, Q B^* \rangle \Big\}
+ tr\Big( R K_{Z}^*\Big)\leq  \kappa
\end{align}
where the expectation is with respect to the invariant distribution ${\bf P}_{B}^{g^{B.1,*}}(db)$ of the optimal output process $\{B_i^*: i=0, \ldots, n\}$ corresponding to $(g^{B.1,*}(\cdot),K_{Z}^*)$. \\
(f)  $J^{B.1,*}\Big|_{s=s^*} = {C}_{A^\infty \rar B^\infty}^{FB, B.1} (\kappa)$, where $s^*$ is the Lagrange multiplier found from (\ref{RP_R1}) or the average constraint  via  (\ref{RP_R1}).\\
(g) The information density and the constraint evaluated at the optimal stationary strategy are information stable, (see (\ref{IS-O_1}), (\ref{IS-O_2}) for precise definition). Specifically,  for any initial distribution ${\bf P}_{B_{-1}}(db_{-1})=\mu(db_{-1}) \in {\cal M}({\mb R}^p)$, the following hold.
\begin{align}
{C}_{A^\infty \rar B^\infty}^{FB, B.1} (\kappa) \equiv   & J(\pi^{g,\infty,*},\mu)=J^{B.1,*}\Big|_{s=s^*}, \hso  \forall \mu(\cdot) \in {\cal M}({\mb R}^p), \\
J^0(\pi^{g,\infty, *}, \mu)=&J^{B.1,*}\Big|_{s=s^*}, \hso {\bf P}_\mu^{\pi^{g,\infty,*}}-a.s., \hso \forall \mu(\cdot) \in {\cal M}({\mb R}^p)
\end{align} 
where 
\begin{align}
J^{0}(\pi^{g,\infty}, \mu) \tri &\liminf_{n \longrightarrow \infty} \frac{1}{n+1} \sum_{i=0}^n  
\log\Big(\frac{dQ_i(\cdot|B_{i-1},A_i)}{v_i^{ \pi^{g,\infty}}(\cdot|B_{i-1})}(B_i)\Big), 
 \label{CT-B.8} \\
\sr{\circ}{\cal P}_{[0,\infty]}^{\infty, B.1}(\kappa) \tri   & \Big\{    \pi^{g,\infty}(da_i|b_{i-1}),  i=0, 1, \ldots, n: 
\limsup_{n \longrightarrow \infty} \frac{1}{n+1} \sum_{i=0}^n  \Big(\langle A_i^g, R A_i^g \rangle + \langle B_{i-1}^g, Q B_{i-1}^g \rangle \Big)     \leq \kappa\Big\} .\label{CT-B.7_a}
\end{align} 

\end{theorem} 
\begin{proof} (a) This follows as in Theorem~\ref{DP-UMCO}.\\
 (b)-(e) By  the stabilizability  and detectability conditions, i), ii) the dynamic programming equation (\ref{DP_IH}) holds (see \cite{kumar-varayia1986,vanschuppen2010}). By repeating the derivation of Theorem~\ref{DP-UMCO}, if necessary,  (c)-(d) are obtained. \\
 (f), (g) These follow from the fact that $N(0,K_B)$ is the unique invariant Gaussian distribution of  (\ref{DP_UMCO_C2_TI}), corresponding to the stabilizing  optimal policy (\ref{DP-UMCO_C13_IH_aaa}) (i.e, (\ref{cond_STA}) holds), and the ergodic properties of LQG stochastic optimal control theory \cite{kumar-varayia1986,vanschuppen2010}. 
\end{proof}

Theorem~\ref{DP-UMCO_IH} gives sufficient conditions in terms of detectability and stabilizability, i.e., (\ref{STA}), (\ref{DET}),   for existence of feedback capacity, irrespectively of whether the eigenvalues of channel matrix $C$ are stable or unstable, that is, whether they lie in the open unit disc of complex numbers, $spec\big(C)\subset {\mb D}_o \tri \big\{c \in {\mb C}: |c| <1\big\}$ or outside $spec\big(C)\subset {\mb D}_o^c \tri \big\{c \in {\mb C}: |c| > 1\big\}$. \\
In fact, the above theorem demonstrates that feedback capacity (i.e., the supremum of all achievable rates) depends on the \'a priori assumptions on the channel model coefficients, $\{C, D, R, Q, K_V\}$, because these determine whether  the conditions of  stabilizability of the pair $\big(C, D\big)$ and detectability of the pair $\big(G, C\big)$, i.e., (\ref{STA}), (\ref{DET}) are satisfied. \\
Indeed, whether feedback capacity exists at all,  is directly related to these conditions of stabilizability and detectability, and the structure of the matrix $Q \succeq 0$ entering the transmission cost function, plays a significant role, on the ability of the optimal feedback strategy $\{g^{B.1, *}(b_{i-1}): i=0, \ldots, \big\}$ given by (\ref{DP-UMCO_C13_IH_aaa}) to stabilize even unstable channels, that is, when the eigenvalues of channel matrix $C$ do not lie in the open unit disc of complex numbers.  \\
Appendix~\ref{appendix_B},  Theorem~\ref{lemma_vanschuppen} and Theorem~\ref{RIC},  summarize the implications of stabilizability and detectability on the optimal capacity achieving channel input distribution, and the corresponding ergodic properties of the optimal joint process $\{(A_i^{g,*}, B_i^{g,*}): i=0, \ldots, \big\}$ and the output process $\{B_i^{g,*}: i=0, \ldots, \big\}$, via   properties of algebraic and discrete-time recursive Lyapunov equations and Riccati equations.  \\
To apply Theorem~\ref{RIC} to the feedback capacity of Theorem~\ref{DP-UMCO_IH}, in oder to determine the properties of solutions to the algebraic Riccati equation), the following substitutions are invoked.
\begin{align}
A \longmapsto C^T, \hso C^T \longmapsto D, \hso B K_W B^T\tri G G^T \longmapsto s Q\tri G^T G, \hso N K_V N^T  \longmapsto s R. \label{sub_rel}
\end{align}

The following implications hold.
\begin{description}
\item[(i)] If the pair $\big(C, D\big)$ is stabilizable and the pair $\big(G, C\big)$ is detectable, and $s>0$ and finite, then by Theorem~\ref{RIC}, (a), (d) the deterministic part of the optimal feedback strategy ensures stability, thus establishing validity of   (\ref{cond_STA}), irrespectively of the eigenvalues of channel matrix $C$. By Appendix~\ref{appendix_B}, if the pair is $\big(C, D\big)$ is controllable then it is stabilizable and  if  pair is $\big(G, C\big)$ detectable then it is observable. 

\item[(ii)] If the conditions in (i) hold, and in addition $(C, K_V^{\frac{1}{2}}), K_V \tri K_V^{\frac{1}{2}}K_V^{\frac{1}{2}, T}$ is a controllable pair, by Theorem~\ref{lemma_vanschuppen}, (d),  then the Lyapunov matrix equation (\ref{MP_1_New_3}) has  a unique positive definite solution $K \succ 0$, which implies the channel output process $\{B_i^{g,*}: i=0, \ldots, \}$ has a unique invariant distribution.  

\end{description}

The next example further illustrates the importance of stabilizability and detectability conditions, in determining feedback capacity, and the role of zero matrix $Q =0$ versus $Q \neq 0$.\\  
 
 \begin{example}(Consequences of Theorem~\ref{DP-UMCO_IH}.\\
Consider the  feedback capacity given in Theorem~\ref{DP-UMCO_IH}. \\
{\bf (a) Scalar with $p=q=1, R=1, Q=0$}. This is discussed in Section~\ref{mo-dr}.  Specifically, (\ref{Exact_Sol_1})-(\ref{Exact_Sol_2}), are obtained from the expressions of Theorem~\ref{DP-UMCO_IH}, and this example demonstrates that whether feedback increases capacity depends on the channel parameters and transmission cost parameters $\{C, D, R, Q, K_V\}$. \\
{\bf (b) MIMO with  $Q=0$.} Since $Q=0$, the algebraic Riccati equation (\ref{DP-UMCO_C12_IH_new}) reduces to the following matrix equation. 
\begin{align}
 P=C^T P C-C^T P D\Big(D^T P D+s R\Big)^{-1} \Big(C^T P D\Big)^T \hso \Longrightarrow P=0 \hso \mbox{i.e., the zero matrix is one solution.} \label{DP-UMCO_C12_IH_new_Q}
\end{align}
It is shown next, that feedback capacity ${C}_{A^\infty \rar B^\infty}^{FB, B.1} (\kappa)\equiv J^{B.1,*}\Big|_{s=^*}$ depends on whether the eigenvalues lie inside the unit disc of the space of complex numbers ${\mb D}_o$, and whether feedback increases capacity is determined from the solutions of the algebraic Riccation equation.

{\bf (i) Case 1: MIMO Stable Channel}, $spec\big(C\big) \subset {\mb D}_o$.  Since $spec\big(C\big) \subset {\mb D}_o$ then $\big(G, C\big), Q \tri G^T G$ is detectable even though, $G=0$, because by Definition~\ref{def_st-de}, there exists an matrix $L$ such that $spec(C-LG) \subset {\mb D}_o$, i.e., take $L=0$. Similarly, $(C, D)$ is stabilizable.  By invoking Theorem~\ref{RIC}, (a) (with substitutions (\ref{sub_rel})) the Riccati matrix equation (\ref{DP-UMCO_C12_IH_new_Q}) with $P \succeq 0$ has at most one solution, and hence $P=0$ is the only solution. Substituting $P=0$ into the Lyapunov equation (\ref{MP_1_New_3}) and  (\ref{DP-UMCO_C13_IH_aa})-(\ref{DP-UMCO_C13_IH_aaa}) the following are obtained.
\begin{align}
\Gamma^*=&0, \hso  C^{B.1}(b)=0, \hso 
J^{B.1,*}=   \sup_{ K_{Z}\in  {\mb S}_+^{q\times q}}\Big\{  \frac{1}{2} \log \frac{ | D K_{Z} D^T+K_{V}|}{|K_{V}|}+s\kappa- tr\Big(s\; R K_{Z}\Big)\Big\},  \label{DP-UMCO_New}  \\
K=&C   K C^T +D K_{Z} D^T +K_{V}. \label{MP_1_New_3_Q}
\end{align}
Recall that $K$ is the covariance of the channel ouput process $\{B_i^*: i=0, \ldots, ..\}$. By Theorem~\ref{lemma_vanschuppen}, (d),    if $K_V$ is full rank, then 1), 2) imply $K\succ 0$. Further, by Theorem~\ref{lemma_vanschuppen}, (b), $K \succ 0$  is the unique solution of (\ref{MP_1_New_3_Q}), and hence the channel output process $\{B_i^*: i=0, \ldots, \}$ has a unique invariant distribution. \\
Finally, by  (\ref{DP-UMCO_New}) and the fact that  $``s''$ correspond to  the Lagrange multiplier of the transmission cost constraint, then the following holds. 
\begin{align}
{C}_{A^\infty \rar B^\infty}^{FB, B.1} (\kappa)= 
\sup_{ K_{Z}\in  {\mb S}_+^{q\times q}: tr\big( R K_{Z}\big)\leq \kappa   }  \frac{1}{2} \log \frac{ | D K_{Z} D^T+K_{V}|}{|K_{V}|}={C}_{A^\infty \rar B^\infty}^{noFB, B.1} (\kappa) .  \label{DP-UMCO_New_1}  
\end{align}
where ${C}_{A^\infty \rar B^\infty}^{noFB, B.1} (\kappa)$ is the capacity of (\ref{LCM-A.1_a_TI}),  (\ref{LCM-A.1_aa_TI}) (with $Q=0$) without feedback. \\
Moreover, from (\ref{DP-UMCO_New_1}) it follows that the capacity achieving channel input distribution without feedback is stationary (even if $\{Z_i: i=0, 1\ldots\}$ is not restricted to a stationary process,  and satisfies conditional independence
\bea
P_{A_i|A^{i-1}}^*(da_i|a^{i-1})= P_{A_i}^*(da_i), \hso i=0,1 \ldots,  
 \eea
 The above discussion generalizes  the scalar example discussed in Section~\ref{mo-dr}, (\ref{Exact_Sol_1})-(\ref{Exact_Sol_2}), to MIMO channels. 
 
{\bf (ii) Case 2: MIMO Unstable Channel}, $spec\big(C\big) \in {\mb D}_o^c \tri \big\{c \in {\mb C}: |c|> 1\big\}$.  
 For unstable channels, $\big(G, C\big), Q\tri G^T G$ is not detectable (i.e., since $Q=0$ and $C$ is unstable), hence condition (\ref{DET}) is violated. However, even if detectability is violated, by Theorem~\ref{lemma_vanschuppen}, d) if the pair $(C, K_{V}^{\frac{1}{2}})$ is controllable and Lyapunov equation (\ref{MP_1_New_3}) has a positive definite solution $K \succ 0$, then (\ref{cond_STA}) holds, that is, the feedback optimal strategy is stabilizing, i.e., $spec\Big(C+D \Gamma^*\Big)$, and if in addition $K_V\succ 0$, by Theorem~\ref{RIC}, (e), the matrix Riccati equation has a unique solution,  which is neccesary $P \succ 0$ (otherwise is not stabilizing). 
  
 \end{example}

The above example illustrates the link between LQG stochastic optimal control theory, and feedback capacity of G-LCM-B.1.

\subsubsection{\bf Feedback Capacity of G-LCM-B. \& Infinite Horizon LQG Theory}
Consider the G-LCM-B.J defined in Section~\ref{opt_ex-LCM-B.1}. Then  Theorem~\ref{DP-UMCO_IH} is easily  generalized to G-LCM-B.J; this is left to the reader.

\ \

\begin{remark}(Generalization and Relations to Cover and Pombra \cite{cover-pombra1989} AGN Channel)\\
\label{R-cp1989_2}
(a)  It is possible to derive analogous results for the time-invariant version of the G-LCM-A and G-LCM.B.J, by invoking the  formulation in Section~\ref{opt_ex-LCM-B.1},  Section~\ref{G-LCM-A-LQG}.     \\
(b) The per unit time limiting version  of the  (scalar) Cover and Pombra \cite{cover-pombra1989} AGN channel model, whose  characterization of FTFI capacity is  given by  (\ref{c-p1989}),  is extensively analyzed  by Kim in \cite{kim2010},  under the assumptions that $\{(A_i, B_i): i=0, \ldots, \}$ are jointly stationary ergodic, and the noise process is stable,  stationary, and of limited memory  (i.e., its power spectral density does not have unstable poles).  To this date no multidimensional examples are found in the literature, for which the  capacity expression of Cover and Pombra \cite{cover-pombra1989} or its limited memory noise version \cite{kim2010}, are computed explicitly.   \\ 
The material of this section, illustrate that for MIMO TI-G-LCMs, by invoking stochastic optimal control theory,  the characterizations of  feedback capacity can be computed. Moreover,  detectability and stabilizability  are sufficient conditions, for   the optimal channel input distribution to induce an invariant distribution for  the joint process $\{(A_i^g, B_i^g): i=0,1, \ldots, \}$ and its marginals. This illustrates the direct connection between ergodic LQG stochastic optimal control theory and feedback capacity. 
\end{remark}

\ \

\section{Relations Between Characterizations of FTFI Capacity and Coding Theorems}
\label{c-thms}
In this section the importance of  the characterizations of FTFI capacity 
are illustrated, with respect to the derivations of the  converse and the direct part of the channel coding theorems. Moreover, sufficient conditions are identified so that the per unit time limits of the characterizations of FTFI capacity, corresponds to feedback capacity. \\  
Consider  the following definition of a code.\\

\begin{definition}(Achievable rates of  codes  with feedback)\\
\label{ccode}
Given a channel distribution of Class A or B and a transmission cost function of Class A or B, an $\{(n, { M}_n, \epsilon_n):n=0, 1, \dots\}$  code with feedback 
consists of the following.\\
(a)  A set of messages ${\cal M}_n \tri \{ 1,  \ldots, M_n\}$ and a set of encoding maps,  mapping source messages  into channel inputs of block length $(n+1)$, defined by
\begin{align}
{\cal E}_{[0,n]}^{FB}(\kappa) \tri & \Big\{g_i: {\cal M}_n \times {\mb B}^{i-1}  \longmapsto {\mb A}_i, \hso  a_0=g_0(w, b^{-1}), a_i=e_i(w, b^{i-1}),\hso  w\in {\cal M}_n,      i=0,1, \ldots, n:\nonumber \\
& \frac{1}{n+1} {\bf E}^g\Big(c_{0,n}(A^n,B^{n-1})\Big)\leq \kappa  \Big\}. \label{block-code-nf-non}
\end{align}
The codeword for any $w \in {\cal M}_n$  is $u_w\in{\mb A}^n$, $u_w=(g_0(w,b^{-1}), g_1(w, b^0),
,\dots,g_n(w, b^{n-1}))$, and ${\cal C}_n=(u_1,u_2,\dots,u_{{M}_n})$ is  the code for the message set ${\cal M}_n$. Note that in general, the code for Class A channels (and also Class B channels) depends on the initial data $B^{-1}=b^{-1}$. However, if  it can be shown that in  the limit, as $n \longrightarrow \infty$, the induced channel output process  has a unique invariant distribution then the code is independent of the initial data.  \\
(b)  Decoder measurable mappings $d_{0,n}:{\mb B}^n\longmapsto {\cal M}_n$, ${Y}^n= d_{0,n}(B^{n})$, such that the average
probability of decoding error satisfies 
\bea
{\bf P}_e^{(n)} \tri \frac{1}{M_n} \sum_{w \in {\cal M}_n} {\mb  P}^g \Big\{d_{0,n}(B^{n}) \neq w |  W=w\Big\}\equiv {\mathbb P}^g\Big\{d_{0,n}(B^n) \neq W \Big\} \leq \epsilon_n
\eea
where  $r_n\tri \frac{1}{n+1} \log M_n$ is  the coding rate or transmission rate (and the messages are uniformly distributed over ${\cal M}_n$). \\  
A rate $R$ is said to be an achievable rate, if there exists  a  code sequence satisfying
$\lim_{n\longrightarrow\infty} {\epsilon}_n=0$ and $\liminf_{n \longrightarrow\infty}\frac{1}{n+1}\log{{M}_n}\geq R$. The feedback capacity is defined by $C\tri \sup \{R: R \: \mbox{is achievable}\}$.
\end{definition}

\ \

With respect to the above definition of a code or variants of it,  direct and converse coding theorems are derived in  \cite{cover-pombra1989,kim2008,permuter-weissman-goldsmith2009,tatikonda-mitter2009,kim2010}. These are can be  separated into  those which treat  Gaussian channels with memory, and those which treat finite alphabet spaces. \\
The  coding theorems in  \cite{cover-pombra1989,kim2008,permuter-weissman-goldsmith2009,tatikonda-mitter2009,kim2010} are directly applicable to channels of Class A or B and transmission cost functions of Class A or B,   provided, the assumptions based on which these  are derived, are adopted, or they are  modified to account for additional generalities. For example, the coding theorems derived by Cover and Pombra \cite{cover-pombra1989} for scalar nonstationary nonergodic AGN channels with memory, are directly applicable to the G-LCM-A 
presented in Section~\ref{exa_CG-NCM-A.1}.  
For finite alphabet spaces $\{{\mathbb A_i}={\mathbb A}, {\mathbb B}_i={\mathbb B}:i=0, \ldots, n\}$, the  coding theorems derived by Kim \cite{kim2008}, for the class of stationary channels with feedback,   are directy applicable to NCM-A and NCM-B (without transmission cost), given in Definition~\ref{exa_A_D}, and they can be extended to include transmission cost constraints.  The coding theorem derived by Chen and Berger \cite{chen-berger2005} for the UMCO (i.e., $\{{\bf P}_{B_i|B_{i-1}, A_i}: i=0, \ldots, n\}$) with finite alphabet spaces, is directly applicable, while a transmission cost function of Class B with $K=1$ can be  easily incorporated. The  various  coding theorems derived by Permuter, Weissman, and Goldsmith in  \cite{permuter-weissman-goldsmith2009,kim2010} for finite alphabet spaces  (without transmission cost constraints), under the assumption of time-invariant deterministic feedback,  are directly applicable to channel distributions of Class A or B, and since their method is based on irreducibility of the channel distribution, they also extend to problems with time-invariant transmission cost functions of Class A or B. 

However, 
the converse part of the coding theorem is based on establishing a tight upper bound on any achievable rate, and  the characterization of FTFI capacity gives such a tight bound. For example, if the channel is memoryless the tight upper bound on any rate is given by the two-letter expression $C\tri \max_{{\bf P}_A}I(A; B) =  \lim_{n \longrightarrow \infty} \max_{{\bf P}_{A^n}} \frac{1}{n+1}I(A^n; B^n)$, and   the characterizations of FTFI capacity are the analogs for channels with memory. On the other hand, the direct part is often shown by using random coding arguments, which implies codes are generated independently according to the distribution, which maximizes mutual information if there is no feedback, and directed information if there is feedback. The optimal channel input distribution, which corresponds to the characterization of FTFI capacity is the one to be used in code generation.     For example, if the channel is memoryless codes are generated independently according to ${\bf P}_{A^n}^*(da^n)\tri \otimes_{i=0}^n {\bf P}_{A}^*(da_i)$, where ${\bf P}_A^*$ is the one which corresponds to $C$.  

For the converse to the coding theorem, it is sufficient to identify conditions for  existence of  the optimal channel input distribution corresponding to the characterizations of the FTFI capacity, and convergence of the  per unit time limiting version (provided condition (\ref{CI_Massey}) holds). \\
For the direct part of the coding theorem,   for any channel distribution  of Class A (resp. Class B), and transmission cost of Class A or B, with corresponding information density and transmission cost, evaluated at the optimal channel input distribution,   $\big\{\pi^*(da_i|b^{i-1}): i=0,1,\ldots, n\big\} \in  \overline{\cal P}_{[0,\infty]}^{A} \bigcap {\cal P}_{[0,n]}(\kappa)$ (resp.  $\sr{\circ}{\cal P}_{[0,\infty]}^{B.J} \bigcap {\cal P}_{[0,n]}(\kappa)$) (i.e., obtained in   Theorem~\ref{thm-ISR} (resp. Theorem~\ref{cor-ISR_C4})), it is sufficient to identify conditions for information  stability  in the sense of Dobrushin \cite{pinsker1964}.
 Information stability  implies the asymptotic equipartition property (AEP) of directed information holds, from which the direct part of the coding theorem follows by standard arguments. 

The  coding theorem stated below, is generic, in the sense that sufficient conditions are imposed, to ensure both the converse part and direct part of coding theorem hold.\\

%

\begin{theorem}(Coding theorem)\\
\label{lemma:DPI-cc}
Consider any channel distribution and transmission cost function of  Class A or B,  with corresponding characterizations of FTFI capacity, and optimal channel input conditional  distributions denoted by   $\big\{\pi_i^*(da_i| {\cal I}_i^P): i=0,1,\ldots, n\big\} \in  {\cal P}_{[0,n]}(\kappa)$, where  $\{{\cal I}_i^P: i=0, \ldots, n\}$ is the information structure of optimal channel input distribution.  \\
 Suppose the following two conditions hold.\\
i) Conditional independence  (\ref{CI_Massey}) holds;\\
ii) there exists an optimal channel input conditional distribution,  which achieves the characterization of the FTFI capacity, and its per unit time limit exist (if not replace it by $\liminf$).\\
Define the following.\\
For $\big\{\pi_i^*(da_i| {\cal I}_i^P): i=0,1,\ldots, n\big\} \in  {\cal P}_{[0,n]}(\kappa)$ (assuming condition ii)), the directed information density is called stable, if  
\begin{align}
 \lim_{n \longrightarrow \infty} {\bf P}^{\pi^*} \Big\{(A^n ,B^n) \in {\mb A}^n \times {\mb B^n}:\frac{1}{n+1} \Big|{\bf E}^{\pi^*} \big\{ {\bf i}^{\pi^*}(A^n, B^n)\big\} - {\bf i}^{\pi^*}(A^n, B^n) \Big| > \varepsilon \Big\} =0,  \label{IS-O_1}
\end{align}
where for channel distribution of Class A, and transmission cost of Class A or B, the directed information density is  ${\bf i}^{\pi^*}(A^n, B^n)\tri \sum_{i=0}^n \log\Big(  \frac{{\bf P}(\cdot|B^{i-1}, A_i)}{{\bf P}^{\pi^*}(\cdot|B^{i-1})}(B_i)\Big)$, $i=0, \ldots, n$ (and similarly for Clannels of Class B). \\
For  $\big\{\pi_i^*(da_i| {\cal I}_i^P): i=0,1,\ldots, n\big\} \in  {\cal P}_{[0,n]}(\kappa)$ (assuming condition ii)), the transmission cost  is called  stable, if  
  \begin{align}
 \lim_{n \longrightarrow \infty} {\bf P}^{\pi^*} \Big\{ (A^n ,B^n) \in {\mb A}^n \times {\mb B^n}: \frac{1}{n+1} \Big| {\bf E}^{\pi^*} \big\{ c_{0,n}(A^n, B^n)\big\} - c_{0,n}(A^n, B^n) \Big| > \varepsilon \Big\} =0. \label{IS-O_2}
\end{align}
Then the  following hold.\\
(a) (Converse) If conditions i), ii) hold, then any achievable rate $R$ of  codes with feedback  given in  Definition~\ref{ccode},   satisfies the following inequalities.
\begin{align} 
 R \leq & \liminf_{n \longrightarrow\infty}\frac{1}{n+1}\log{{M}_n} 
{\leq} \:  \liminf_{n \longrightarrow\infty}  \sup_{ \big\{g_i(\cdot, \cdot): i=0, \ldots, n\big\} \in {\cal E}_{[0,n]}^{FB}(\kappa)}  \frac{1}{n+1}  \sum_{i=0}^n I(A_i; B_i| B^{i-1})
 \label{CCIS_6a_BC}\\
\leq  \: & \liminf_{n \longrightarrow \infty} \sup_{ \big\{ \pi(a_i| {\cal I}_i^P): i=0,1,\ldots, n\big\} \in {\cal P}_{[0,n]}(\kappa)}  \frac{1}{n+1}\sum_{i=0}^n I(A_i; B_i|B^{i-1}) \equiv C_{A^\infty \rar B^\infty}^{FB}(\kappa) \label{CCIS_8_BC}  
\end{align}
(b) (Direct) If  conditions i), ii) hold and in addition \\
iii) the directed information density is stable,\\
iv) the  transmission cost is stable, \\
then  
any rate $R < C_{A^\infty \rar B^\infty}^{FB}(\kappa)$ is achievable.
\end{theorem} 
\begin{proof} 
(a) Condition i) implies the well-known data processing inequality, while condition ii) implies  existence of the optimal channel input distribution and finiteness of the corresponding characterizations of the FTFI capacity and its per unit time limit. Hence, the statements of inequalities follow by applying  Fanon's inequality. The derivation is also found in many references, (i.e.,\cite{cover-pombra1989,kim2008,permuter-weissman-goldsmith2009,tatikonda-mitter2009,kim2010}). \\
(b) This is standard, because conditions ii)-iv) are sufficient to ensure the AEP holds, and hence   standard random coding arguments hold   (i.e., following  Ihara \cite{ihara1993}, by replacing the information density of mutual information by that of   directed information).  
\end{proof}

It is noted that alternative achievability theorems can be obtained by combining the achievability theorem  derived by Permuter, Weissman and Goldsmith \cite{permuter-weissman-goldsmith2009}, which is based on bounding the error of Maximum Likelihood (ML) decoding, and the characterizations of FTFI capacity and feedback capacity. \\
Finally, the following are  noted. 

(a) For the  TI-G-LCM-B.1, Theorem~\ref{DP-UMCO_IH} gives sufficient conditions, in terms of  the channel variables $\big\{C, D, R, Q, K_{V}\}$, expressed in terms of detectability and stabilizability, for $J^{B.1,*}\Big|_{s=s^*} = {C}_{A^\infty \rar B^\infty}^{FB, B.1} (\kappa)$,     defined by    (\ref{DP-UMCO_C13_IH_aaa}), to correspond to Feedback Capacity, irrespectively of whether the channel is stable or unstable.  

(b) For the  TI-G-LCM-B.J, similarly to Theorem~\ref{DP-UMCO_IH},   sufficient conditions can be obtained, for the corresponding solution of the dynamic programming, denoted by  $J^{B.J,*}\Big|_{s=s^*} = {C}_{A^\infty \rar B^\infty}^{FB, B.J} (\kappa)$  to correspond to Feedback Capacity.

(c) For Multidimensional Gaussian sources to be encoded and transmitted over any one of the channels,  G-LCM-A, G-LCM-B.1, G-LCM-B.J, coding  strategies can be  constructed, which   achieve the corresponding characterizations of the FTFI capacity, and Feedback capacity. 
This construction is a subject for further research.

\section{Conclusion}
In this second part of the two-part investigation, the information structures of optimal channel input conditional distributions of the first part, are applied to derive alternative characterizations of FTFI capacity, based on randomized information lossless strategies, driven by uniform RVs.  Their per unit time limiting versions are analyzed, without imposing \'a priori assumptions, which rule out the dual role of such strategies, to achieve the FTFI capacity characterizations and feedback capacity, to control the channel output process and to transmit new information through the channel.\\
The characterizations of FTFI capacity and feedback capacity are investigated for  application examples of MIMO Gaussian Linear Channel Models (G-LCMs) with memory. In such application examples, information lossless strategies decompose into a deterministic part, which corresponds to the  control process, and a random part, which corresponds to an innovations process.  Via this decomposition a separation principle is established;  the deterministic control part is shown to be directly  related to the role of optimal control strategies of Linear-Quadratic-Gaussian control theory, to control output processes, and, in general, to the feedback control theory of linear stochastic systems, while the random or innovations part is shown to be directly related to role of  encoders to achieve capacity, by  transmitting  new information over the channel.   Moreover, whether feedback increases capacity is shown to be directly linked to the role of the deterministic part of the information lossless randomized strategies, to control the channel output process. 


\begin{appendix}
\label{appendix}
\subsection{Proof of Theorem~\ref{thm-AC-1}.}
\label{appendix_thm-AC-1}
(a) By Theorem~\ref{thm-ISR}, (b),   the optimal channel input distributions belong to     $\overline{\cal P}_{[0,n]}^{A}=\{ \pi_i(da_i| b^{i-1})\equiv {\bf P}_{A_i| B^{i-1}}(a_i |b^{i-1}) : i=0,1, \ldots, n\}$, and  satisfy the transmission cost constraint. By Lemma~\ref{lemma-gs1989}, { CON.A.(1)} holds. Moreover, by the channel distribution, Assumption A.(iii), the property of optimal channel input distribution, and by virtue of  (\ref{LCM-A.1_5_new}),  { CON.A.(2)}, also holds.  Clearly,  ii) and iii) imply the processes $\{U_i: i=0, \ldots, n\}$ and $\{V_i: i=0, \ldots, n\}$ are independent. \\
(b) This is a direct consequence of (a), because the set of all channel input distribution $\overline{\cal P}_{[0,n]}^{A}$   is realized by  strategies $\{e_i^{A}(\cdot, \cdot): i=0,1, \ldots, n\}$. 
 \\
(c) To ensure no information is lost,   when the above randomized strategies $\{e_i^{A}(\cdot, \cdot): i=0,1, \ldots, n\}$ are used in the characterization of the FTFI capacity,  consider the  restricted class of randomized strategies defined by (\ref{alt_A.1_NCM}).
To show the set ${\cal E}_{[0,n]}^{IL-A} (\kappa)$ is information lossless, with respect to the characterization of FTFI capacity, recall that for 
channels of Class A, $I(A^i; B_i|B^{i-1})=I(A_i; B_i|B^{i-1}), i=0, \ldots, n,$ as defined by  (\ref{CIS_18}).  By an application of Theorem 3.7.1 in Pinsker~\cite{pinsker1964}  (and Corollary following it), it can be verified that 
the  following sequence of identities hold.
\begin{align}
I(A_i; B_i|B^{i-1})=I(A_i, U_i; B_i|B^{i-1})=& I(U_i;B_i|B^{i-1}), \hso i=0, \ldots, n \label{IL_A}\\
& \mbox{if and only if} \hso \{e_i^{A}(\cdot, \cdot): i=0,\ldots, n\} \in {\cal E}_{[0,n]}^{IL-A} (\kappa). \nonumber    
\end{align}
Hence, the class of randomized strategies ${\cal E}_{[0,n]}^{IL-A} (\kappa)$ is information lossless with respect to directed information, in the sense of identity (\ref{IL_A}). Further,  utilizing the definition of information lossless strategies ${\cal E}_{[0,n]}^{IL-A} (\kappa)$, the  alternative characterization of FTFI capacity is  obtained,  where (\ref{CMNon-A.1_10_new}) is due to the definition of ${\cal E}_{[0,n]}^{IL-A} (\kappa)$, (\ref{CMNon-A.1_10_new_o}) is by definition,   and (\ref{LCMNon-A.1_4_new}) is due to  (\ref{LCM-A.1_5_new}), that is, $U_i$ is independent of $B^{i-1}$ (i.e., a consequence of CON.A.(2).(ii)), for $i=0,1,\ldots, n$. This completes the prove.

\subsection{Proof of Theorem~\ref{G-LCM-A-CA}.}
\label{appendix_G-LCM-A-CA}
(a) By { Assumption A.1.(i)},  the conditional  distribution of the channel  is Gaussian, given by
\begin{align}
{\mb P}\Big\{B_i \leq b_i \Big| B^{i-1}=b^{i-1}, A^i=a^i\Big\}=& {\mb P}\Big\{V_{i}\leq b_i - \sum_{j=0}^{i-1} C_{i,j} b_{j} -D_{i,i} \; a_i\big)  \Big\},  \label{LCM-A.1_1} \\
\sim & N\Big(( \sum_{j=0}^{i-1}C_{i,j} b_j+ D_{i,i}  a_i ), K_{V_i}\Big),  \hso  i=0,1, \ldots, n. \label{CG-LCM-A.1_1} 
\end{align}
The transition probability distribution of $\big\{B_i: i=0,\ldots, n\big\}$ is given by
\begin{align}
{\mb P}\Big\{B_i \leq b_i \Big| B^{i-1}=b^{i-1}\Big\}=& \int_{{\mathbb A}_i} {\mb P}\Big\{V_{i}\leq b_i - \sum_{j=0}^{i-1} C_{i,j} b_{j}- D_{i,i} \; a_i  \Big\} \pi_i(da_i|b^{i-1}),  \hso  i=0,1, \ldots, n. \label{LCM-A.1_4} 
\end{align}
In view Assumption A.(iii), and properties of conditional entropy, then $H(B_i|B^{i-1}, A_i)= H(V_i|B^{i-1}, A_i)=H(V_i), i=0, \ldots, n$, and  directed information  is given by 
\begin{align}
I(A^n \rar B^n)=\sum_{i=0}^n \Big\{ H(B_i|B^{i-1})- H(B_i|B^{i-1}, A_i)\Big\}=  \sum_{i=0}^n H(B_i|B^{i-1}) -\sum_{i=0}^n H(V_i) . \label{CG-LCM-A.1_1_A}
\end{align}
Hence, the  characterization of FTFI Feedback Capacity is given by the following expression.
\begin{align}
{C}_{A^n \rar B^n}^{FB,A} (\kappa) 
 \tri &  \sup_{\big\{ {\pi}_i(da_i | b^{i-1}),  i=0,\ldots,n: \frac{1}{n+1} \sum_{i=0}^n  {\bf E} \big\{ \langle A_i, R_{i,i} A_i \rangle + \langle B^{i-1}, Q_{i}(i-1) B^{i-1}\rangle  \big\}\leq \kappa   \big\}    } H(B^{n})- H(V^n).    \label{CG-LCM_A.1_3_new}
\end{align}
From (\ref{CG-LCM-A.1_1_A}), it follows, that   the optimal channel input conditional distribution is Gaussian, with non-zero mean and conditional covariance which is independent of channel outputs. This will be verified using the characterization of FTFI capacity and the alternative equivalent alternative characterization.
  By  the entropy maximizing property of the Gaussian distribution the right hand side of  (\ref{CG-LCM-A.1_1_A}) is bounded above by the inequality $H(B^n) \leq H(B^{g,n})$, where $B^{g,n}\tri \{B_i^g: i=0,1, \ldots, n\}$ is jointly Gaussian distributed, and the average transmission cost constraint is satisfied. By (\ref{LCM-A.1}) the output process is jointly Gaussian if and only if $\big\{A_i, B_i, V_i: i=0, \ldots, n\big\}$ are jointly Gaussian.
 The upper bound is achieved if and only if   the channel input distribution is   Gaussian, denoted by  $\{ \pi_i^g(da_i|b^{i-1})\equiv {\bf P}_{A_i| B^{i-1}}^g(a_i| b^{i-1}) : i=0,1, \ldots, n\}$, having conditional mean which is a linear combination of $\big\{B_i: i=0, \ldots, n-1\big\}$, conditional covariance which is independent of the channel output process,  and the average transmission cost in (\ref{LCM-A.1}) is satisfied. Hence, the upper bound is achieved if and only if  the channel  output conditional distribution denoted by  $\{ \Pi_i^{\pi^g}(db_i|b^{i-1})\equiv {\bf P}_{B_i|B^{i-1}}^g(db_i|b^{i-1}): i=0,1, \ldots, n\}$ is also  Gaussian, with conditional mean which  is a linear combination of $\big\{B_i: i=0, \ldots, n-1\big\}$ and conditional covariance which is independent of the channel output process. Finally, by the linearity of the model (\ref{LCM-A.1}),  the upper bound is achieved if and only if  the channel input process is  Gaussian, denoted by $\big\{A_i^g: i=0, \ldots, n\big\}$.   Hence,   (\ref{CG-LCM_A.1_3}) is obtained.\\
(b) The above conclusion can be established via the alternative characterization given by (\ref{CMNon-A.1_10_new})-(\ref{LCMNon-A.1_4_new}), as follows. 
Since the optimal channel input distribution is $\big\{  {\bf P}_{A_i| B^{i-1}}(a_i| b^{i-1}) : i=0,1, \ldots, n\big\}$,  by Lemma~\ref{lemma-gs1989} there  exists a  measurable function $e_i^{A}: {\mathbb B}^{i-1} \times {\mathbb U}_i \longmapsto {\mathbb A}_i, {\mb U}_i \tri [0,1],  a_i= e_i^{A}(b^{i-1}, u_i),  i=0,1, \ldots, n$ such that 
\begin{align}
 {\bf P}_{U_i} \big( U_i : e_i^{A}(b^{i-1}, U_i)  \in da_i \big) = {\bf P}_{A_i| B^{i-1}}(da_i| b^{i-1}),  \hso  i=0,1, \ldots, n.\label{LCM-A.1_5}
 \end{align}
Substituting the randomized  strategy into the channel model (\ref{LCM-A.1}),  then 
\begin{align}
&B_i   = \sum_{j=0}^{i-1} C_{i,j} B_{j} + D_{i,i} \; e_i^{A}(B^{{i-1}}, U_i) + V_{i}, \hso i= 1, \ldots, n, \\
&{\cal E}_{[0,n]}^{IL-G-A} (\kappa) \tri \Big\{ e_i^{A}(b^{i-1}, u_i), i=0, \ldots, n: \hso \mbox{for a fixed $b^{i-1}$ the function } \;  \: e_i^{A}(b^{i-1}, \cdot)  \hso \mbox{is one-to-one} \nonumber \\
&\mbox{ and onto ${\mb A}_i$ for $i=0,\ldots, n$}, \hso  \frac{1}{n+1}\sum_{i=0}^n {\bf E}^{e^{A}} \big\{\langle e_i^{A}(B^{i-1}, U_i),R_{i,i} e_i^{A}(B^{i-1}, U_i)\rangle+ \langle B^{i-1}, Q_{i}(i-1) B^{i-1} \rangle \big\}    \leq \kappa         \Big\}. \label{LCM-A.1_6}
\end{align} 
By  the entropy maximizing property of the Gaussian distribution the right hand side of  (\ref{CG-LCM-A.1_1_A}) (with $\{a_i=e_i^A(b^{i-1},u_i): i=0, \ldots, n\}$) is bounded above by the inequality\footnote{The superscript indicates the distribution depends on the strategy $\{e_i^A(\cdot): i=0, \ldots, n\}$.} $H^{e^A}(B^n) \leq H^{e^A}(B^{g,n})$, where $B^{g,n}\tri \{B_i^g: i=0,1, \ldots, n\}$ is jointly Gaussian distributed.
However,  since  linear combination of any sequence of  RVs is Gaussian distributed  if and only if the sequence of RVs is also jointly Gaussian distributed, then necessarily the functions $\{e_i^{A}(\cdot,\cdot): i=0,1, \ldots, n\}$ are linear combinations of Gaussian RVs. Hence,
\bea
 a_i= e_i(b^{i-1}, u_i)\equiv \overline{e}_i^{A}(b^{i-1}, g_i(u_i)),\hso g_i: {\mathbb U}_i \longmapsto {\mathbb Z}\tri {\mathbb R}^q, \hso z_i=g_i(u_i),  \hso i=0,1, \ldots, n
\eea 
where $\overline{e}_i^A(b^{i-1}, z_i)$ is a linear combination of $(b^{i-1}, z_i)$ for $i=0, \ldots, n$, and such  $\big\{g_i(\cdot): i=0, \ldots, n\big\}$ always exist. Hence, the upper bound is achieved if and only if 
  $A^n = A^{g,n}\tri \{A_i^g: i=0,1, \ldots, n\}$ is  Gaussian distributed satisfying the average transmission cost constraint, and $\{Z_i=g_i(U_i): i=0,1, \ldots, n\}$ is a Gaussian sequence.  Hence, the alternative characterization of the FTFI capacity is given by (\ref{EXTR_A})-(\ref{LCM-A.1_8}), 
where the independence properties follow from   Assumption~A.(iii) and  the information structure of the maximizing channel input distribution, $\Big\{ {\bf P}_{A_i|A^{i-1}, B^{i-1}}(a_i| a^{i-1}, b^{i-1}) ={\bf P}_{A_i| B^{i-1}}^g(a_i| b^{i-1})-a.s. : i=0,1, \ldots, n\Big\}$, or CON.A.(2), and (\ref{G-lCM-A-CR}) is easily verified.

\subsection{Proof of Theorem~\ref{DP-UMCO}}
\label{appendix_C}
(a)  (\ref{LCM_B.1_3_C30}), (\ref{LCM_B.1_3_C40}), follows directly from the re-formulation of the problem.\\
 (b)  Clearly, (\ref{LCM_B.1_3_C50}) is the cost-to-go for (\ref{LCM_B.1_3_C30}). \\(c) The dynamic programming recursions 
follow directly from (\ref{LCM_B.1_3_C50}) \cite{vanschuppen2010,kumar-varayia1986}. \\
(d)-(e) The rest of the statements are obtained by solving the 
 dynamic programming equations, as done for  LQG stochastic optimal control problems \cite{vanschuppen2010}, with some modifications to account for the fact that the strategies are randomized (instead of deterministic). Let $C_n^{B.1}(b_{n-1})=-\langle b_{n-1}, sQ_{n,n-1}b_{n-1}\rangle +r(n) $,  $P(n)=s Q_{n,n-1}$, and $r(n)$ given by (\ref{DP-UMCO_C20}). It can be verified this is indeed the solution at the last stage of the dynamic programming recursions, i.e., (\ref{DP-UMCO_C10}), and that $g_n^{B.1,*}(b_{n-1})=0$. Then $P(n)=P^T(n)\geq 0$. Suppose for $j=i+1, i+2, \ldots, n$, $P(j)=P^T(j)\geq 0$, $C_j^{B.1}(b_{j-1})=-\langle b_{j-1}, P(j)b_{j-1}\rangle + r(j)$. It will be shown that $P(i)=P^T(i)\geq 0$, $C_i^{B.1}(b_{i-1})=-\langle b_{i-1}, P(i)b_{i-1}\rangle + r(i)$, as stated in (d), (e).\\ The following calculations follow directly from Assumptions B.1.(i) (i.e., ${\bf E}^{g^{B.1}}\Big\{Z_i\Big|B_{i-1}\Big\}=0, {\bf E}^{g^{B.1}}\Big\{V_i\Big|B_{i-1}\Big\}=0$, and $Z_i$ independent of $V_i$).
\begin{align}
&-s \Big\{\langle u_i, R_{i,i} u_i\rangle + \langle b_{i-1}, Q_{i,i-1} b_{i-1} \rangle\Big\} + {\bf E}^{g^{B.1}}\Big\{ C_{i+1}^{B.1}(B_{i}^g)\Big| B_{i-1}^g=b_{i-1}\Big\} \\
&=-s \Big\{\langle u_i, R_{i,i} u_i\rangle + \langle b_{i-1}, Q_{i,i-1} b_{i-1} \rangle\Big\} + {\bf E}^{g^{B.1}}\Big\{ C_{i+1}^{B.1}(C_{i,i-1} B_{i-1}^g+ D_{i,i} U_i^g+D_{i,i}Z_i+ V_i)\Big| B_{i-1}^g=b_{i-1}\Big\} \\
&=-\left[ \begin{array}{c} b_{i-1}\\ u_i \end{array} \right]^T  \left[ \begin{array}{cc} sQ_{i,i-1} & 0 \\ 0 &sR_{i,i}\end{array} \right]   \left[  \begin{array}{c} b_{i-1}\\ u_i \end{array} \right] +r(i+1) \nonumber \\
& - {\bf E}^{g^{B.1}}\Big\{\langle C_{i,i-1} B_{i-1}^g+ D_{i,i} U_i^g+D_{i,i}Z_i+ V_i, P(i+1) \Big(C_{i,i-1} B_{i-1}^g+ D_{i,i} U_i^g+D_{i,i}Z_i+ V_i\Big)\rangle   \Big| B_{i-1}=b_{i-1}\Big\}        \\
&= -\left[ \begin{array}{c} b_{i-1}\\ u_i \end{array} \right]^T  \left[ \begin{array}{cc} C_{i,i-1}^T P(i+1) C_{i,i-1} +sQ_{i,i-1} & C_{i,i-1}^T P(i+1) D_{i,i} \\ D_{i,i}^T P(i+1) C_{i,i-1} & D_{i,i}^T P(i+1) D_{i,i}+sR_{i,i}\end{array} \right]   \left[  \begin{array}{c} b_{i-1}\\ u_i \end{array} \right] +r(i+1)\nonumber \\
&- tr \Big(P(i+1) \Big[D_{i,i} K_{Z_i}D_{i,i}^T+ K_{V_i}\Big]\Big) \\
&= -\left[ \begin{array}{c} b_{i-1}\\ u_i \end{array} \right]^T  \left[ \begin{array}{cc} H_{11}(i) & H_{12}(i) \\ H_{12}^T(i) & H_{22}(i)\end{array} \right]   \left[  \begin{array}{c} b_{i-1}\\ u_i \end{array} \right] +r(i+1)
- tr \Big(P(i+1) \Big[D_{i,i} K_{Z_i}D_{i,i}^T+ K_{V_i}\Big]\Big) 
\end{align}
Note that 
\begin{align}
H_{11}(i)=H_{11}^T(i)\geq 0, \hso H_{22}(i)=H_{22}^T(i)= D_{i,i} P(i+1) D_{i,i}+sR_{i,i}\geq  s R_{i,i} > 0, \hso \mbox{if} \hso s>0. 
\end{align}
By the induction hypothesis and $R_{i,i} \in S_{++}^{q\times q}, Q_{i,i-1} \in S_{+}^{p\times p}$, the following hold.
\begin{align}
&  \sup_{  (u_i,K_{Z_i})\in {\mathbb R}^q\times S_+^{q\times q}  }   \Big\{ \frac{1}{2} \log \frac{ | D_{i,i} K_{Z_i} D_{i,i}^T+K_{V_i}|}{|K_{V_i}|}- tr\Big(s R_{i,i} K_{Z_i}\Big)  \nonumber \\
&- s \Big[ \langle u_i, R_{i,i} u_i\rangle + \langle b_{i-1}, Q_{i,i-1} b_{i-1} \rangle  \Big]   + {\bf E}^{g^{B.1}}\Big\{  C_{i+1}^{B.1}(B_{i}^g)       \Big| B_{i-1}^g=b_{i-1}\Big\}  \Big\}\\
&=  \sup_{  (u_i,K_{Z_i})\in {\mathbb R}^q\times S_+^{q\times q}  } \Big\{ \frac{1}{2} \log \frac{ | D_{i,i} K_{Z_i} D_{i,i}^T+K_{V_i}|}{|K_{V_i}|}-  tr\Big(s R_{i,i} K_{Z_i}\Big)- tr \Big(P(i+1) \Big[D_{i,i} K_{Z_i}D_{i,i}^T+ K_{V_i}\Big]\Big) \\
&-\left[ \begin{array}{c} b_{i-1}\\ u_i \end{array} \right]^T  \left[ \begin{array}{cc} H_{11}(i) & H_{12}(i) \\ H_{12}^T(i) & H_{22}(i)\end{array} \right]   \left[  \begin{array}{c} b_{i-1}\\ u_i \end{array} \right] +r(i+1)
 \Big\}\\
&= \sup_{ K_{Z_i}\in  S_+^{q\times q}  } \sup_{  u_i\in {\mathbb R}^q }\Big\{-\left[ \begin{array}{c} b_{i-1}\\ u_i+H_{22}^{-1}(i)H_{12}^T(i) b_{i-1}  \end{array} \right]^T  \left[ \begin{array}{cc} H_{11}(i)-H_{12}(i)H_{22}^{-1}(i)H_{12}^T(i) & 0 \\ 0 & H_{22}(i)\end{array} \right]  \nonumber \\
& \left[  \begin{array}{c} b_{i-1}\\ u_i+H_{22}^{-1}(i)H_{12}^T(i) b_{i-1}  \end{array} \right] +\frac{1}{2} \log \frac{ | D_{i,i} K_{Z_i} D_{i,i}^T+K_{V_i}|}{|K_{V_i}|}  - tr\Big(s R_{i,i} K_{Z_i}\Big)\Big]\nonumber \\
&- tr \Big(P(i+1) \Big[D_{i,i} K_{Z_i}D_{i,i}^T+ K_{V_i}\Big]\Big)+r(i+1)\Big\}
\end{align}
\begin{align}
&= \sup_{ K_{Z_i}\in  S_+^{q\times q}  }\Big\{- \langle b_{i-1}, [H_{11}(i)-H_{12}(i)H_{22}^{-1}(i)H_{12}^T(i) ]b_{i-1} \rangle +\frac{1}{2} \log \frac{ | D_{i,i} K_{Z_i} D_{i,i}^T+K_{V_i}|}{|K_{V_i}|}\nonumber \\
& - tr\Big(s R_{i,i} K_{Z_i}\Big)- tr \Big(P(i+1) \Big[D_{i,i} K_{Z_i}D_{i,i}^T+ K_{V_i}\Big]\Big)+r(i+1)\Big\}, \hso \mbox{if} \hso s>0, \hso \\
&\mbox{(because} \hso H_{22}(i) >0, \hso \mbox{and the optimal control is} \hso u_{i}= -H_{22}^{-1}(i) H_{12}^T(i) b_{i-1}), \\
&= - \langle b_{i-1}, [H_{11}(i)-H_{12}(i)H_{22}^{-1}(i)H_{12}^T(i) ]b_{i-1} \rangle + \sup_{ K_{Z_i}\in  S_+^{q\times q}  }\Big\{ \frac{1}{2} \log \frac{ | D_{i,i} K_{Z_i} D_{i,i}^T+K_{V_i}|}{|K_{V_i}|}\nonumber \\
& - tr\Big(s R_{i,i} K_{Z_i}\Big)- tr \Big(P(i+1) \Big[D_{i,i} K_{Z_i}D_{i,i}^T+ K_{V_i}\Big]\Big)+r(i+1)\Big\}\\
&=- \langle b_{i-1}, P(i) b_{i-1} \rangle + r(i), \hso \mbox{if} \hso (\ref{DP-UMCO_C12}), (\ref{DP-UMCO_C13}) \hso \mbox{hold.}
\end{align}
Note that  $P(\cdot)\tri H_{11}(\cdot)-H_{12}(\cdot)H_{22}^{-1}(i)H_{12}^T(\cdot)$ is precisely (\ref{DP-UMCO_C12}), and it is not affected by $\{K_{Z_i}: i=0, \ldots, n\}$, that is, the deterministic part of the strategy is separated from the random part of the strategy. For each $i=0, 1, \ldots, n-1$, define $g_i^{B.1,*}(\cdot)$ by 
\begin{align}
g_i^{B.1,*}(b_{i-1}) \tri &-H_{22}^{-1}(i)H_{12}^T(i)b_{i-1} 
=-\Big[D_{i,i}P(i+1) D_{i,i} +s R_{i,i}\Big]^{-1} D_{i,i}^T P(i+1)C_{i,i-1} b_{i-1}, \hso \mbox{then} \\
C_{i}^{B.1}(b_{i-1})= &  -\langle b_{i-1}, P(i) b_{i-1}\rangle + r(i), \hso \mbox{since $P(\cdot)$ does not depend on $K_{Z_i}$, then,}\\
=&-\langle b_{i-1}, P(i) b_{i-1}\rangle +s (n+1)\kappa \nonumber \\
&+ \sup_{ K_{Z_j}\in  S_+^{q\times q}: j=i, \ldots, n} \Big\{    \frac{1}{2}\sum_{j=i}^n \log \frac{ | D_{j,j} K_{Z_j} D_{j,j}^T+K_{V_j}|}{|K_{V_j}|}  -\sum_{j=i}^n  tr\Big(s\; R_{j,j} K_{Z_j}+P(j+1) \big[D_{j,j} K_{Z_j}D_{j,j}^T+ K_{V_j}\big]\Big) \Big\}.
\end{align}
Finally, since $P(\cdot)\tri H_{11}(\cdot)-H_{12}(\cdot)H_{22}^{-1}(i)H_{12}^T(\cdot)$,  the Riccati difference equation (\ref{DP-UMCO_C12_a}) is obtained. \\
 (f), (g) follow from the constraint and expression of cost-to-go. This completes the prove.

\subsection{Lyapunov  \& Riccati Equations of Gaussian Linear Stochastic Systems and LQG Theory}
\label{appendix_B}
In this section, some of the basic concepts of linear systems are introduced, and  fundamental theorems relating Lyapunov stability and Riccati equations to stability of linear systems are given. These are found in \cite{kumar-varayia1986,caines1988,vanschuppen2010}. \\ 
The open unit disc of the space of complex number ${\mathbb C}$, is defined by ${\mathbb D}_o \tri \big\{c \in {\mathbb C}: |c| <1\big\}$. The Spectrum of a matrix $A \in {\mathbb R}^{q \times q}$ (the set of all its eigenvalues), is denoted by $spec(A) \subset {\mathbb C}$. A  matrix $A \in {\mathbb R}^{q \times q}$ is called exponentially stable if all its eigenvalues are within the open unit disc, that is,  $spec(A) \subset {\mathbb D}_o$. \\
Consider the time-invariant representation of a finite-dimensional Gaussian system described by the following equations.  
\begin{align}
&X_{i+1}=A X_{i} + B W_{i}, \hso X_0=x_0, \hso i=1, \ldots, n, \label{LSSM_K1} \\
&Y_i =C X_i + N V_i, \hso i=0, \ldots, n \label{LSSM_K2}\\
& X_0 \in {\mathbb R}^q, \: X_0 \sim N(0,K_{X_0}), \hso W_i \sim N(0, K_W), W_i \in {\mb R}^k, \hso Y_i \in {\mathbb R}^p, \hso   V_i \sim N(0, K_V), V_i \in {\mathbb R}^m, i=0, \ldots, n \\
& \mbox{and $\{(W_i, V_i): i=0, \ldots, n\}$ mutual independent and  independent of $X_0$}. \label{LSSM_K3}
\end{align}
Answers to questions of convergence of covariance matrices and existence of  invariant distribution  of the joint process $\{(X_i, Y_i): i=0, \ldots, n\}$ (and its marginals), and convergence of conditional covariances   and existence of conditional invariant  distribution of minimum mean-square error estimates of  $\{X_i: i=0, \ldots, n\}$ from data $\{Y_i: i=0, 1, \ldots,\}$, governed by the Kalman filter recursions,  are directly related to certain properties of the matrices $\{A, B, C, N, K_W, K_V\}$.
These are defined below.  \\

\begin{definition}(Stabilizablity and Detectability)\\
\label{def_st-de}
Let  $(A, B, C) \in  {\mathbb R}^{q \times q} \times  {\mathbb R}^{q \times k} \times {\mb R}^{p\times q}$. \\
(a) The pair of matrices $(A, B)$ is called stabilizable of there exists an matrix $K \in {\mb R}^{k\times q}$ such that the eigenvalues of $A-B K$ lie in ${\mb D}_o$ (i.e., $spec\big(A-B K\big) \subset {\mb D}_o$). \\
(b) The pair of matrices $(C, A)$ is called detectable if $(A^T, C^T)$ is stabilizable, i.e., there exists an matrix $L \in {\mb R}^{q\times p}$ such that the eigenvalues of $A-L C$ lie in ${\mb D}_o$ (i.e., $spec\big(A-L C\big) \subset {\mb D}_o$).\\
(c) The pair of matrices $(A, B)$ is called controllable if 
\begin{align}
rank \Big({\cal C}\Big)=q, \hso {\cal C} \tri \left[ \begin{array}{cccc} B & A B & \ldots & A^{q-1} B \end{array} \right] \in {\mathbb R}^{q \times q k}.
\end{align}
(d) The pair of matrices $(C, A)$ is called observable of $(A^T, C^T)$ is controllable, i.e., if
\begin{align}
rank \Big({\cal O}\Big)=q, \hso {\cal O} \tri \left[ \begin{array}{l}  C \\ C A \\ \vdots \\ C A^{q-1}\end{array} \right] \in {\mathbb R}^{pq \times q}.
\end{align}
\end{definition}

Note that  $(A, B)$  controllable pair implies $(A, B)$  stabilizable pair, and $(C, A)$  observable  pair implies $(C, A)$  detectable pair. \\
The next theorem relates covariances of time-invariant finite dimensional Gaussian systems (\ref{LSSM_K1})-(\ref{LSSM_K1}) to Lyapunov equations. It is borrowed from \cite{vanschuppen2010}. \\

\begin{theorem} \cite{vanschuppen2010}(Properties of Lyapunov equations)\\
\label{lemma_vanschuppen}
Consider the covariance function $\Sigma: \{0,1,\ldots, n\} \longmapsto  {\mathbb R}^{q\times q}$ of the process $\{X_i: i=0, \ldots, n\}$ satisfying the recursion 
\begin{align}
\Sigma_i=A \Sigma_{i-1} A^T  +B  K_{W} B^T, \hso \Sigma_0=\mbox{given}, \hso i=1, \ldots, n. \label{MP_1_new_TV}
\end{align}
Consider the discrete Lyapunov equation for $\Sigma \in {\mb R}^{q \times q}$: 
\begin{align}
\Sigma=A  \Sigma A^T +B K_{W} B^T. \label{MP_1_new}
\end{align}
The following hold.
 \begin{description}
\item[(a)] If $A$ is an exponentially stable matrix then $\lim_{i\longrightarrow \infty } \Sigma_i =\Sigma$ exists and $\Sigma$ is a solution of the equation (\ref{MP_1_new}) (irrespectively of initial condition).

\item[(b)] If $A$ is an exponentially stable matrix then (\ref{MP_1_new}) has an unique solution, which satisfies $\Sigma=\Sigma^T\succeq  0$. 

\item[(c)] Let $r \in \{1,2,\ldots,\}$, $G\in {\mathbb R}^{q \times r}$ be such that   $B K_W B ^T= G G^T$.  Assume that $\{A,G\}$ is a  stabiizable  pair and there exists a $\Sigma \in {\mathbb R}^{q \times q}$ which satisfies
\begin{align}
\Sigma=A  \Sigma A^T +B K_{W} B^T, \hso \mbox{and} \hso \Sigma=\Sigma^T \succeq 0.
\end{align}
Then $A$ is an exponentially stable matrix.

\item[(d)] Let $\Sigma \in {\mathbb R}^{q \times q}$ be a solution of (\ref{MP_1_new}). Any two of the following three statements implies the third:
\begin{enumerate}
\item $A$ is an exponentially stable matrix ($spec(A) \subset {\mathbb D}_o$);

\item $\big( A,G\big)$ is a controllable pair ($Rank({\cal C})=q$);

\item $\Sigma \succ 0$.
\end{enumerate}
\end{description}
\end{theorem}
Note that if the initial condition of (\ref{MP_1_new_TV}) is set to  $\Sigma_0= \Sigma$, where $\Sigma$ is a solution of (\ref{MP_1_new}), then $\Sigma_i=\Sigma, i=1, 2, \ldots, n$, that is, the solution of the discrete recursion (\ref{MP_1_new_TV}) is stationary.

Consider the problem of  estimating $\{X_i: i=0, \ldots,\}$ from $\{Y_i: i=0, 1, \ldots, \}$, for  the time-invariant finite dimensional Gaussian system (\ref{LSSM_K1})-(\ref{LSSM_K3}), with respect to the following criterion.
\bea
\inf_{g_i(\cdot): i=0, \ldots, n } {\bf E} \Big\{|| X_i- g_i(Y^{i-1})||_{{\mb R}^q}^2\Big\}, \hso \mbox{where $g_i(\cdot)$ is a measurable function of $y^{i-1}$}, \hso i=0, \ldots, n.
\eea
Then the optimal estimator exists, it is unique, and it is given by the conditional expectation $g_i^*(y^{i-1})={\bf E}\big\{ X_i| Y^{i-1}\big\}= \int x {\bf P}(dx|y^{i-1}), i=0, \ldots, n$. 
The conditional distribution  $\big\{ {\bf P}(dx|y^{i-1}): i=0, \ldots, n\big\}$ is finite dimensional, and it is  described by only two statistics, the conditional mean and the conditional covariance, defined by 
\bea
\widehat{X}_{i|i-1} \tri {\bf E}\Big\{ X_i| Y^{i-1}\Big\},  \hso  Q_{i|i-1} \tri {\bf E} \Big\{ \Big(X_i- \widehat{X}_{i|i-1}\Big)\Big(X_i- \widehat{X}_{i|i-1}\Big)^T\Big| Y^{i-1}\Big\}, \hso i=0, \ldots, n.
\eea
The  conditional covariance is independent of the data and it is equal to the unconditional covariance,
\bea
 Q_{i|i-1} = {\bf E} \Big\{ \Big(X_i- \widehat{X}_{i|i-1}\Big)\Big(X_i- \widehat{X}_{i|i-1}\Big)^T\Big\}\hso i=0, \ldots, n.
 \eea
Moreover, $\big\{\widehat{X}_{i|i-1}:i=0, \ldots, n\big\}$ satisfies a recursive equation known as the Kalman-filter equation, and $\big\{Q_{i|i-1}: i=0, \ldots, n\big\}$ satisfies a recursive equation, known as the filtering Riccati difference matrix  equation. \\
The properties of the Kalman-filter, such as, the convergence of the  covariance (of the error) and the existence of invariant conditional distribution are determined from the properties of Riccati difference and algebraic equations. \\
The following theorem is borrowed from \cite{vanschuppen2010}; it  summarizes properties of Riccati difference and algebraic equations. \\

\begin{theorem}\cite{vanschuppen2010}
(Properties of Riccati equations)\\ 
\label{RIC}
Assume $N K_V N^T\succ 0$.  Let $f: {\mb R}^{q\times q} \longmapsto {\mb R}^{q \times q}, G \in {\mb R}^{q \times q}$ be defined by
\begin{align}
f(Q) \tri A Q A^T + B K_W B^T -A Q C^T \big[C QC^T + N K_V N^T\big]^{-1} \big(A Q C^T\big)^T, \hst G G^T \tri B K_W B^T.
\end{align}
Let $F : {\mb R}^{q\times q} \longmapsto {\mb R}^{q \times p}, Q \longmapsto F(Q),$ and $A: {\mb R}^{p \times p} \longmapsto {\mb R}^{p \times p}, A \longmapsto A(Q)$ be defined by
\begin{align}
F(Q)\tri A Q C^T \big[C Q C^T + N K_V N^T\big]^{-1}, \hst A(Q)=A- F(Q) C
\end{align}
Define the discrete-time Riccati recursion for 
$Q: \{ 0,1, \ldots, n\} \longmapsto {\mathbb R}^{q\times q}$ by
\bea
Q_{i+1}=f(Q_i), \hso Q_0=\mbox{given}, \hso i=1, \ldots, n. 
\eea
and the algebraic Riccati equation for the matrix $Q \in {\mb R}^{q \times q}$:
\bea
Q=f(Q). \label{ARE_1}
\eea
The following hold.\\
(a) If $(C, A)$ is a detectable pair and $(A, G)$ is a stabilizable pair, then there exists a positive definite solution $Q \in {\mathbb R}^{q\times q}$ to the algebraic Riccati equation
\bea
Q=f(Q), \hso Q=Q^T\succeq 0. \label{RIC_1}
\eea
(b) If $(A, G)$ is a stabilizable pair then the algebraic Riccati equation (\ref{RIC_1}) has at most one solution.\\
(c) Under the assumptions of (a) the limit $\lim_{n \longrightarrow \infty} Q_i =Q$ exists and $Q$ is the positive definite solution of the algebraic Riccati equation (\ref{ARE_1}).\\
(d) If $(A, G)$ is a stabilizable pair and if there exists a positive definite solution $Q$ to the algebraic Riccati equation (\ref{ARE_1}), then $spec\big(A(Q)\big) \subset {\mb D}_o$. \\
(e) Consider the algebraic Riccati equation for $Q \in {\mb R}^{n \times n}$ given by (\ref{ARE_1}), 
with the conditions that $C Q C^T +N K_V N^T \succ 0$ and $spec\big(A(Q)\big) \subset {\mb D}_o$ (but without the condition that $Q=Q^T \succeq 0$). The  algebraic Riccati equation with these conditions has at most one solution $Q \in {\mb R}^{q\times q}$.\\
(f) Assume $(A, G)$ is a controllable pair and that there exists a $Q \in {\mb R}^{q \times q}$ such that $Q=f(Q)$ and $Q=Q^T \succeq 0$. Then $Q \succ 0$.
\end{theorem}


\end{appendix}

\bibliographystyle{IEEEtran}
\bibliography{Bibliography_capacity}

\begin{thebibliography}{10}
\providecommand{\url}[1]{#1}
\csname url@samestyle\endcsname
\providecommand{\newblock}{\relax}
\providecommand{\bibinfo}[2]{#2}
\providecommand{\BIBentrySTDinterwordspacing}{\spaceskip=0pt\relax}
\providecommand{\BIBentryALTinterwordstretchfactor}{4}
\providecommand{\BIBentryALTinterwordspacing}{\spaceskip=\fontdimen2\font plus
\BIBentryALTinterwordstretchfactor\fontdimen3\font minus
  \fontdimen4\font\relax}
\providecommand{\BIBforeignlanguage}[2]{{%
\expandafter\ifx\csname l@#1\endcsname\relax
\typeout{** WARNING: IEEEtran.bst: No hyphenation pattern has been}%
\typeout{** loaded for the language `#1'. Using the pattern for}%
\typeout{** the default language instead.}%
\else
\language=\csname l@#1\endcsname
\fi
#2}}
\providecommand{\BIBdecl}{\relax}
\BIBdecl

\bibitem{kourtellaris-charalambousIT2015_Part_1}
\BIBentryALTinterwordspacing
C.~Kourtellaris and C.~D. Charalambous, ``Information structures of capacity
  achieving distributions for feedback channels with memory and transmission
  cost: Stochastic optimal control \& variational equalities-part {I},''
  \emph{IEEE Transactions on Information Theory}, 2015, submitted, November
  2015. [Online]. Available: \url{http://arxiv.org/abs/1512.04514}
\BIBentrySTDinterwordspacing

\bibitem{cover-thomas2006}
T.~M. Cover and J.~A. Thomas, \emph{Elements of Information Theory},
  2nd~ed.\hskip 1em plus 0.5em minus 0.4em\relax John Wiley \& Sons, Inc.,
  Hoboken, New Jersey, 2006.

\bibitem{kumar-varayia1986}
P.~R. Kumar and P.~Varaiya, \emph{Stochastic Systems: Estimation,
  Identification, and Adaptive Control}.\hskip 1em plus 0.5em minus 0.4em\relax
  Prentice Hall, 986.

\bibitem{cover-pombra1989}
T.~Cover and S.~Pombra, ``{G}aussian feedback capacity,'' \emph{IEEE
  Transactions on Information Theory}, vol.~35, no.~1, pp. 37--43, {J}an. 1989.

\bibitem{yang-kavcic-tatikonda2007}
S.~Yang, A.~Kavcic, and S.~Tatikonda, ``On feedback capacity of
  power-constrained {G}aussian noise channels with memory,'' \emph{Information
  Theory, IEEE Transactions on}, vol.~53, no.~3, pp. 929--954, March 2007.

\bibitem{kim2010}
Y.-H. Kim, ``Feedback capacity of stationary {G}aussian channels,'' \emph{IEEE
  Transactions on Information Theory}, vol.~56, no.~1, pp. 57--85, 2010.

\bibitem{butman1969}
S.~Butman, ``A general formulation of linear feedback communications systems
  with solutions,'' \emph{IEEE Transactions on Information Theory}, 1969.

\bibitem{butman1976}
------, ``Linear feedback rate bounds for regressive channels,'' \emph{IEEE
  Transactions on Information Theory}, 1976.

\bibitem{permuter-cuff-roy-weissman2010}
H.~Permuter, P.~Cuff, B.~Van~Roy, and T.~Weissman, ``Capacity of the trapdoor
  channel with feedback,'' \emph{IEEE Transactions on Information Theory},
  2010.

\bibitem{elishco-permuter2014}
O.~Elishco and H.~Permuter, ``Capacity and coding of the ising channel with
  feedback,'' \emph{IEEE Transactions on Information Theory}, vol.~60, no.~9,
  pp. 3138--5149, {J}une 2014.

\bibitem{permuter-asnani-weissman2013}
H.~Permuter, H.~Asnani, and T.~Weissman, ``Capacity of a post channel with and
  without feedback,'' \emph{IEEE Transactions on Information Theory}, vol.~60,
  no.~10, pp. 6041--6057, Oct 2014.

\bibitem{yang-kavcic-tatikonda2005}
S.~Yang, A.~Kavcic, and S.~Tatikonda, ``Feedback capacity of finite-state
  machine channels,'' \emph{Information Theory, IEEE Transactions on}, vol.~51,
  no.~3, pp. 799--810, March 2005.

\bibitem{chen-berger2005}
J.~Chen and T.~Berger, ``The capacity of finite-state {M}arkov channels with
  feedback,'' \emph{IEEE Transactions on Information Theory}, vol.~51, no.~3,
  pp. 780--798, {M}arch 2005.

\bibitem{kourtellaris-charalambous2015}
C.~Kourtellaris and C.~Charalambous, ``Capacity of binary state symmetric
  channel with and without feedback and transmission cost,'' in \emph{IEEE
  Information Theory Workshop (ITW)}, May 2015.

\bibitem{dobrushin1959}
R.~L. Dobrushin, ``General formulation of {S}hannon's main theorem of
  information theory,'' \emph{Usp. Math. Nauk.}, vol.~14, pp. 3--104, 1959,
  translated in Am. Math. Soc. Trans., 33:323-438.

\bibitem{pinsker1964}
M.~Pinsker, \emph{Information and Information Stability of Random Variables and
  Processes}.\hskip 1em plus 0.5em minus 0.4em\relax Holden-Day Inc, San
  Francisco, 1964, translated by Amiel Feinstein.

\bibitem{gallager1968}
R.~T. Gallager, \emph{Information Theory and Reliable Communication}.\hskip 1em
  plus 0.5em minus 0.4em\relax John Wiley \& Sons, Inc., New York, 1968.

\bibitem{blahut1987}
R.~E. Blahut, \emph{{Principles and Practice of Information Theory}}, ser. in
  Electrical and Computer Engineering.\hskip 1em plus 0.5em minus 0.4em\relax
  Reading, MA: Addison-Wesley Publishing Company, 1987.

\bibitem{ihara1993}
S.~Ihara, \emph{Information theory for Continuous Systems}.\hskip 1em plus
  0.5em minus 0.4em\relax World Scientific, 1993.

\bibitem{verdu-han1994}
S.~Verd\'u and T.~S. Han, ``A general formula for channel capacity,''
  \emph{IEEE Transactions on Information Theory}, vol.~40, no.~4, pp.
  1147--1157, {J}uly 1994.

\bibitem{kramer1998}
G.~Kramer, ``Directed information for channels with feedback,'' Ph.D.
  dissertation, Swiss Federal Institute of Technology (ETH), December 1998.

\bibitem{han2003}
T.~S. Han, \emph{Information-Spectrum Methods in Information Theory},
  2nd~ed.\hskip 1em plus 0.5em minus 0.4em\relax Springer-Verlag, Berlin,
  Heidelberg, New York, 2003.

\bibitem{kramer2003}
G.~Kramer, ``Capacity results for the discrete memoryless network,'' \emph{IEEE
  Transactions on Information Theory}, vol.~49, no.~1, pp. 4--21, {J}an. 2003.

\bibitem{tatikonda2000}
S.~C. Tatikonda, ``Control over communication constraints,'' Ph.D.
  dissertation, Mass. Inst. of Tech.~(M.I.T.), Cambridge,~MA, 2000.

\bibitem{tatikonda-mitter2009}
S.~Tatikonda and S.~Mitter, ``The capacity of channels with feedback,''
  \emph{IEEE Transactions on Information Theory}, vol.~55, no.~1, pp. 323--349,
  {J}an. 2009.

\bibitem{permuter-weissman-goldsmith2009}
H.~Permuter, T.~Weissman, and A.~Goldsmith, ``Finite state channels with
  time-invariant deterministic feedback,'' \emph{IEEE Transactions on
  Information Theory}, vol.~55, no.~2, pp. 644--662, {F}eb. 2009.

\bibitem{gamal-kim2011}
E.~A. Gamal and H.~Y. Kim, \emph{Network Information Theory}.\hskip 1em plus
  0.5em minus 0.4em\relax Cambridge University Press, December 2011.

\bibitem{charalambous-stavrou2013aa}
\BIBentryALTinterwordspacing
C.~D. Charalambous and P.~A. Stavrou, ``Directed information on abstract
  spaces: Properties and variational equalities,'' \emph{submitted to IEEE
  Transactions on Information Theory}, 2013. [Online]. Available:
  \url{http://arxiv.org/abs/1302.3971}
\BIBentrySTDinterwordspacing

\bibitem{massey1990}
J.~L. Massey, ``Causality, feedback and directed information,'' in
  \emph{International Symposium on Information Theory and its Applications
  (ISITA~'90)}, {N}ov. 27-30 1990, pp. 303--305.

\bibitem{dluenberger1969}
D.~G. Luenberger, \emph{Optimization by Vector Space Methods}.\hskip 1em plus
  0.5em minus 0.4em\relax John Wiley \& Sons, Inc., New York, 1969.

\bibitem{gihman-skorohod1979}
I.~I. Gihman and A.~V. Skorohod, \emph{Controlled Stochastic Processes}.\hskip
  1em plus 0.5em minus 0.4em\relax Springer-Verlag, 1979.

\bibitem{vanschuppen2010}
J.~H. van Schuppen, \emph{Mathematical Control and System Theory of
  Discrete-Time Stochastic Systems}.\hskip 1em plus 0.5em minus 0.4em\relax
  Preprint, 2010.

\bibitem{teletar1999}
E.~Teletar, ``Capacity of multi-antenna {G}aussian channels,'' \emph{European
  Transactions on Telecommunications}, 1999.

\bibitem{kim2008}
Y.-H. Kim, ``A coding theorem for a class of stationary channels with
  feedback,'' \emph{IEEE Transactions on Information Theory}, vol.~54, no.~4,
  pp. 1488--1499, 2008.

\bibitem{caines1988}
P.~E. Caines, \emph{Linear Stochastic Systems}, ser. Wiley Series in
  Probability and Statistics.\hskip 1em plus 0.5em minus 0.4em\relax John Wiley
  \& Sons, Inc., New York, 1988.

\end{thebibliography}

\end{document}